\documentclass[twocolumn]{article}
\usepackage{imr}
\usepackage{graphicx}
\usepackage{amsfonts}
\usepackage{hyperref}
\usepackage{comment}
\usepackage{xcolor}

\usepackage{amsthm}
\theoremstyle{definition}

\newtheorem{proposition}{Proposition}

\renewcommand{\H}{\mathcal{H}}
\newcommand{\E}{\mathcal{E}}
\newcommand{\V}{\mathcal{V}}
\newcommand{\T}{\mathcal{T}}
\newcommand{\Q}{\mathcal{Q}}
\newcommand{\F}{\mathcal{F}}
\newcommand{\C}{\mathcal{C}}
\newcommand{\D}{\mathcal{D}}
\newcommand{\U}{\mathcal{U}}

\usepackage{algorithmicx}
\usepackage{algorithm}
\usepackage{algpseudocode}
\usepackage{wrapfig}
\usepackage{floatflt}
\usepackage{units}
\algdef{SE}[DOWHILE]{Do}{doWhile}{\algorithmicdo}[1]{\algorithmicwhile\ #1}%
\usepackage[normalem]{ulem} 

\begin{document}

\title{\uppercase{Local Decomposition of Hexahedral Singular Nodes into Singular Curves}}
\author{Paul Zhang$^1$ \and Judy (Hsin-Hui) Chiang$^2$ \and Xinyi (Cynthia) Fan$^3$ \and Klara Mundilova$^4$}
\date{$^1$Massachusetts Institute of Technology, Cambridge, Massachusetts, U.S.A. pzpzpzp1@mit.edu \\
$^2$Univeristy of Illinois at Urbana-Champaign, Urbana-Champaign, Illinois, U.S.A. hsinhui2@illinois.edu\\
$^3$ Davidson College, Davidson, North Carolina, U.S.A. cyfan@davidson.edu\\
$^4$Massachusetts Institute of Technology, Cambridge, Massachusetts, U.S.A. kmundil@mit.edu}

\abstract{
Hexahedral (hex) meshing is a long studied topic in geometry processing with many fascinating and challenging associated problems. Hex meshes vary in complexity from structured to unstructured depending on application or domain of interest. Fully structured meshes require that all interior mesh edges are adjacent to exactly four hexes. Edges not satisfying this criteria are considered singular and indicate an unstructured hex mesh. Singular edges join together into singular curves that either form closed cycles, end on the mesh boundary, or end at a singular node, a complex junction of more than two singular curves. While all hex meshes with singularities are unstructured, those with more complex singular nodes tend to have more distorted elements and smaller scaled Jacobian values. In this work, we study the topology of singular nodes. We show that all eight of the most common singular nodes are decomposable into just singular curves. We further show that all singular nodes, regardless of edge valence, are locally decomposable. Finally we demonstrate these decompositions on hex meshes, thereby decreasing their distortion and converting all singular nodes into singular curves. With this decomposition, the enigmatic complexity of 3D singular nodes becomes effectively 2D. 

}


\maketitle
\thispagestyle{empty}
\pagestyle{empty}

\section{Introduction}
Hexahedral meshes are commonly used to model complex geometries and to solve numerical PDEs. The results they produce with tri-linear basis functions are often superior to those produced with linear basis functions on tetrahedral meshes \cite{weingarten1994controversy,cifuentes1992performance}. They can be preferable to other types of meshes for their natural local coordinate systems and have been shown to perform better with quadratic basis functions in the context of nonlinear elasto-plastic simulation \cite{Brett1995}. Due to the persistent demand for hex meshes, a variety of methods have been developed to generate them. 

Of particular relevance to our work are frame field based methods, where a smooth boundary aligning coordinate system is computed over the domain, followed by parameterization and hex extraction \cite{ray_practical_2016,Lyon:2016:HRH,cubecover,huang_boundary_2011,solomon_boundary_2017}. Frame field based hex meshing has especially elucidated the significance of singularities within a hex mesh since frame fields inevitably contain singular structures that are reflected in the resulting mesh. The singularities of a frame field and hex mesh typically consist of a set of singular curves that join up in space at singular nodes. These nodes and curves form the singular graph of a field or mesh as illustrated in \autoref{fig:singular_def}.
Since hex meshable singularities are a subset of frame field singularities, much attention has been devoted to the restriction and correction of singularities in frame field computation \cite{li2012, Jiang:2013:AHM}. Other works have derived conditions and algorithms to compute frame fields obeying singular constraints \cite{liu_singularity-constrained_2018, corman_symmetric_2019}. 

Various works have also targeted the enumeration or simplification of hex mesh singularities. \cite{gao2017robust, gao2015hexahedral, xu2021singularity} derive algorithms to simplify the singular structures of hex meshes with collapse operations on a coarsened mesh. \cite{liu_singularity-constrained_2018} provide an enumeration algorithm for all hex mesh singular nodes, as well as an exhaustive list of the most practically relevant singular node types. While this list only contains eight singular nodes, they already form complex junctions that are challenging to parse or manipulate. 

In contrast to singular nodes, singular curves are simpler to understand, since their local structure is only a 2D singularity 
extruded into 3D. 
In the fully 2D setting, cross field and quadrilateral (quad) mesh singularities are significantly easier to visualize, and in the case of parameterized cross fields, singularities are governed completely by a few simple conditions \cite{campen2019seamless}. In quad meshes, singular vertex pairs are shown to move almost fluidly within a quad mesh \cite{peng2011connectivity}. None of these results obviously translates to hex meshes, where 3D singular nodes exist as junctions of multiple 2D singularities.


In this paper, we investigate the structure of hex mesh singular nodes. We uncover that singular nodes can be simplified by pulling their constituent elements apart
into singular curves. Our results show theoretically and empirically that singular nodes can be removed from a hex mesh thus reducing the complexity of any local neighborhood in a hex mesh. 
Our contributions are as follows:
\begin{itemize}
	\item We show by construction that all eight of the most practically relevant singular nodes are decomposable into just singular curves.
	\item We show that all singular nodes, regardless of valence, are locally decomposable.
	\item We apply our decompositions to hex meshes demonstrating that entire singular graphs can be separated into independent singular curves.
\end{itemize}

\section{Preliminaries}

Our work is motivated by the following question. What if a singular node is formed when singular curves just barely skim past each other?
If that were the case, then we could separate the curves with a sheet and increase its thickness to force the curves away from each other, thereby untangling the singular node. 
To formalize this idea, 
we begin with the following definitions.

\subsection{Singular Vertices, Curves, and Nodes}
We denote a hex mesh as $\{\V,\H\}$ where $\V$ is a list of vertices embedding the mesh and $\H$ is a list of hexes of the mesh. Let $\F$ denote the quadrilateral faces of the mesh, $\E$ denote the edges of the mesh, and $\deg(e\in\E)$ denote the number of hexes adjacent to edge $e$, i.e., its degree or valence. A \emph{singular edge} is an interior edge $e$ satisfying $\deg(e) \neq 4$. 
We will not treat cases where $\deg(e) \leq 2$ since these are not typically accepted as valid hex meshes.
We will also not consider singular boundary edges in this paper but refer the interested reader to \cite{liu_singularity-constrained_2018} for the corresponding definition. 
For our purposes, one can ignore boundary singularities or push them all to the interior by adding one layer of padding to the hex mesh boundary. 
A \emph{singular vertex} is a vertex of the mesh that is adjacent to any singular edge. A \emph{singular node} is a vertex of the mesh that is adjacent to more than two singular edges. 
A \emph{singular curve} is an alternating sequence of singular edges and singular vertices that either forms closed cycles or, ends at a singular node or boundary vertex.
Note that we have chosen to deviate from the language of \cite{liu_singularity-constrained_2018} by distinguishing singular nodes from singular vertices. Singular nodes are reserved for the junctions of multiple singular curves and will be the primary focus of this work. 
\autoref{fig:singular_def} depicts a summary of this terminology. 

\begin{figure}
    \centering
    \includegraphics[width=.49\columnwidth]{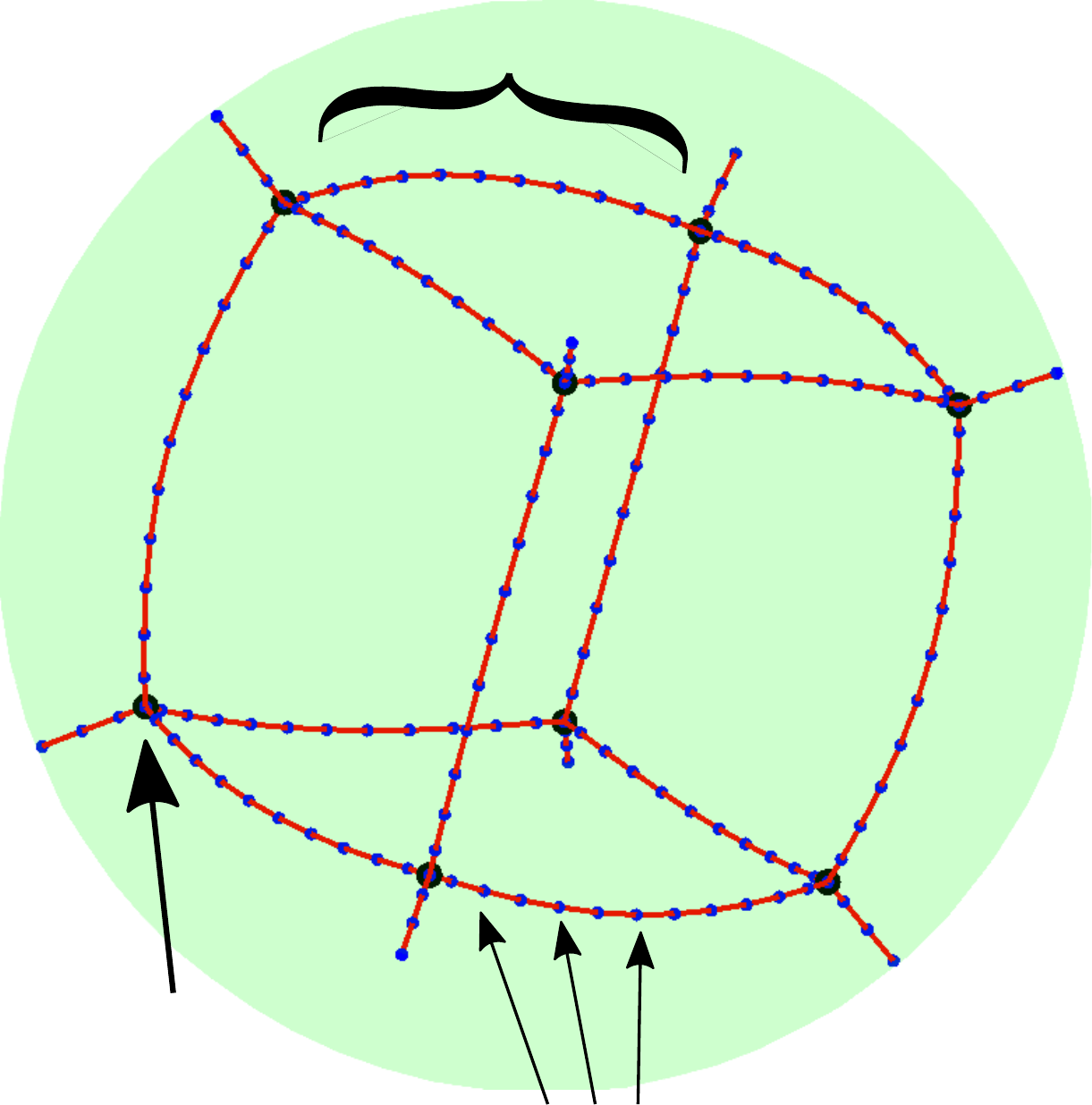}
    \includegraphics[width=.49\columnwidth]{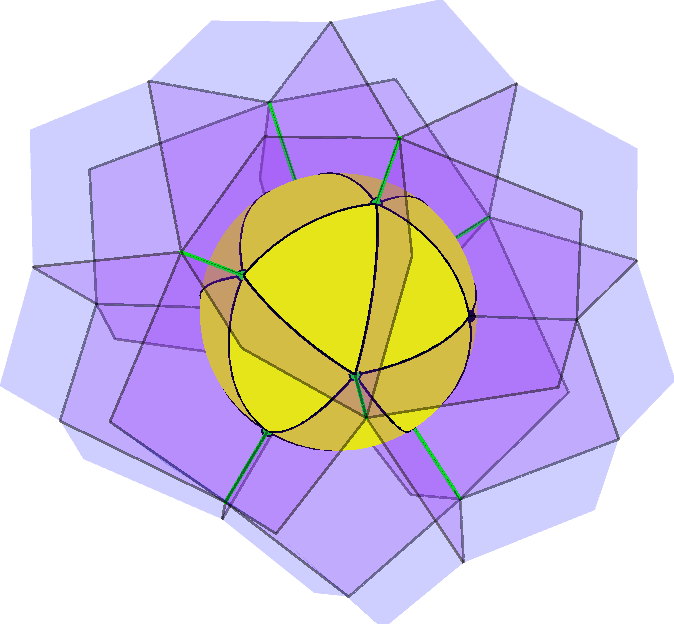}
    \put(-225,4){Singular node}
    \put(-185,105){Singular curve}
    \put(-165,-7){Singular vertices}
    \caption{(Left) The singular graph of a hex mesh of a sphere is shown. Singular edges are colored red, singular nodes are large black circles and singular vertices are small blue vertices. (Right) A close-up view of a singular node. Faces adjacent to the singular node are displayed in purple. A yellow sphere is overlayed on top of the singular node. Its intersection with the local hex mesh partitions the sphere into triangular regions.}
    \label{fig:singular_def}
\end{figure}


For a singular node $v\in\V$, we denote $\T(v)$ as the triangle mesh in bijection with that node according to \cite{cubecover}. This bijection is formed by intersecting the singular node of the hex mesh with an infinitesimally small sphere. Since the intersection of a corner of a hex with a sphere forms a triangle, the hexes adjacent to the singular node partition the sphere into triangular regions thus forming a sphere triangulation. This is depicted in \autoref{fig:singular_def}. The sphere triangulation encodes the \emph{singular node type}. If two sphere triangulations are isomorphic, then their singular nodes are of the same type. The \emph{signature} of a singular node is a list of numbers indicating how many adjacent singular edges are of each degree. Since singular edges have degree 3 or higher, the signature of a node starts with the number of edges with degree 3. For singular nodes whose adjacent singular edges have only valence 3, 4, or 5, the signature also uniquely encodes the singular node type \cite{liu_singularity-constrained_2018}. For this reason we will frequently identify singular nodes by their signature e.g. (4,0,0) is the signature identifying the singular node type generated by subdividing a single tetrahedron into four hexes as illustrated in \autoref{fig:333555}. 

\subsection{Sheet Inflation}
Let a \emph{sheet} $\Q\subset\F$ be a manifold quad mesh 
whose boundary is a subset of the boundary of $\H$.
A \emph{sheet inflation} based on sheet $\Q$ is a mesh modifying operation by which each $q\in\Q$ is thickened from a quad face to a hex cell \cite{ledoux2010topological}. 
A sheet inflation operation is depicted in \autoref{fig:sheet_insert}.

\begin{figure}
    \centering
    \includegraphics[width=.45\columnwidth]{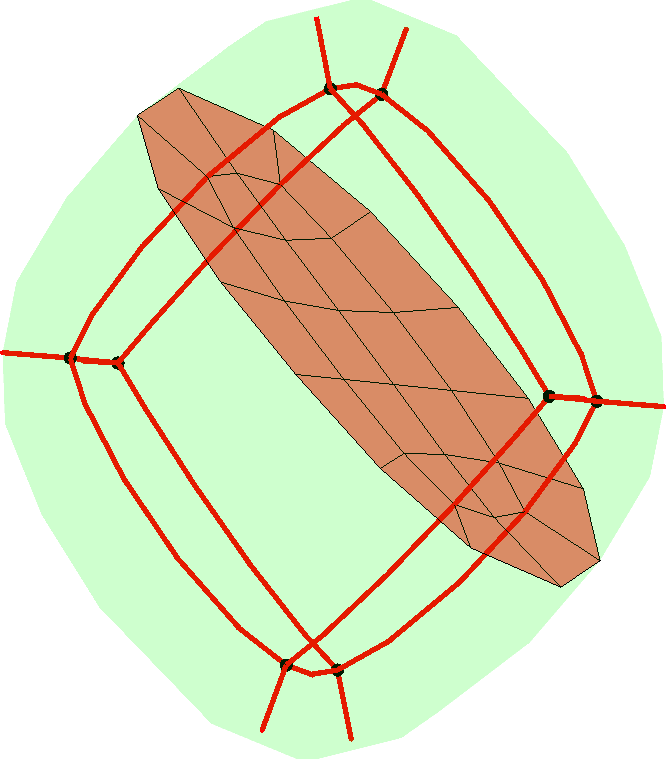}
    \includegraphics[width=.45\columnwidth]{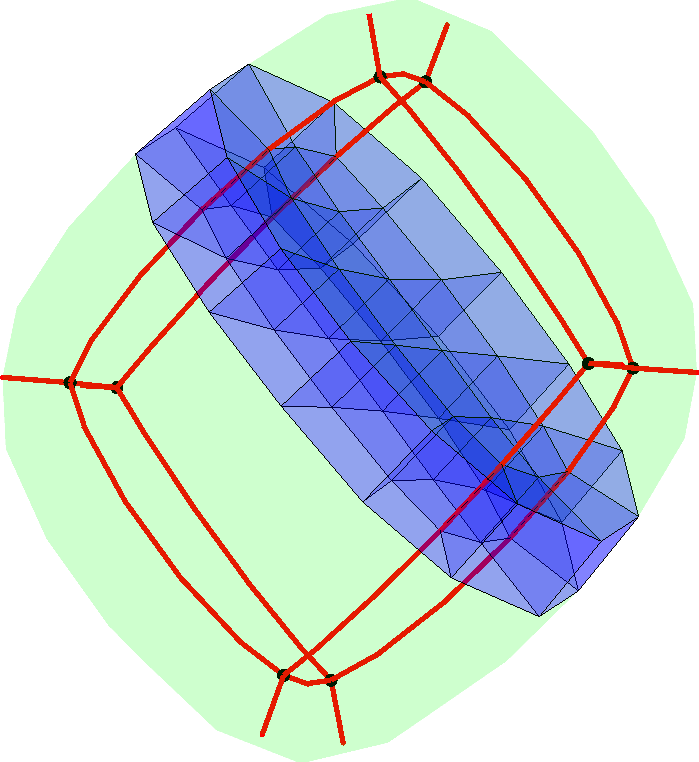}
    \caption{(Left) Red quads indicate a sheet inside of a hex mesh of an ellipsoid. Red curves depict its singular graph. The sheet is manifold with boundary on the boundary of the hex mesh. (Right) Blue quads indicate all faces of the inflated sheet.}
    \label{fig:sheet_insert}
\end{figure}

In the case that the inflated sheet passes through a singular node, the singular type of that node can change. 
One can interpret this as the sheet passing an infinitesimally small gap between singular curves and forcing them apart, or as cutting a singular node into separate pieces. 
As singular nodes are often more easily visualized as sphere triangulations we describe here how sheet inflation through a singular node corresponds to a 
\emph{splitting} of the sphere triangulation. This splitting is illustrated by \autoref{fig:sheet_insert_tri}.

Given a singular node $v$, and a sheet $\Q\subset\F$ that passes through $v$, the faces of $\Q$ map to edges of $\T(v)$. These edges trace out a cycle in the graph of $\T(v)$ effectively partitioning the sphere into two disk triangulations $\D_1$ and $\D_2$. When the sheet is thickened into a layer of hexes, these two disk triangulations are cut apart. Then the disks are patched into sphere triangulations again by adding one new vertex to each disk and attaching triangles from the boundary of each disk to their respective new vertex.
This splitting operation is illustrated in \autoref{fig:sheet_insert_tri}. The end result is two sphere triangulations one built from $\D_1$ and one built from $\D_2$. 

\begin{figure}
    \centering
    \includegraphics[width=.48\columnwidth]{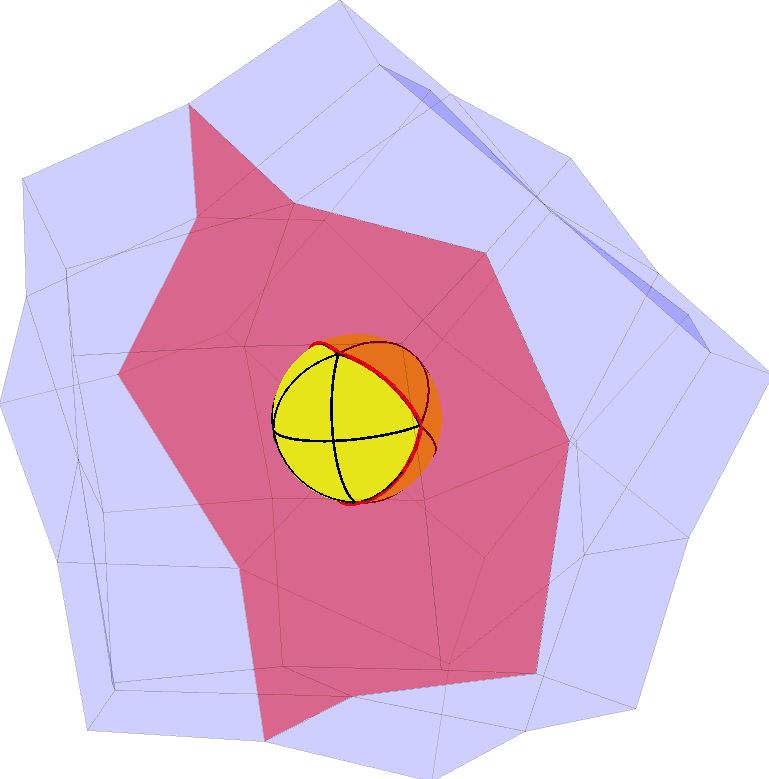}
    \includegraphics[width=.48\columnwidth]{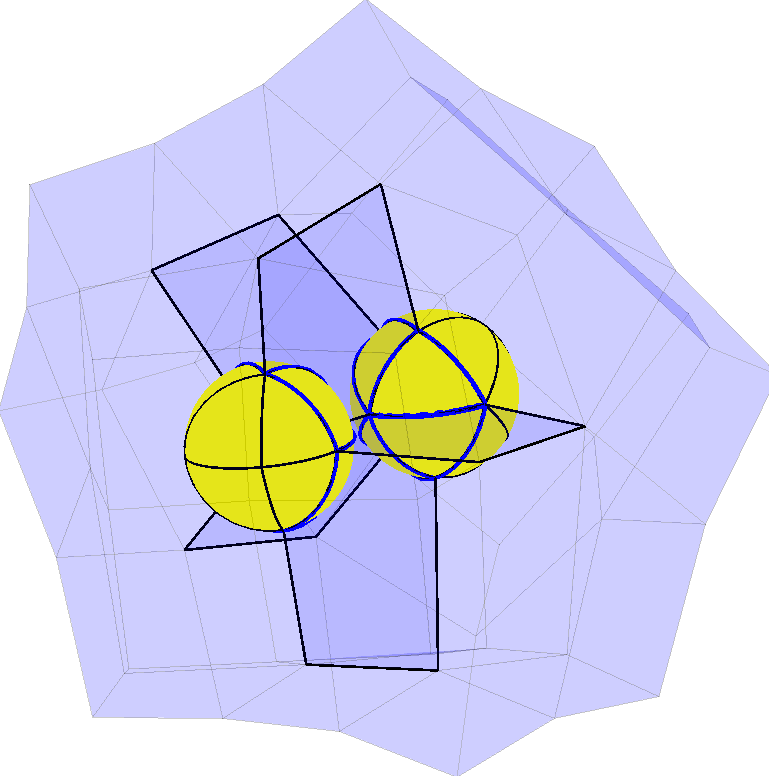}
    \caption{(Left) The yellow sphere triangulation indicates the structure of a singular node. Red quads indicate a sheet intersecting the singular node. The sheet intersects the sphere triangulation on a cycle of red curves that divide the sphere triangulation into two disks. (Right) Two singular nodes are visualized from inflation of the red sheet on the left. Sphere triangulations of the resulting two nodes are shown. Blue quad faces indicate newly created faces from the inflation. Blue edges indicate newly created edges in each sphere triangulation. }
    \label{fig:sheet_insert_tri}
\end{figure}

\section{Singular Node Decomposition (Valence 3,4,5)}
We are now equipped with the mechanism by which all practically relevant singular nodes can be decomposed into singular curves. These nodes are enumerated in \cite[Figure 6]{liu_singularity-constrained_2018} and will also be shown at the beginning of each respective decomposition. Various singular decompositions will be depicted throughout the remainder of this paper. Valence 3 singular curves will be red, valence 5 singular curves will be green and higher valence singular curves will be blue. We omit drawing interior regular edges to minimize clutter. Since these operations can be challenging to understand from static images, we also include supplemental videos for many of the decompositions.

In the first column of \autoref{fig:333555} we start with the (4,0,0) singular node, which consists of four valence 3 singular curves joined at a junction. With a single sheet inflation, this singular node is revealed to actually be two valence 3 singular curves that pass each other orthogonally. A similar theme follows for (2,2,2) in the second column which is revealed to be  a valence 3 and a valence 5 singular curve passing each other orthogonally. Finally in the third column, the (0,4,4) node decomposes into two valence 5 singular curves passing each other orthogonally. 


\begin{figure}
    \centering
    \includegraphics[width=.32\columnwidth]{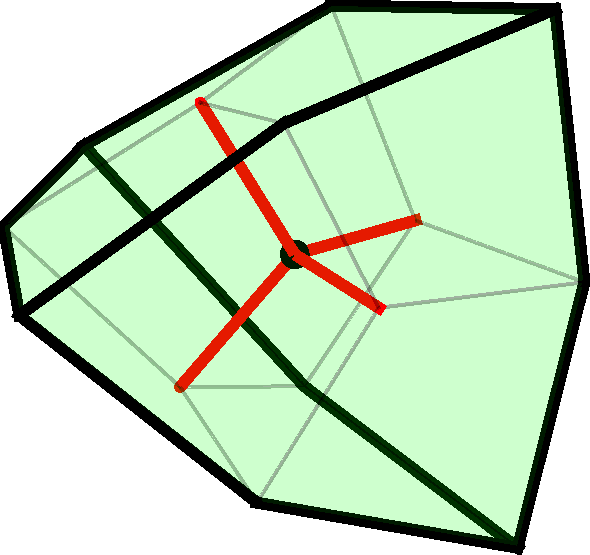}
    \hfill\includegraphics[width=.25\columnwidth]{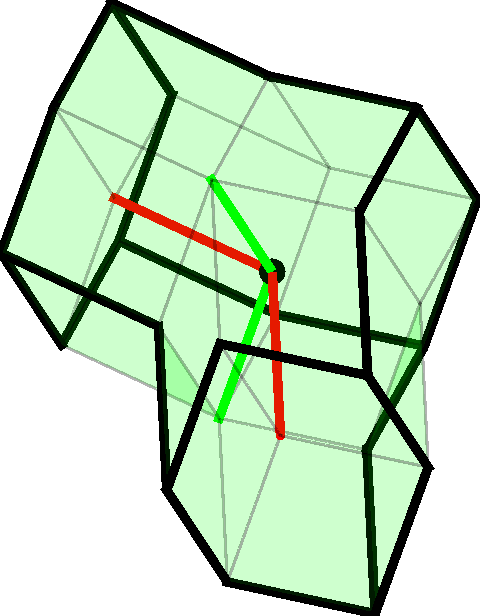}\hfill
\includegraphics[width=.32\columnwidth]{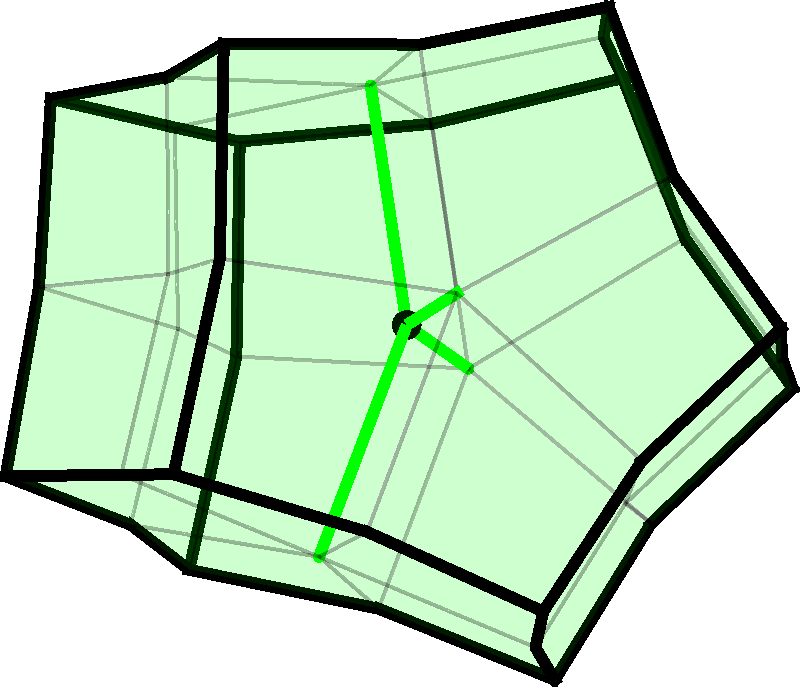}
\\
\includegraphics[width=.32\columnwidth]{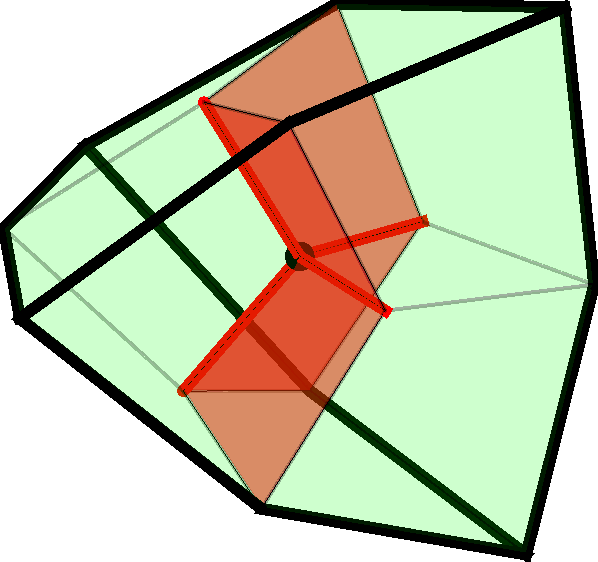}
\hfill\includegraphics[width=.25\columnwidth]{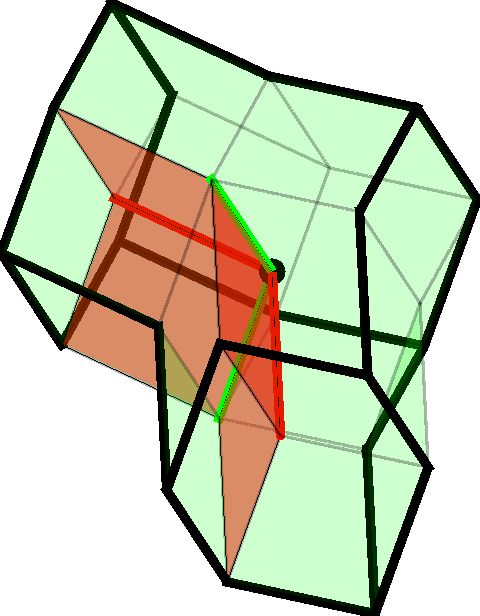}\hfill
\includegraphics[width=.32\columnwidth]{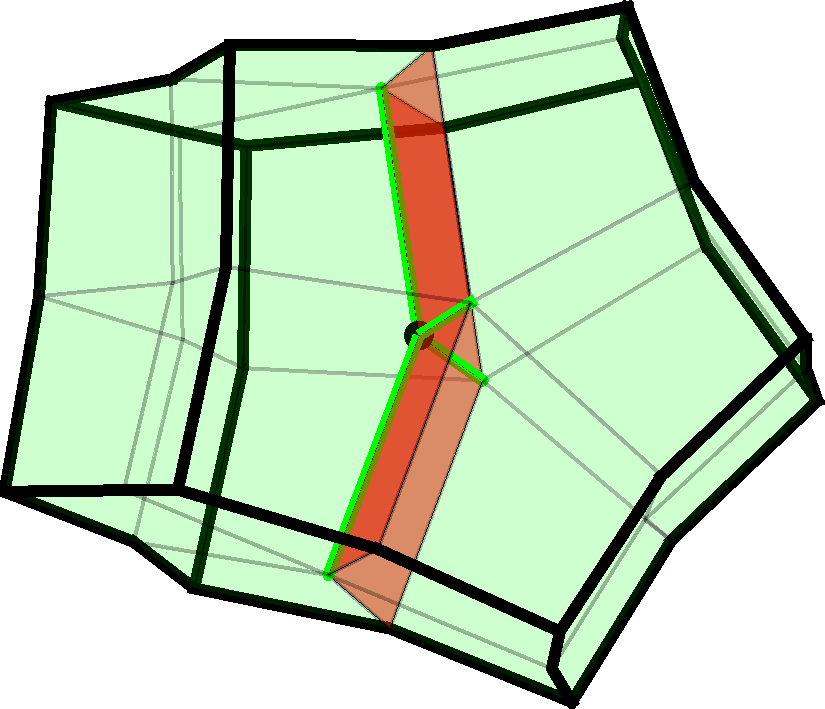}
\\
\includegraphics[width=.32\columnwidth]{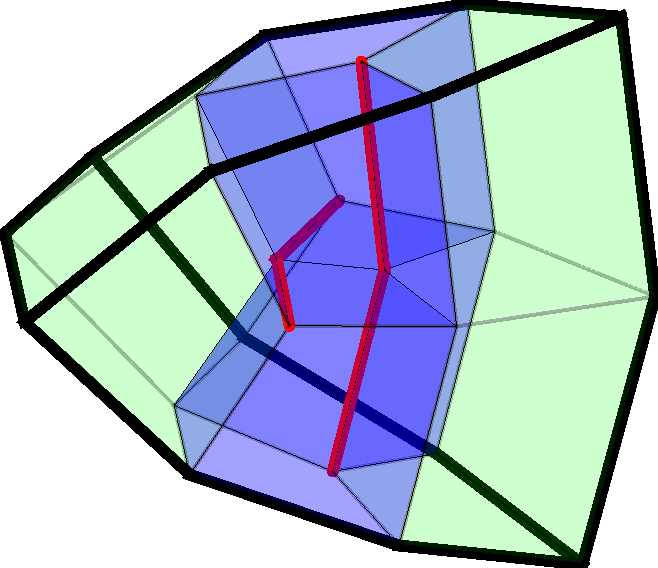}
\includegraphics[width=.32\columnwidth]{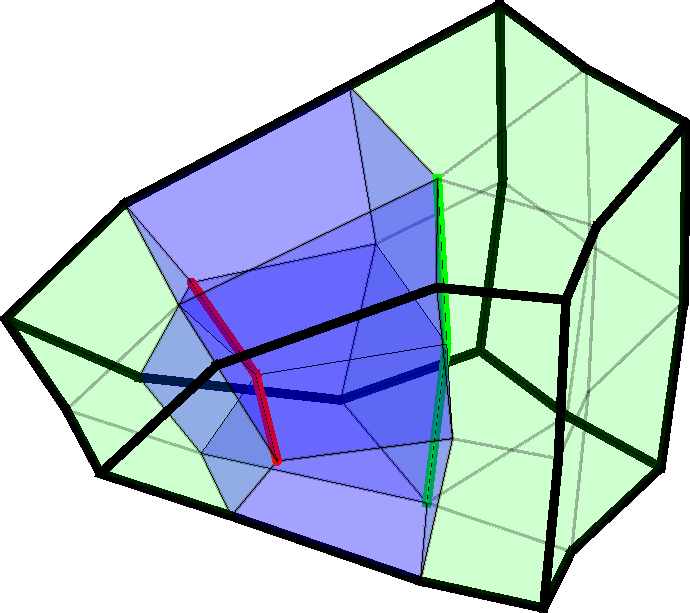}
\includegraphics[width=.32\columnwidth]{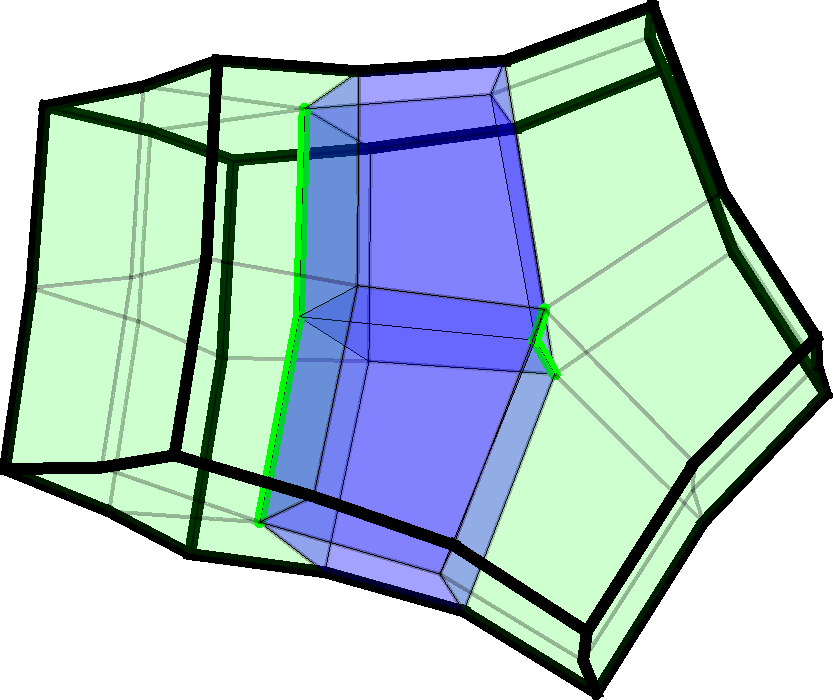}
\\
\includegraphics[width=.32\columnwidth]{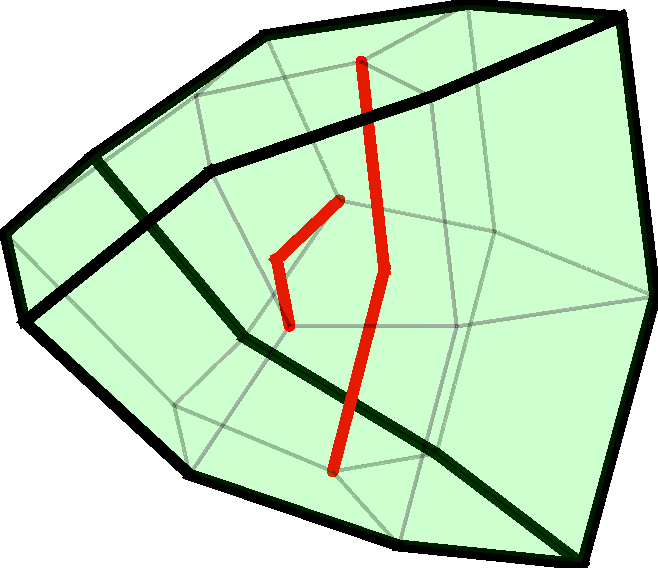}
\includegraphics[width=.32\columnwidth]{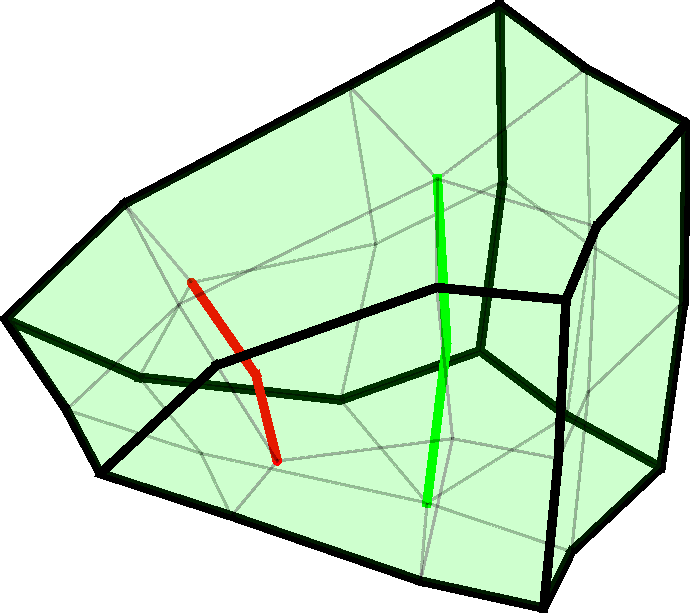}
\includegraphics[width=.32\columnwidth]{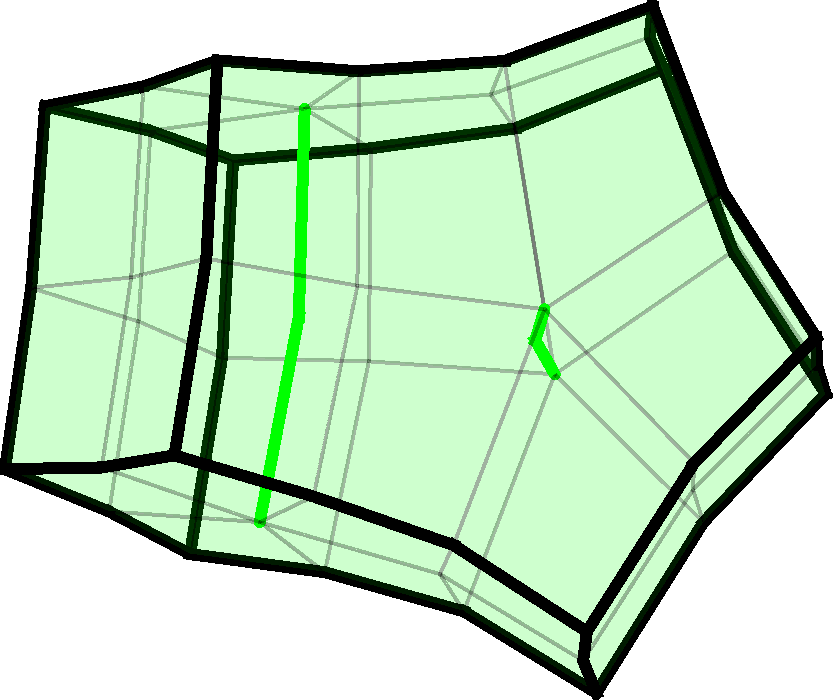}
    \caption{From left to right the (4,0,0), (2,2,2), and (0,4,4) singular nodes are depicted. Top to bottom indicates steps to decompose each singular node. Red quads indicate the sheet to be inflated. Blue quads indicate faces of the newly inflated hexes. Red(Green) edges are valence 3(5) singularities. One sheet inflation is sufficient to decompose each of these nodes. }
    \label{fig:333555}
\end{figure}

From these three examples, it is tempting to think all singular nodes may consist of singular curves glued together orthogonally. One might conclude as well that the number of singular curves of a particular valence meeting at a singular node from this construction must be even. The (1,3,3) singular node, shown in \autoref{fig:133}, presents a curious counterexample. Since it consists of one valence 3 singular curve and three valence 5 singular curves, it is impossible to decompose these into two singular curves passing each other orthogonally in a valid hex mesh. 

This conundrum is resolved by realizing that one of the regular edges adjacent to this singular node is actually a pair of valence 3 and valence 5 singular curves, glued together in parallel. Another way to understand this node is that one valence 5 singular curve has split an otherwise parallel pair of valence 3 and 5 curves. The decomposition of this node into one valence 3 and two valence 5 singular curves is illustrated in \autoref{fig:133}.

\begin{figure*}
    \centering
    \includegraphics[width=.16\textwidth]{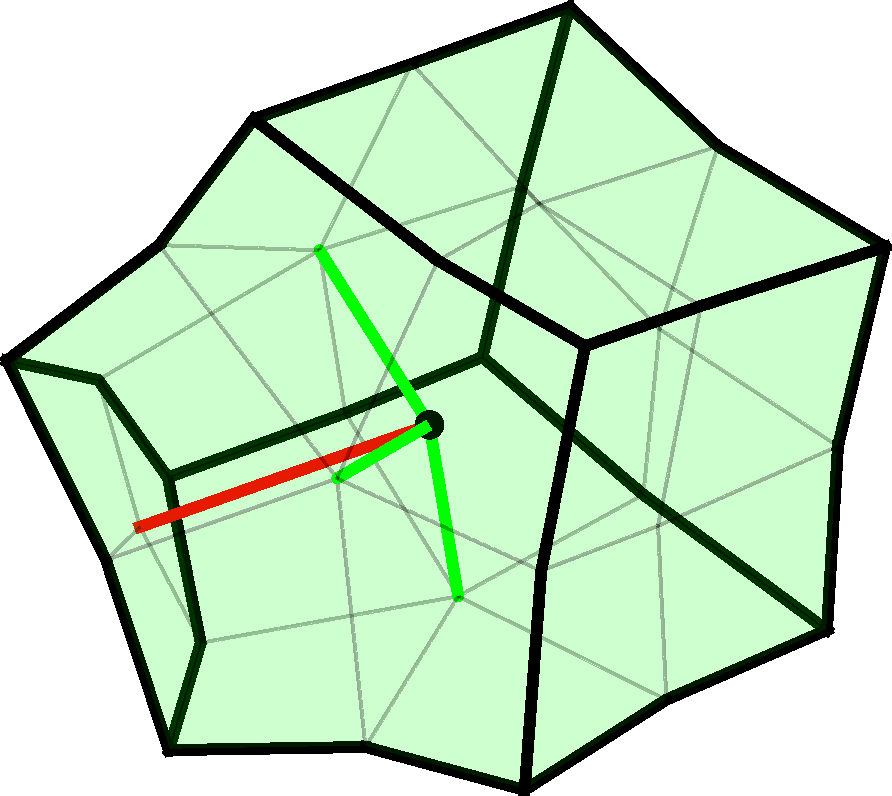}
    \includegraphics[width=.16\textwidth]{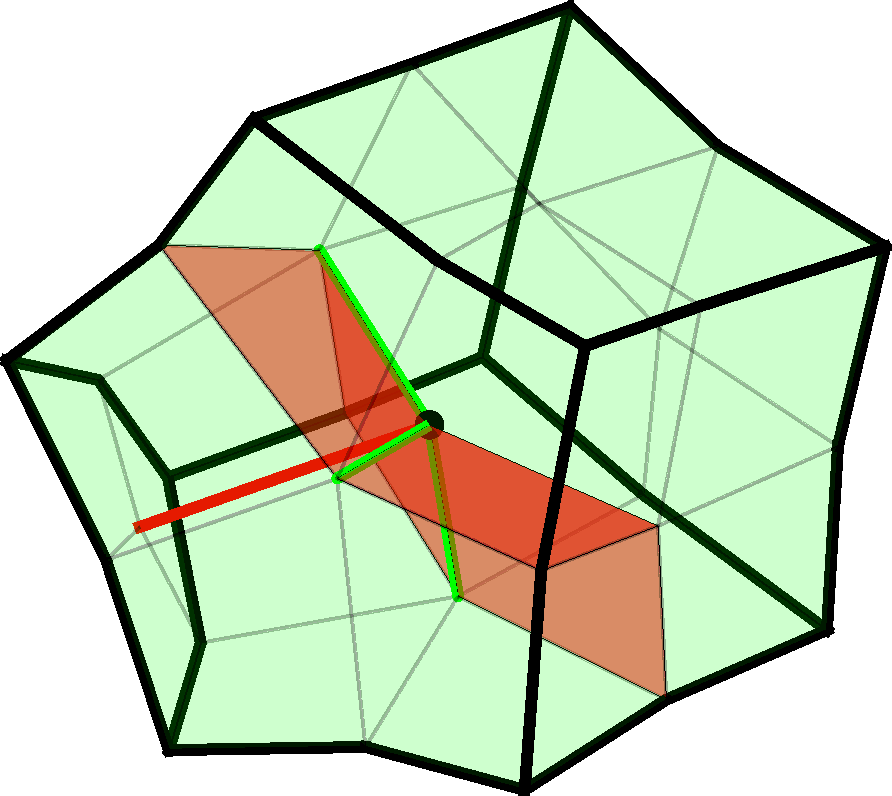}
    \includegraphics[width=.16\textwidth]{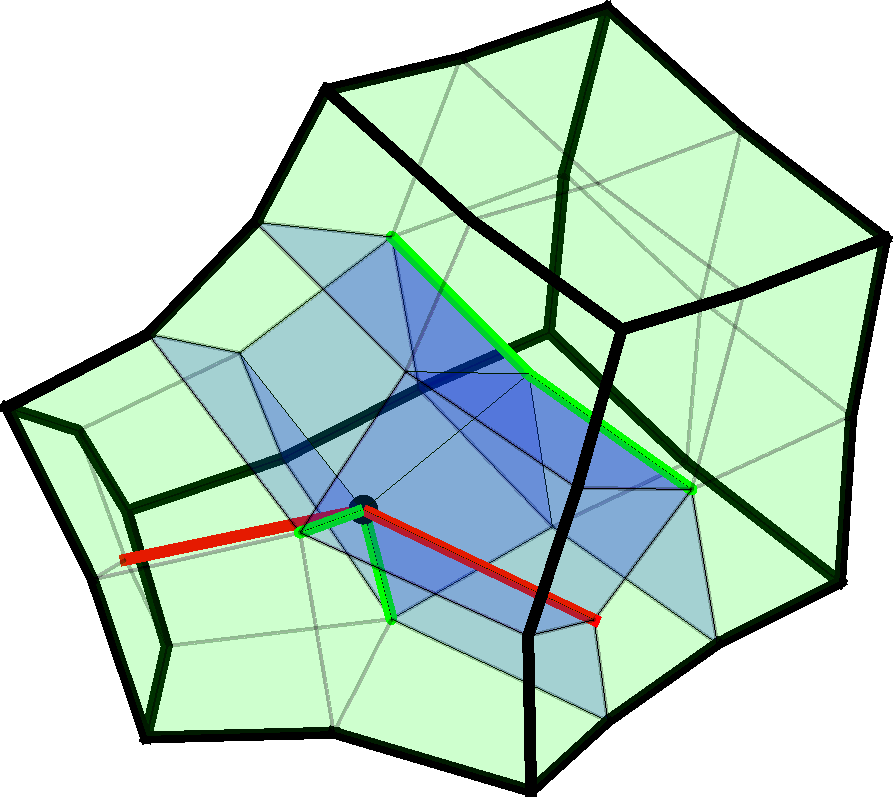}
    \includegraphics[width=.16\textwidth]{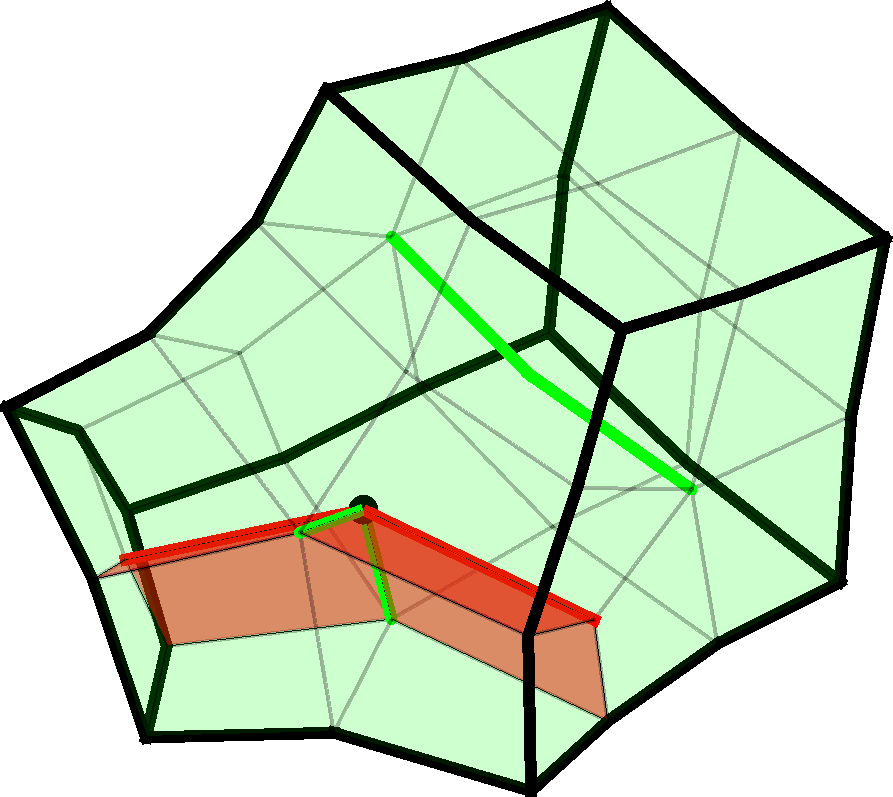}
    \includegraphics[width=.16\textwidth]{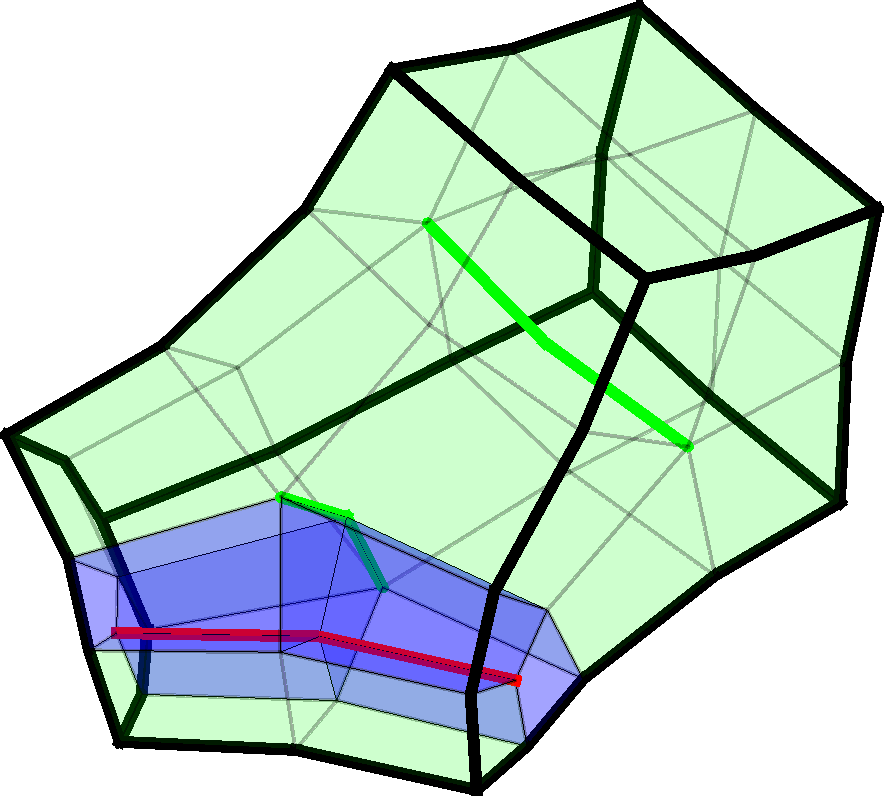}
    \includegraphics[width=.16\textwidth]{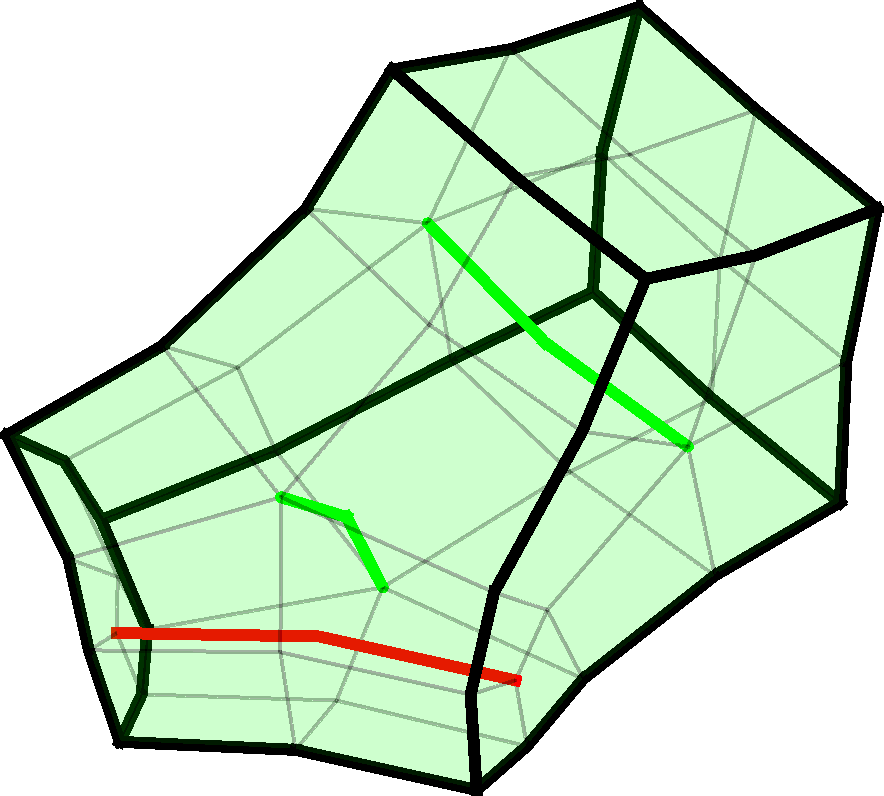}
    \caption{The (1,3,3) singular node is decomposed into two valence 5 and one valence 3 curve via two sheet inflations.}
    \label{fig:133}
\end{figure*}

In \autoref{fig:036} we decompose the (0,3,6) singular node into three valence 5 singular curves. It is also valuable to think of the (0,3,6) node as a combination of two (0,4,4) singular nodes. In this way, one only needs to decompose a singular node into constituent nodes that are already known to be decomposable. The fourth image of \autoref{fig:036} shows exactly this decomposition which we will notate as:
$$(0,3,6)=(0,4,4) +_5 (0,4,4)$$
The subscript 5 on the plus symbol denotes that two singular nodes are joined along a valence 5 edge, followed by an inverse sheet inflation (sheet collapse). This notation serves only as a shorthand and does not uniquely encode how to glue two singular nodes together. It serves more as a recipe than an equation with any algebraic properties. For completion, we show the rest of the decomposition of the constituent (0,4,4) nodes.

\begin{figure*}
    \centering
    \includegraphics[width=.115\textwidth]{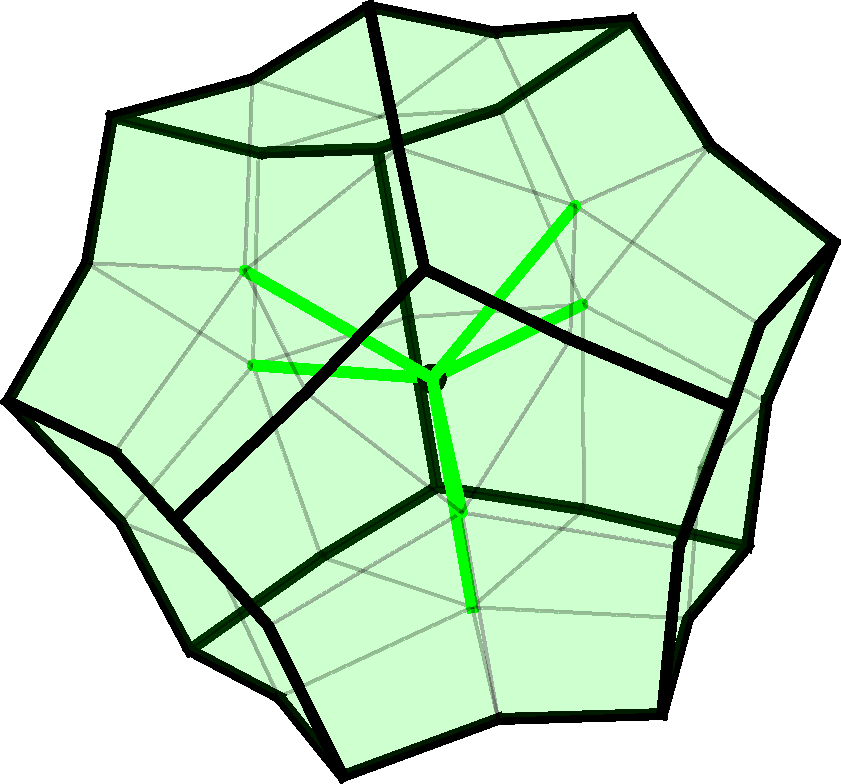}
    \includegraphics[width=.115\textwidth]{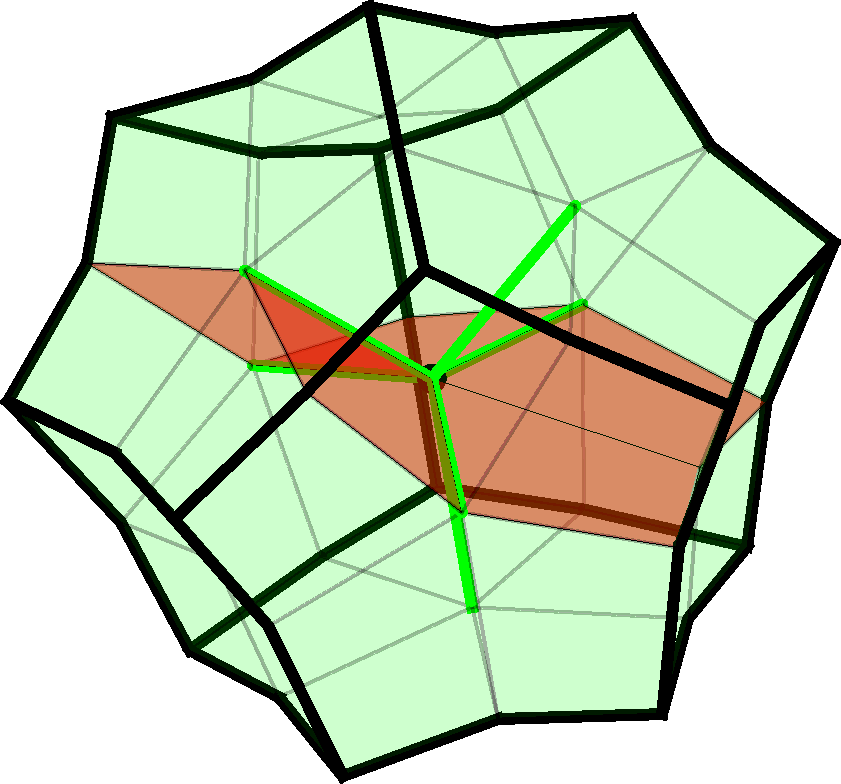}
    \includegraphics[width=.115\textwidth]{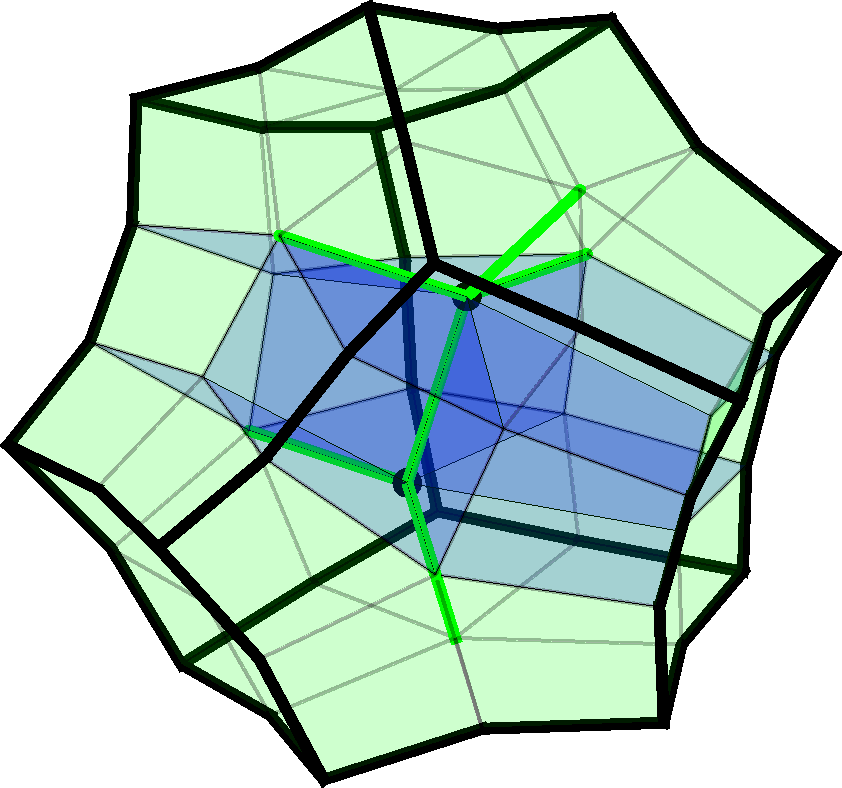}
    \includegraphics[width=.115\textwidth]{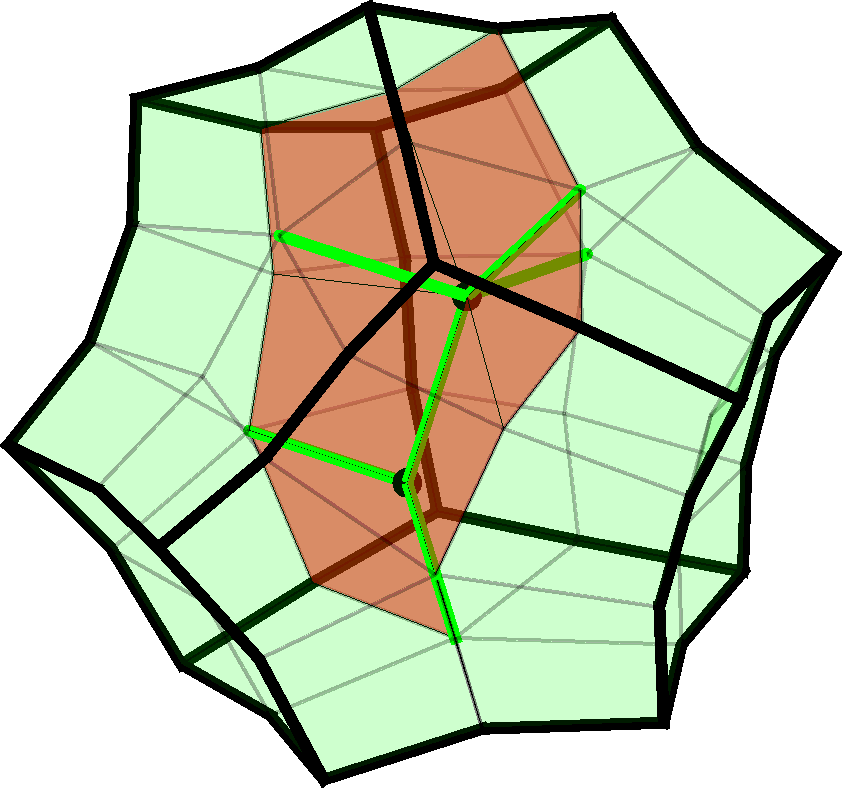}
    \includegraphics[width=.115\textwidth]{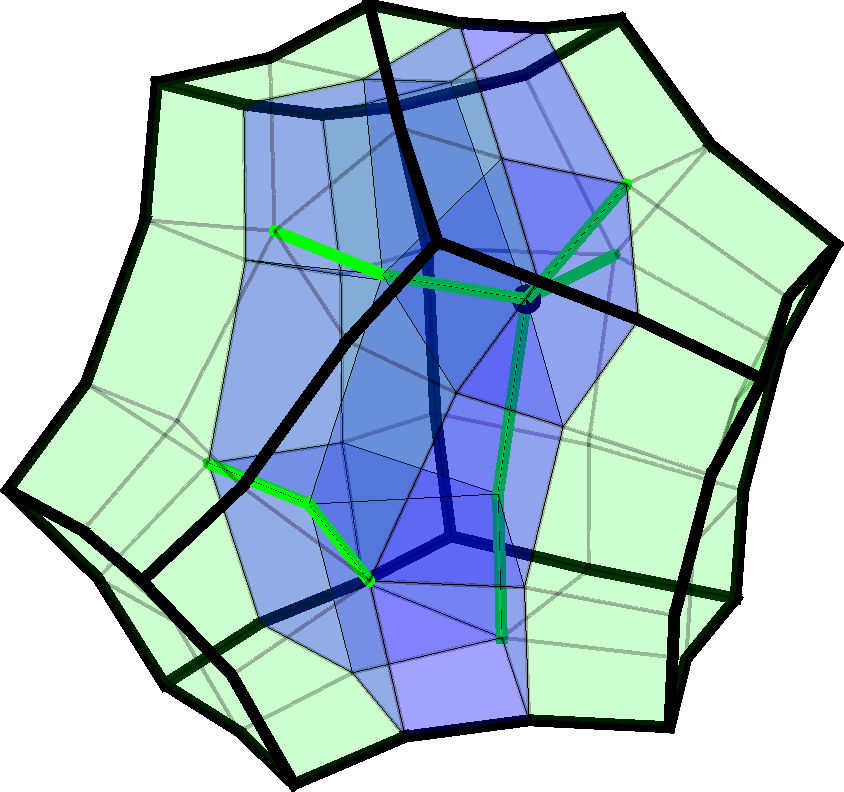}
    \includegraphics[width=.115\textwidth]{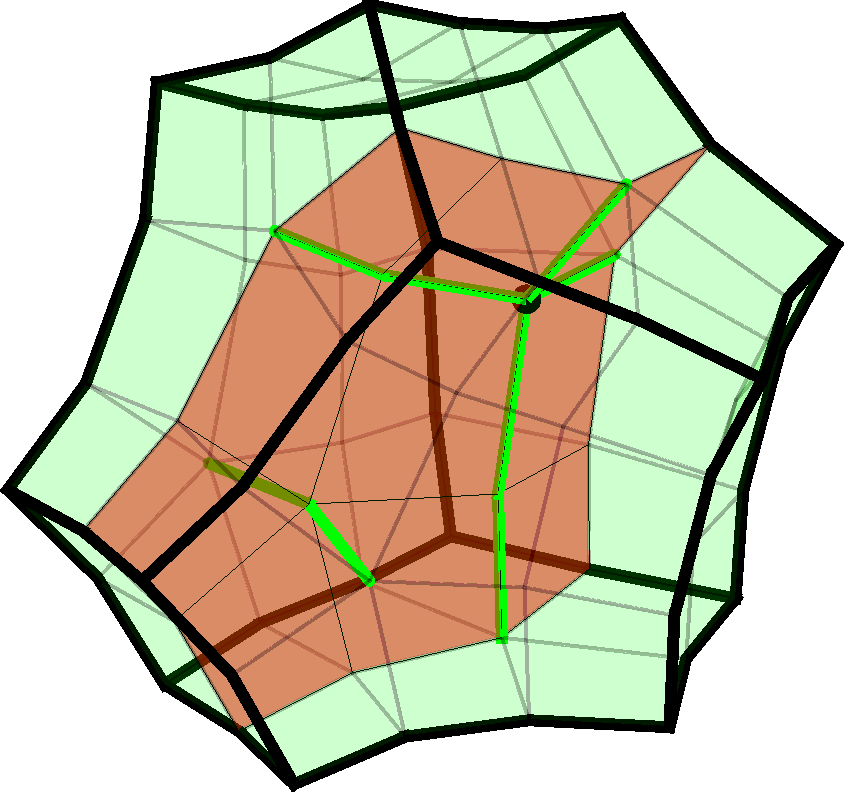}
    \includegraphics[width=.115\textwidth]{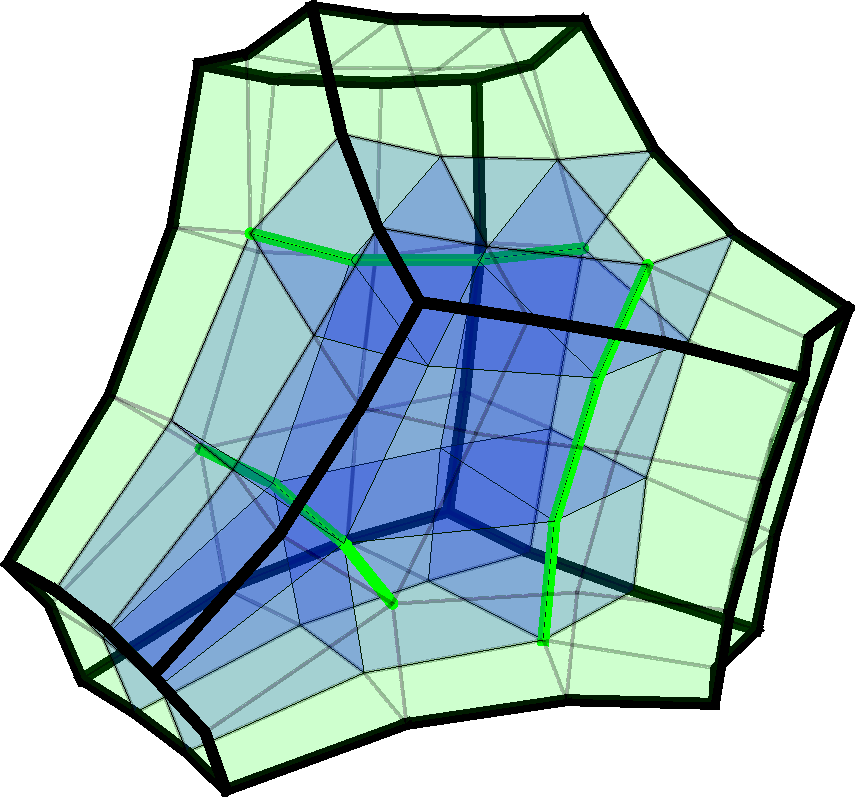}
    \includegraphics[width=.115\textwidth]{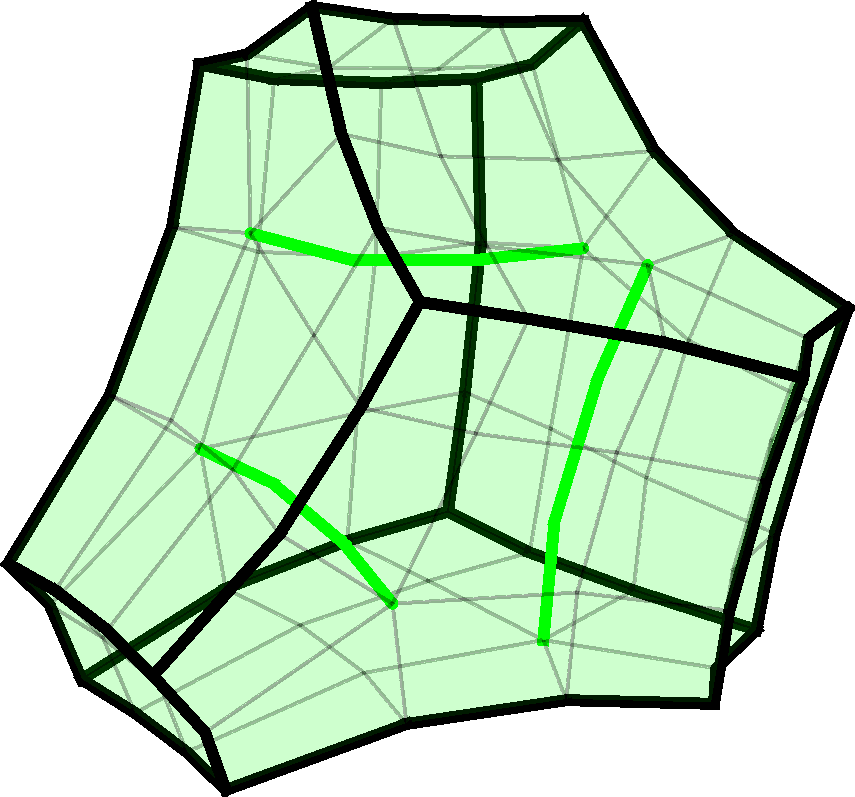}
    \caption{The (0,3,6) singular node is decomposed into three valence 5 curves via three sheet inflations.}
    \label{fig:036}
\end{figure*}

The (0,2,8) singular node is decomposed in \autoref{fig:028} into four valence 5 curves. In the fourth image, we see that 
$$(0,2,8)=(0,3,6) +_5 (0,4,4)$$
and both constituent singular nodes have already been shown to be decomposable. For completion, the rest of the decomposition steps are also shown.

\begin{figure*}
    \centering
    \includegraphics[width=.15\textwidth]{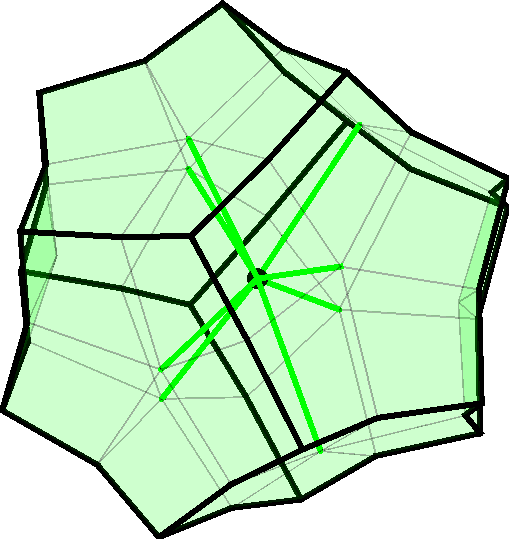}
    \includegraphics[width=.15\textwidth]{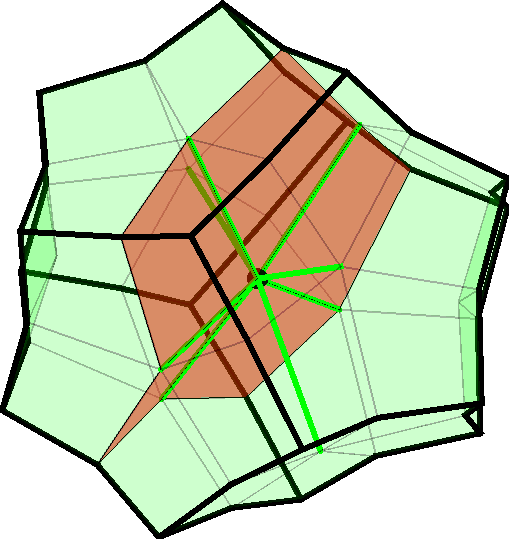}
    \includegraphics[width=.15\textwidth]{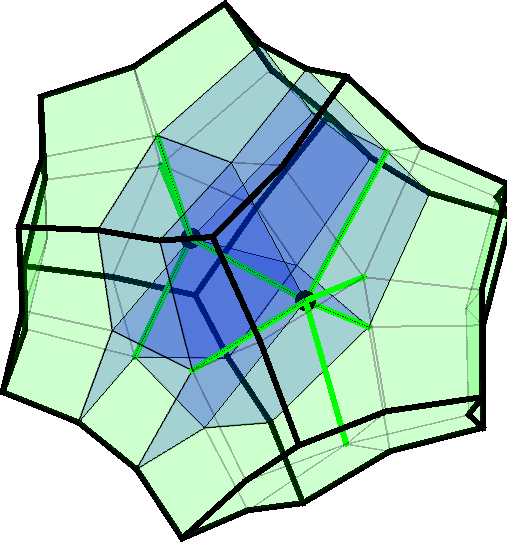}
    \includegraphics[width=.15\textwidth]{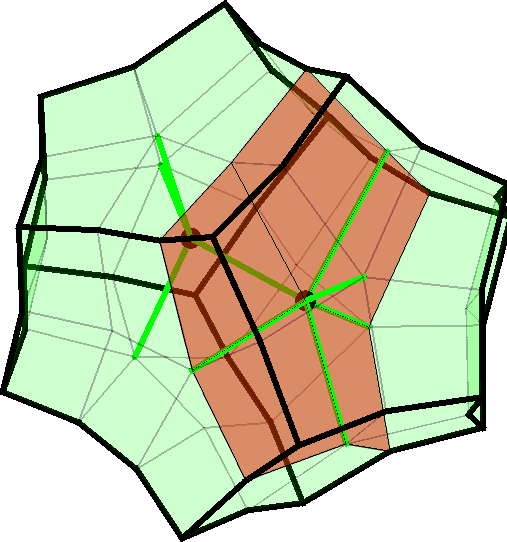}
    \includegraphics[width=.15\textwidth]{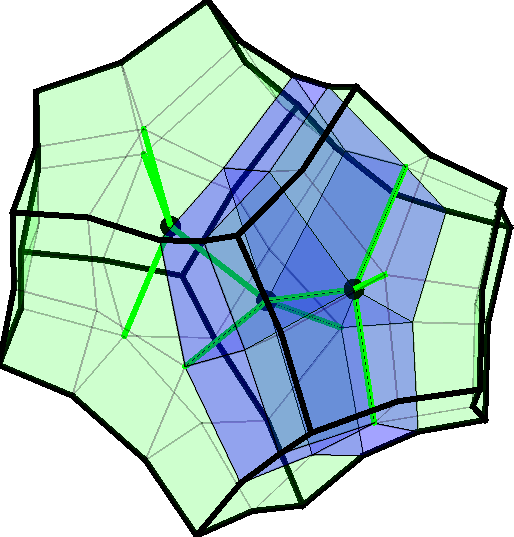}
    \includegraphics[width=.15\textwidth]{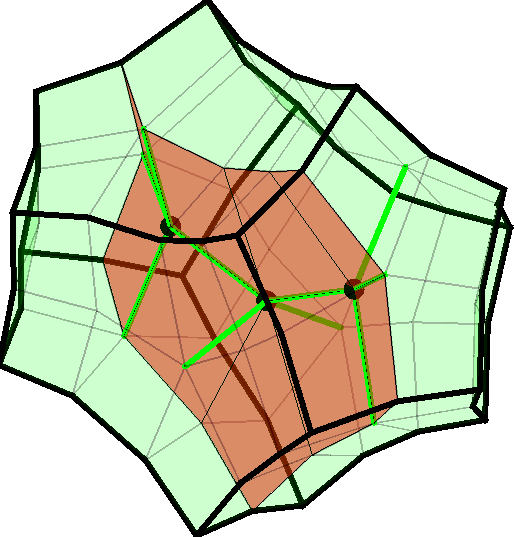}
    \\
    \includegraphics[width=.15\textwidth]{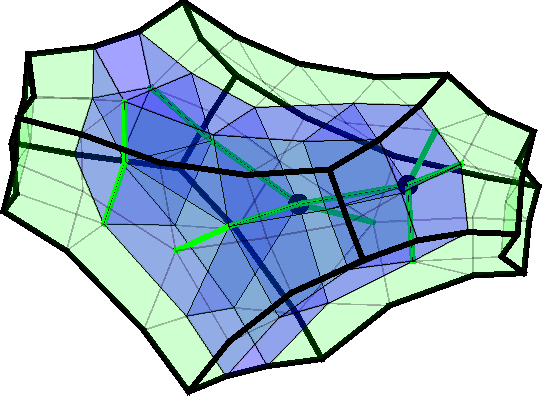}
    \includegraphics[width=.15\textwidth]{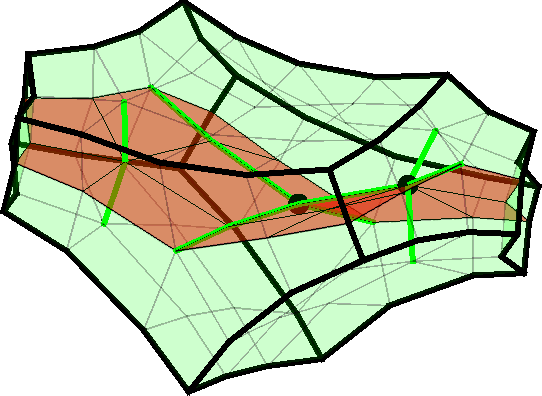}
    \includegraphics[width=.15\textwidth]{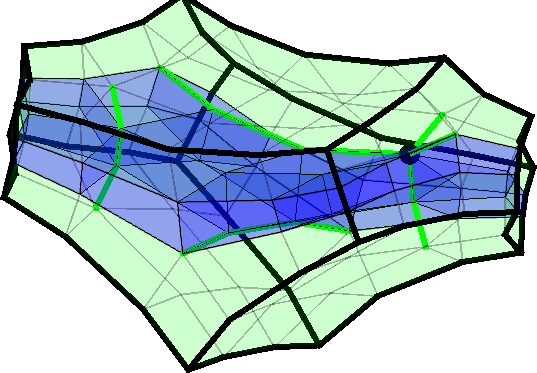}
    \includegraphics[width=.15\textwidth]{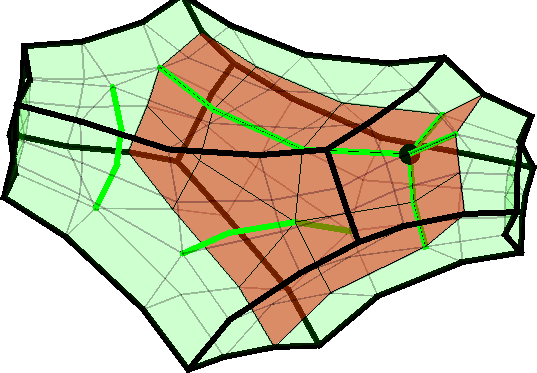}
    \includegraphics[width=.15\textwidth]{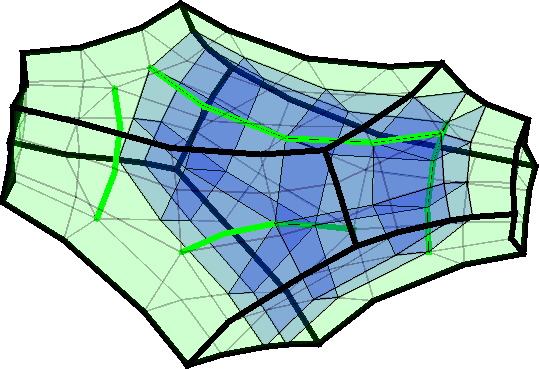}
    \includegraphics[width=.15\textwidth]{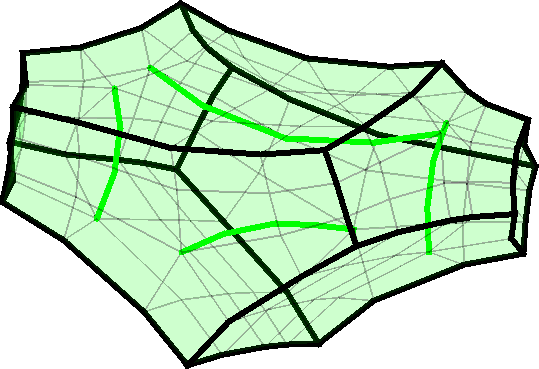}
    \caption{The (0,2,8) singular node is decomposed into four valence 5 curves via five sheet inflations.}
    \label{fig:028}
\end{figure*}

The (2,0,6) singular node is decomposed in \autoref{fig:206} into four valence 5 and two valence 3 curves. By the fourth image we see
$$(2,0,6)=(1,3,3) +_4 (1,3,3).$$
The rest of this decomposition is still interesting however as it introduces a valence 6 singularity in the fifth image and a singular node with signature (1,3,3,1). This valence 6 singular curve is removed in image 9 by decomposing it into two valence 5 curves. 

\begin{figure*}
    \centering
    \includegraphics[width=.115\textwidth]{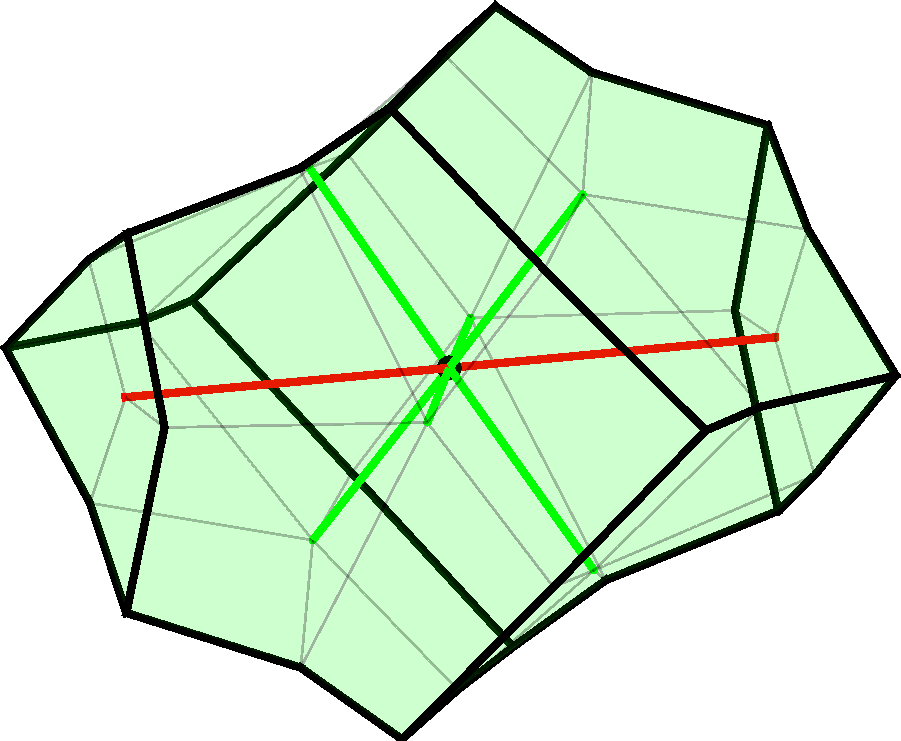}
    \includegraphics[width=.115\textwidth]{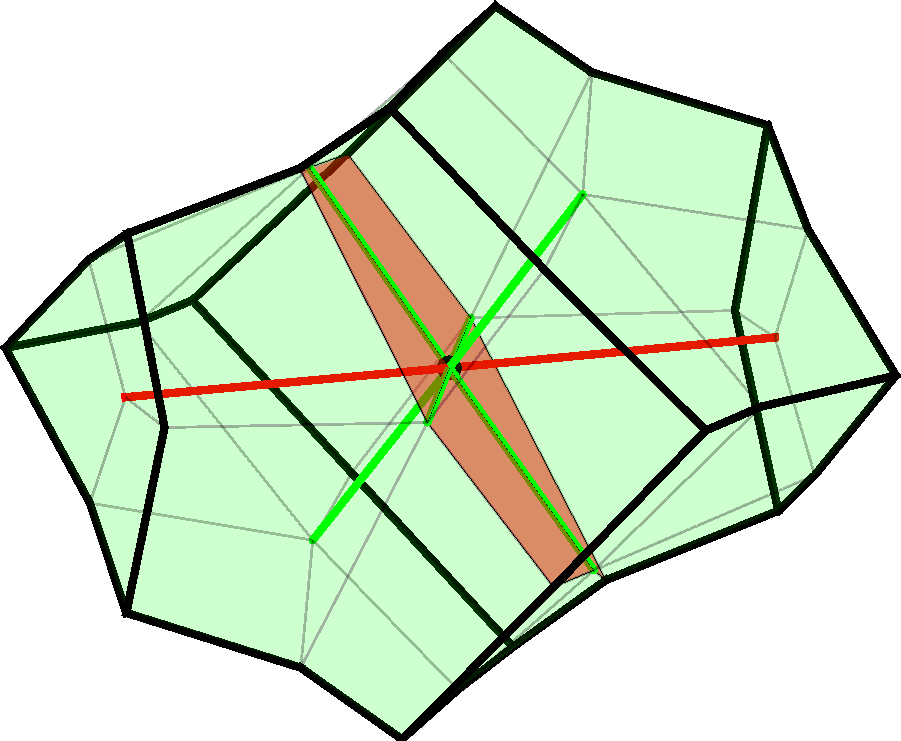}
    \includegraphics[width=.115\textwidth]{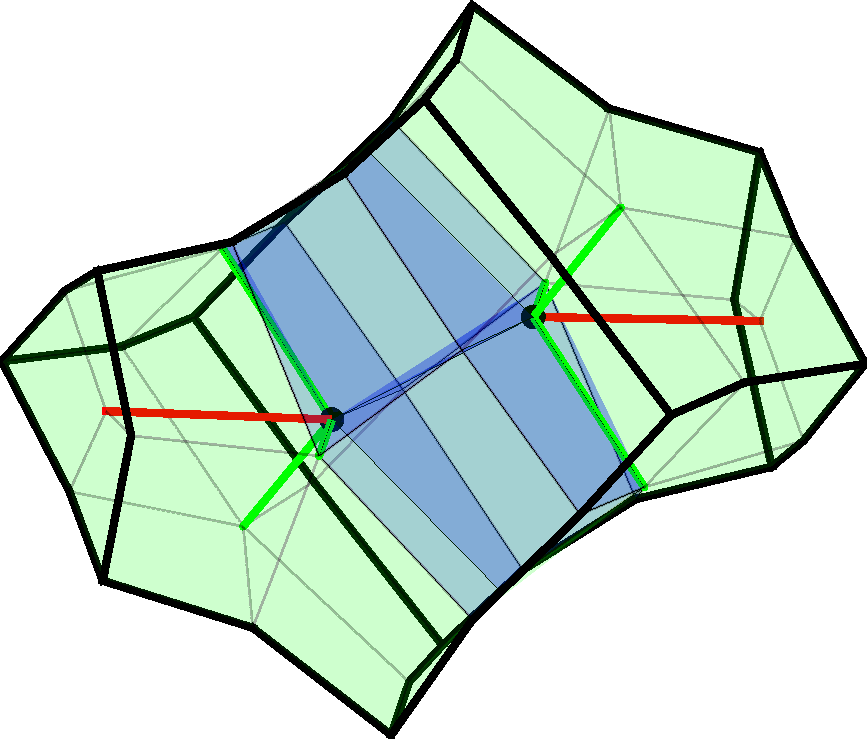}
    \includegraphics[width=.115\textwidth]{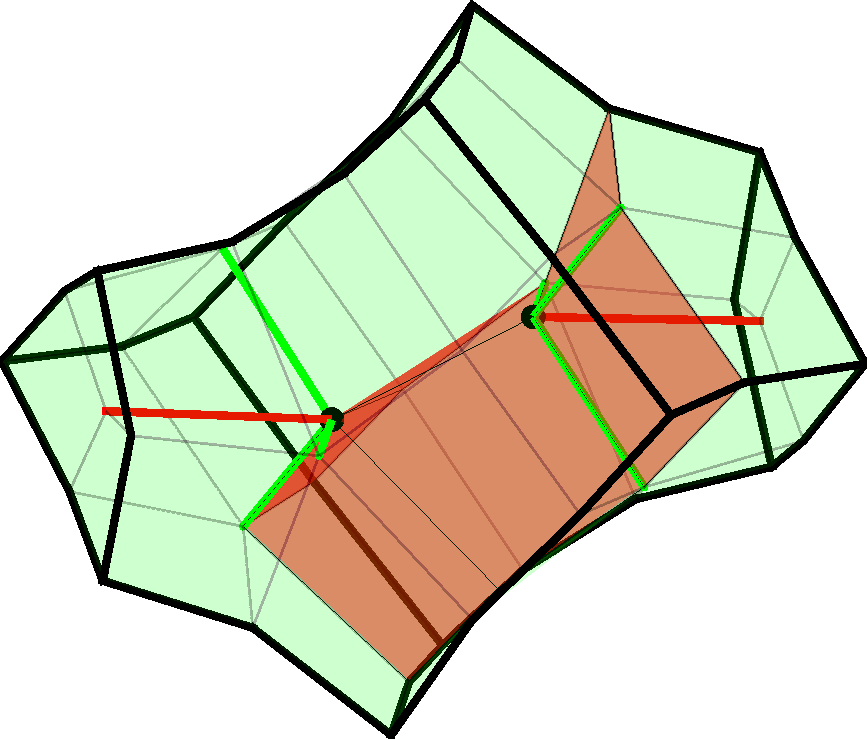}
    \includegraphics[width=.115\textwidth]{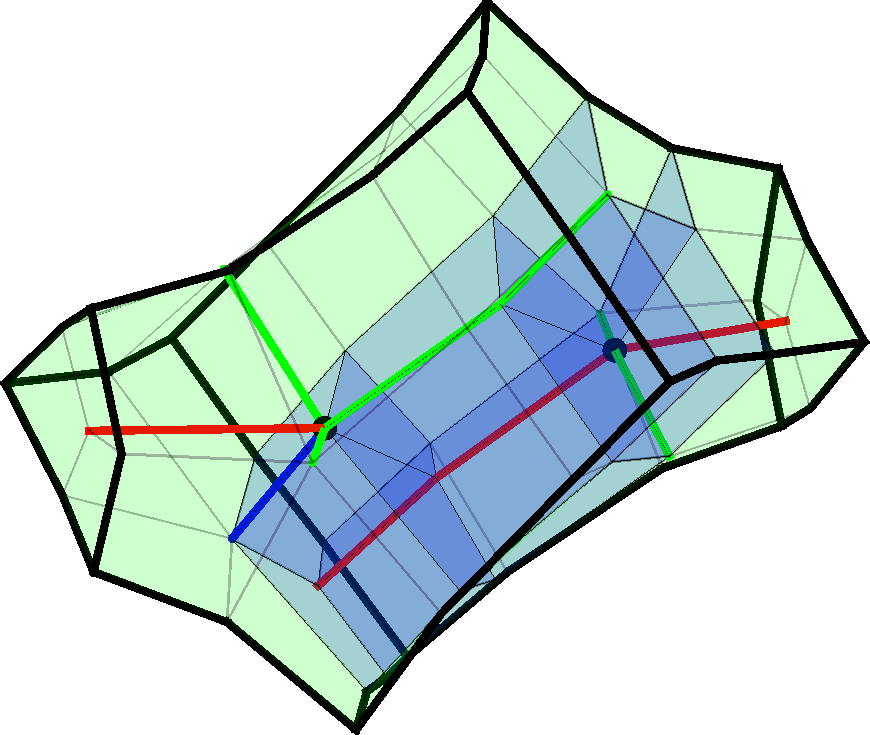}
    \includegraphics[width=.115\textwidth]{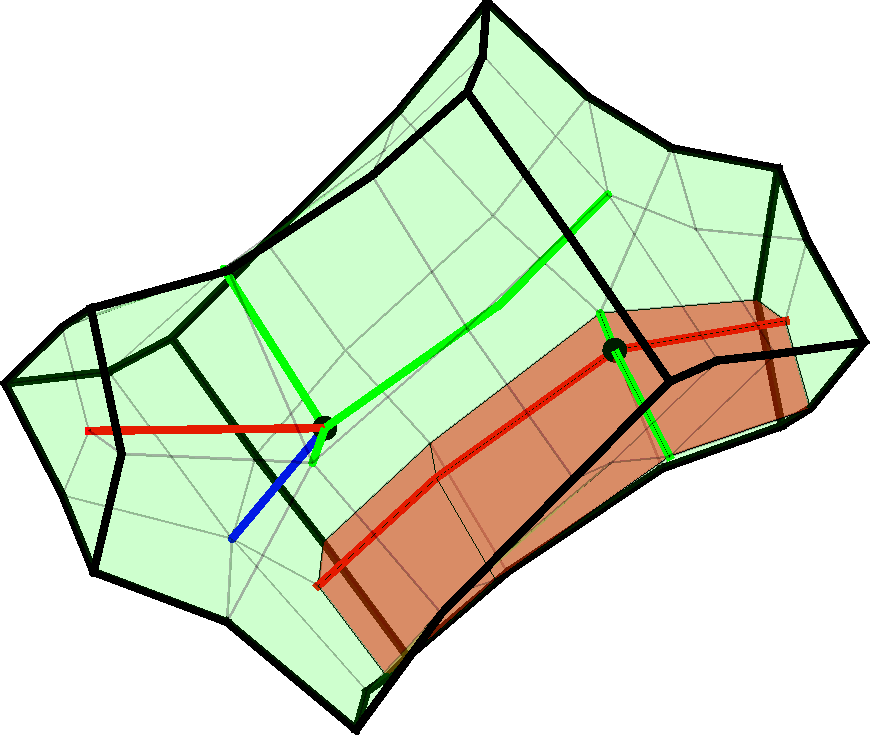}
    \includegraphics[width=.115\textwidth]{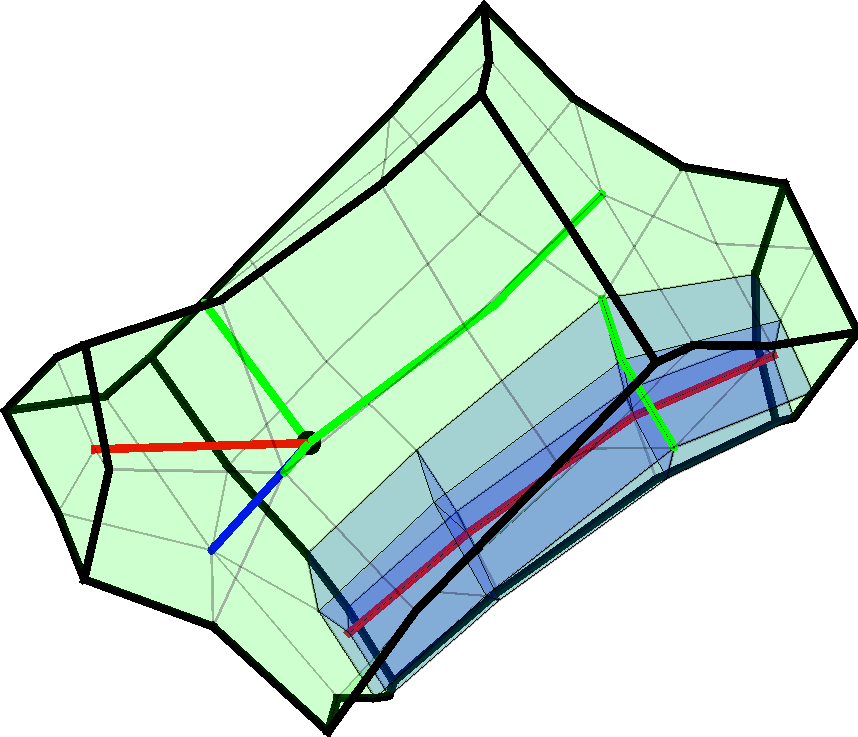}
    \includegraphics[width=.115\textwidth]{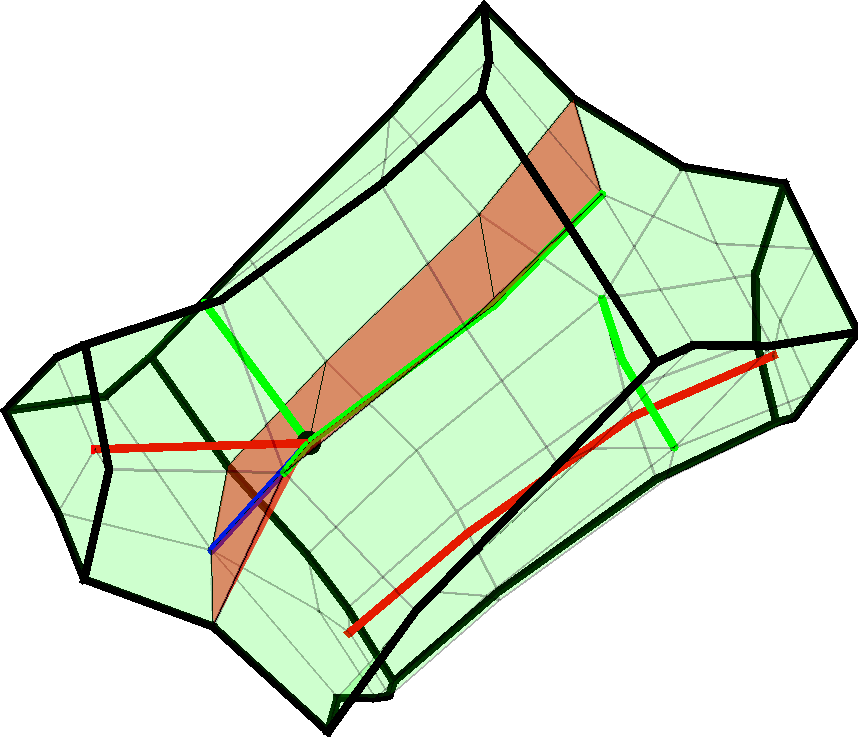}
\\
    \includegraphics[width=.115\textwidth]{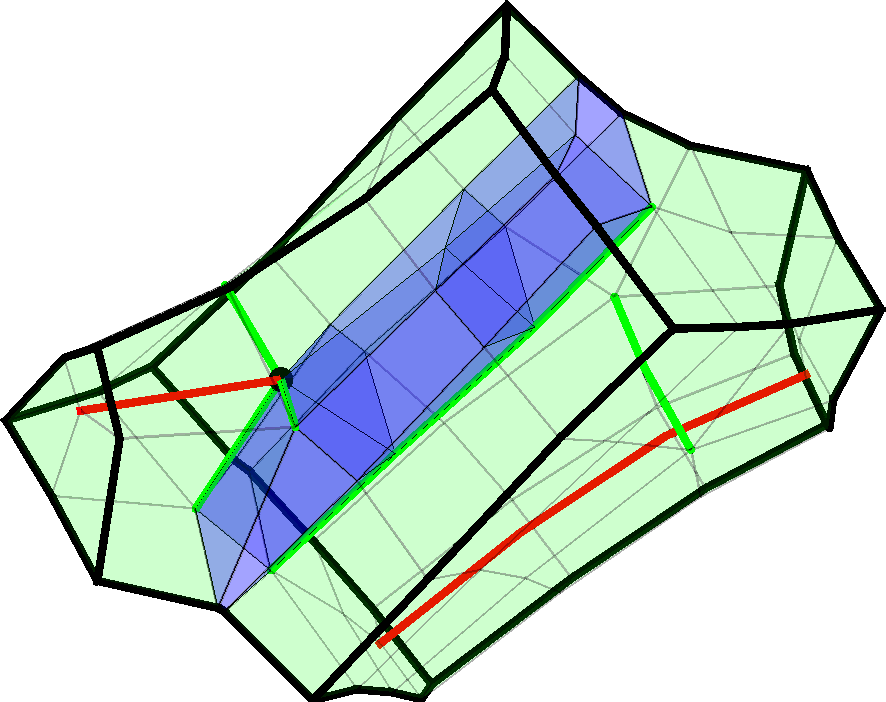}
    \includegraphics[width=.115\textwidth]{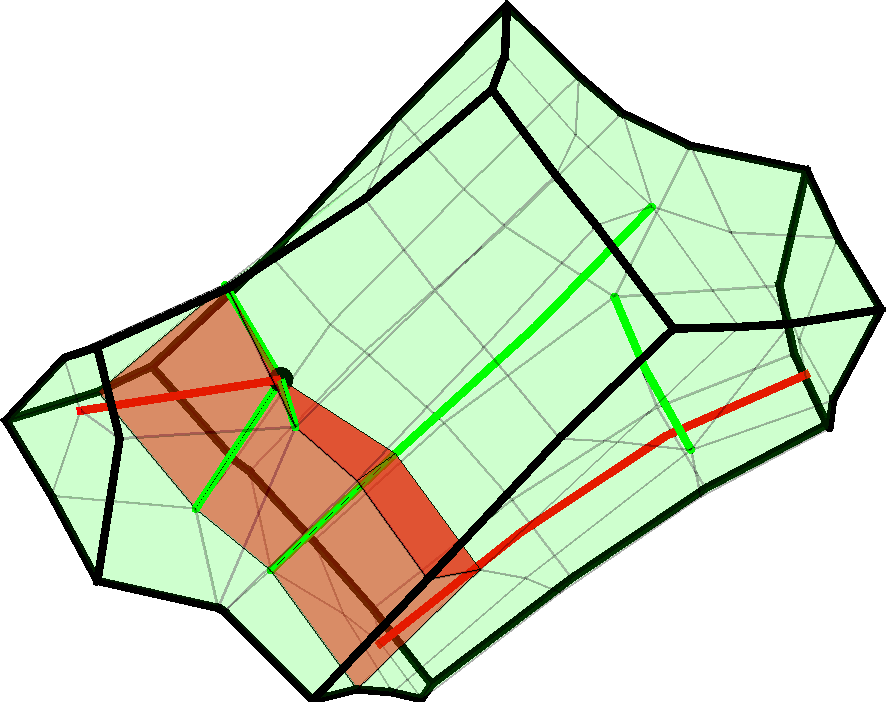}
    \includegraphics[width=.115\textwidth]{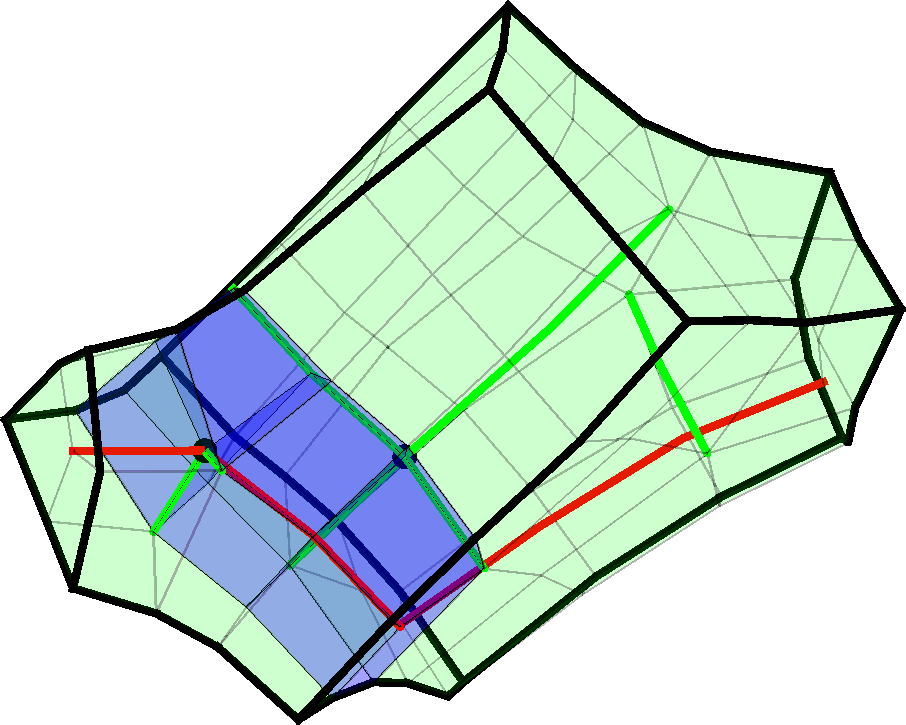}
    \includegraphics[width=.115\textwidth]{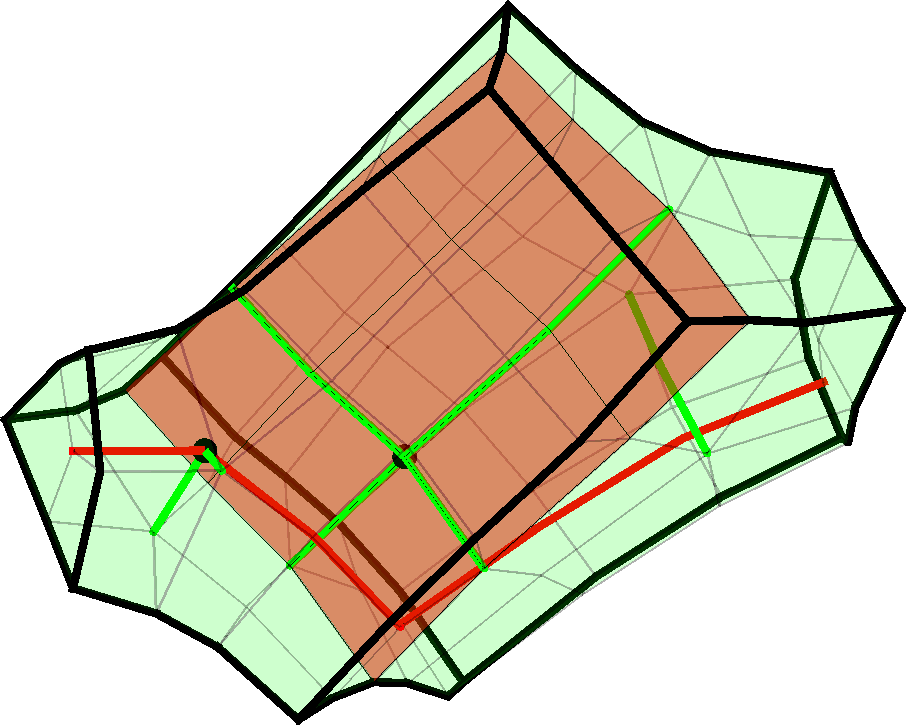}
    \includegraphics[width=.115\textwidth]{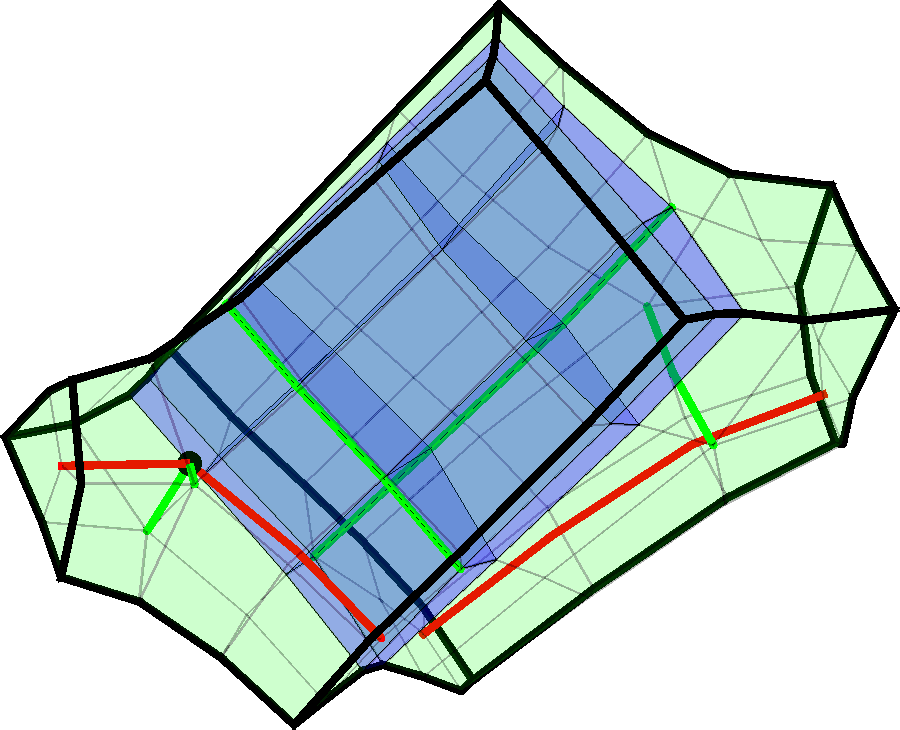}
    \includegraphics[width=.115\textwidth]{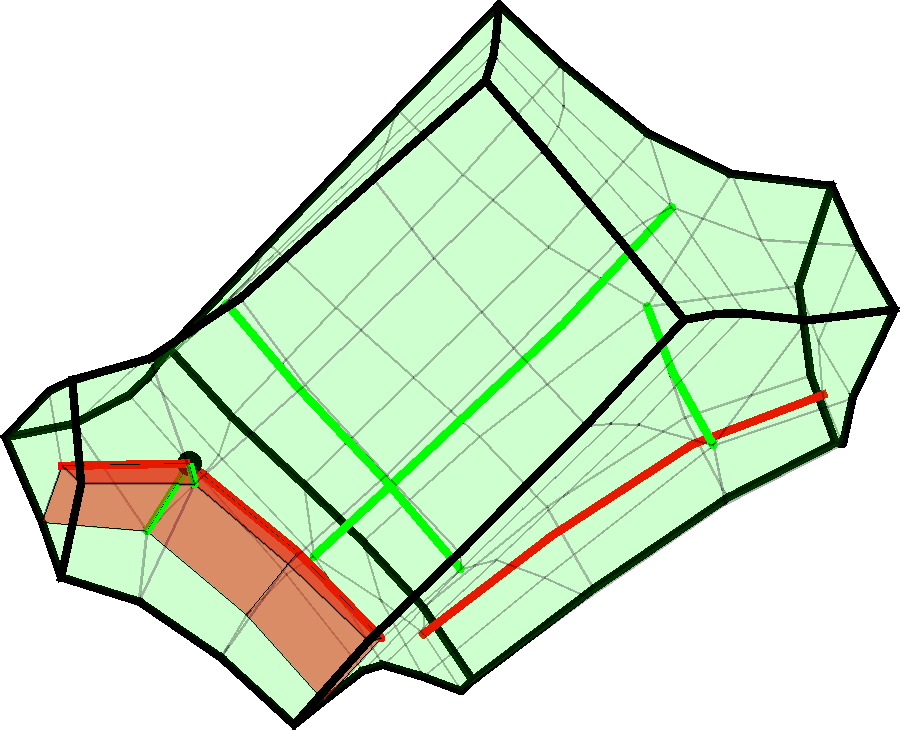}
    \includegraphics[width=.115\textwidth]{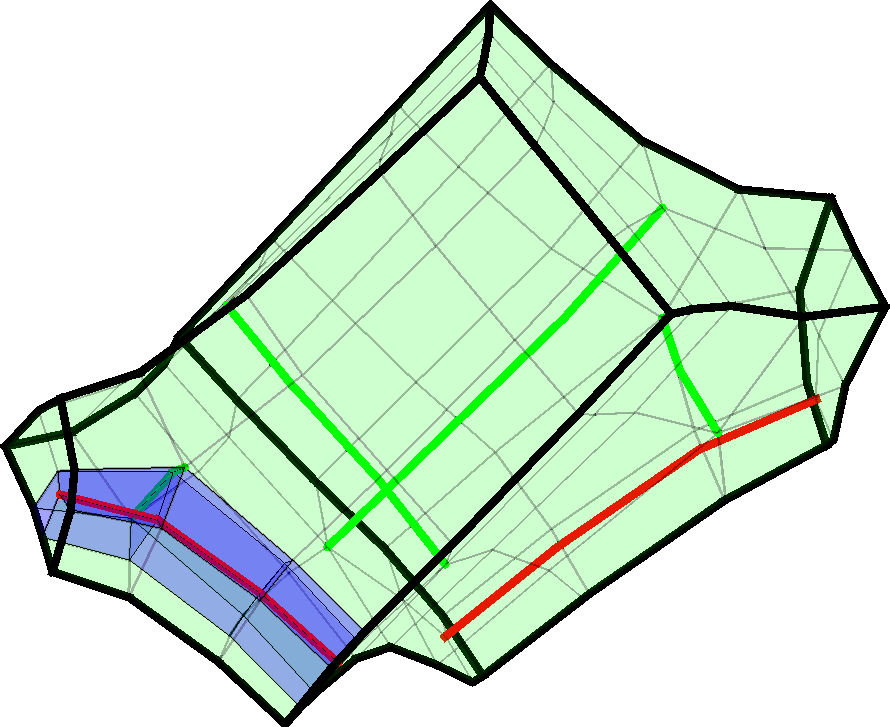}
    \includegraphics[width=.115\textwidth]{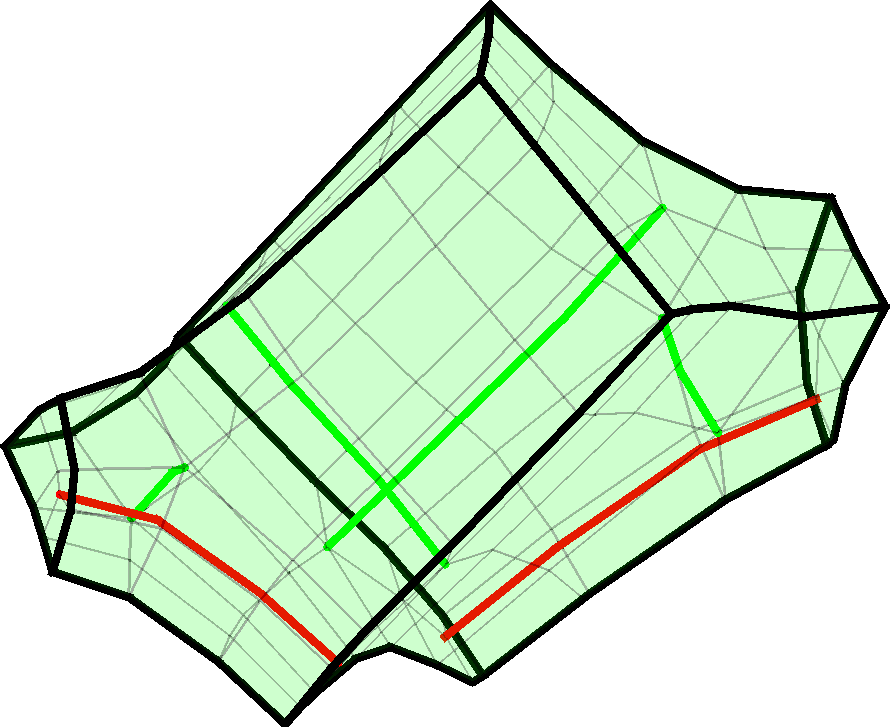}
    \caption{The (2,0,6) singular node is decomposed into four valence 5 and two valence 3 curves via seven sheet inflations.}
    \label{fig:206}
\end{figure*}

Finally the (0,0,12) singular node decomposition is shown in \autoref{fig:0012} to become six valence 5 curves. This singular node is especially interesting as we were unable to show a decomposition of the form
$$(0,0,12) = (0,4,4) +_n (0,2,8).$$
Even though the number of singular curves present is sufficient, we were not able to perform an inverse sheet inflation between (0,4,4) and (0,2,8) to obtain (0,0,12). This shows that the order in which singular curves are combined matters. As this figure is especially complex to comprehend, we offer the following roadmap of how the decomposition is performed.
$$(0,0,12) = \underbrace{\left((0,3,6) +_5 (0,3,6)\right)}_{(0,2,8,1)} \,+_6\, \underbrace{\left((0,4,4) +_5 (0,4,4)\right)}_{(0,4,4,1)}$$

\begin{figure*}
    \centering
    \includegraphics[width=.115\textwidth]{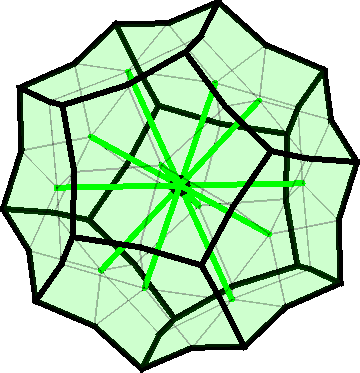}
    \includegraphics[width=.115\textwidth]{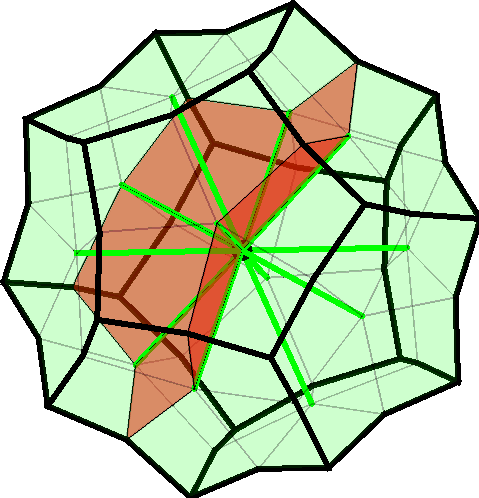}
    \includegraphics[width=.115\textwidth]{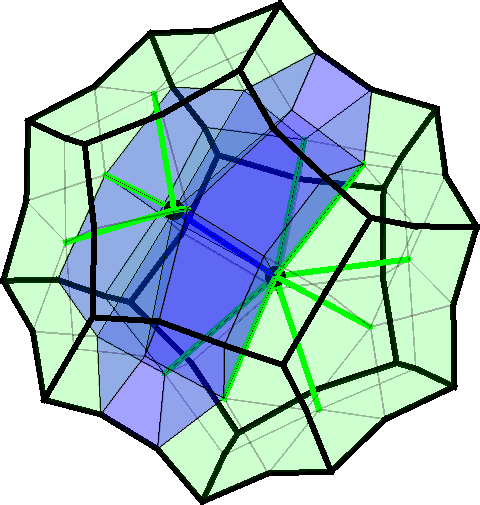}
    \includegraphics[width=.115\textwidth]{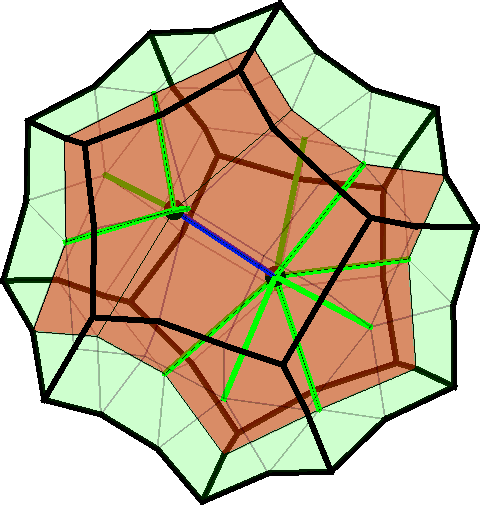}
    \includegraphics[width=.115\textwidth]{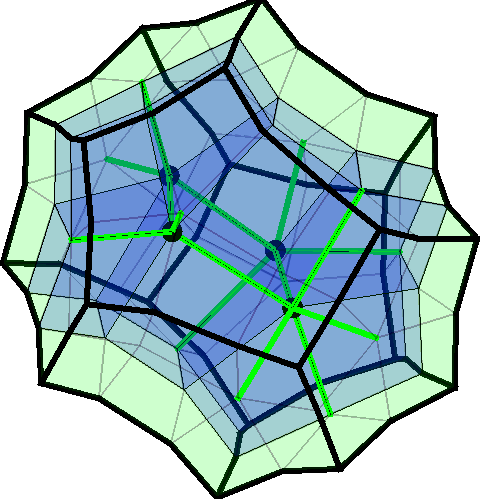}
    \includegraphics[width=.115\textwidth]{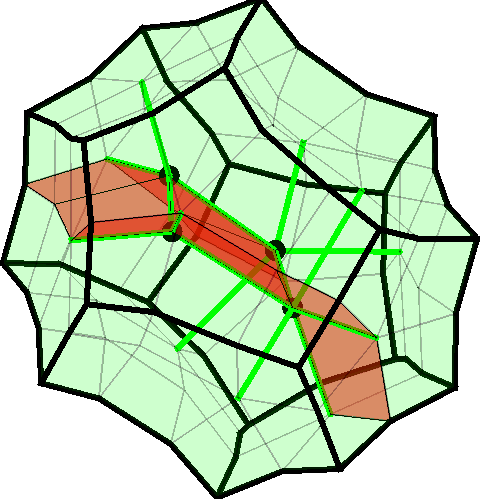}
    \includegraphics[width=.115\textwidth]{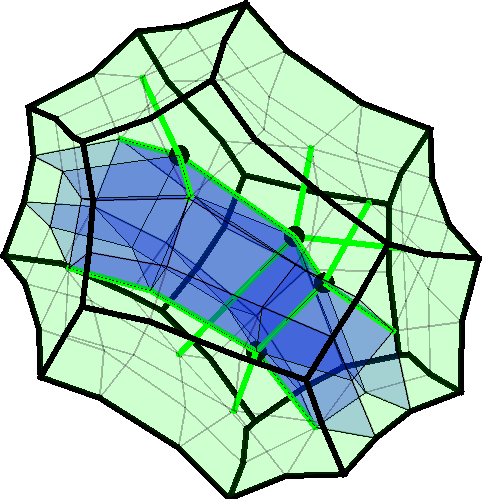}
    \includegraphics[width=.115\textwidth]{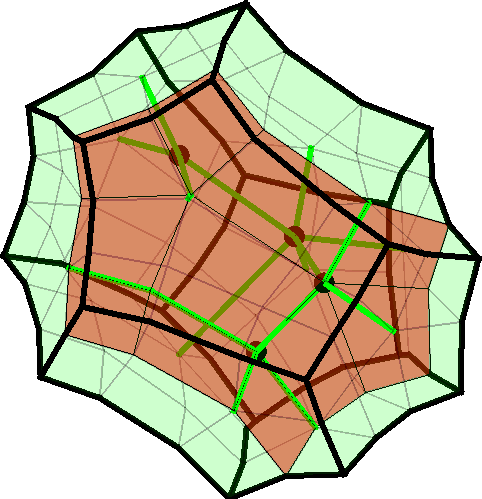}
    \\
    \includegraphics[width=.115\textwidth]{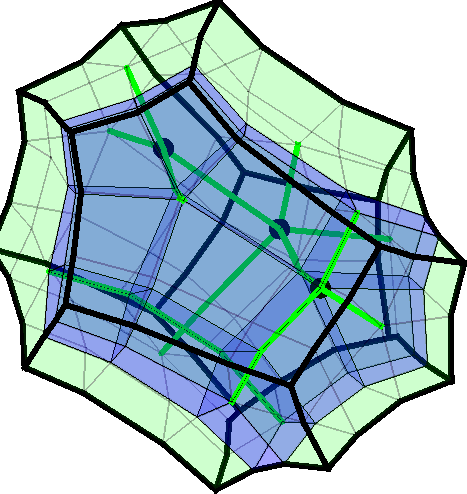}
    \includegraphics[width=.115\textwidth]{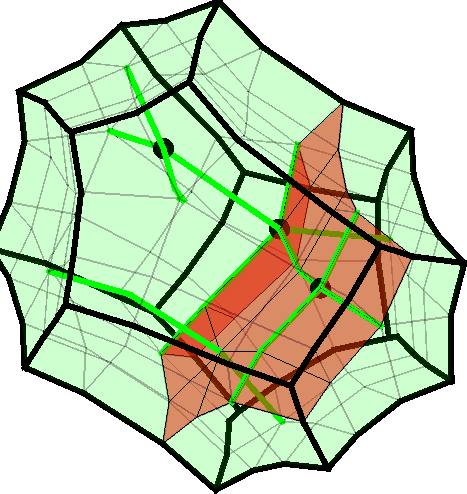}
    \includegraphics[width=.115\textwidth]{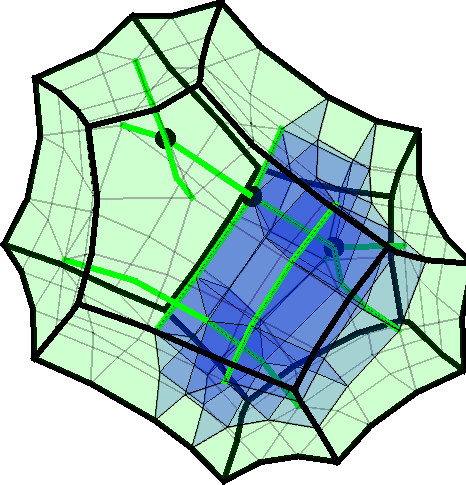}
    \includegraphics[width=.115\textwidth]{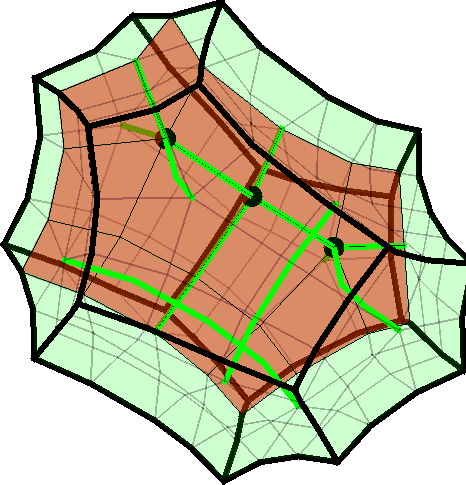}
    \includegraphics[width=.115\textwidth]{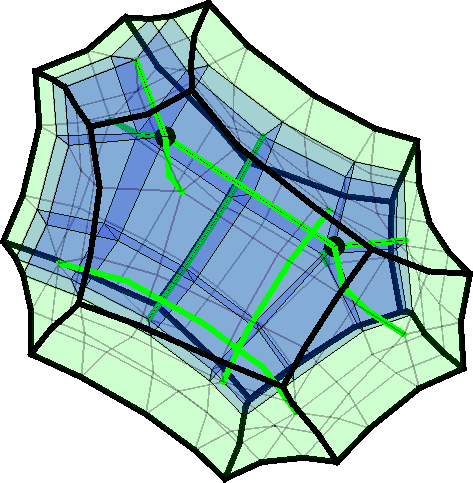}
    \includegraphics[width=.115\textwidth]{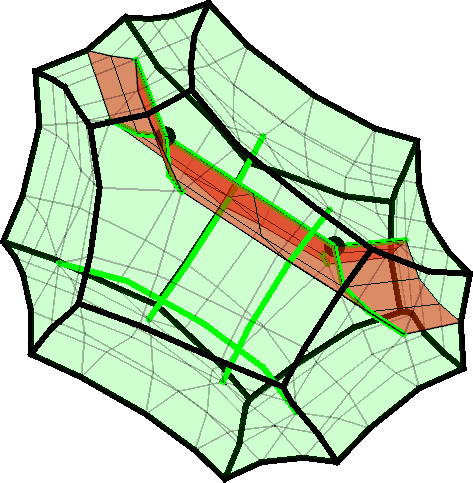}
    \includegraphics[width=.115\textwidth]{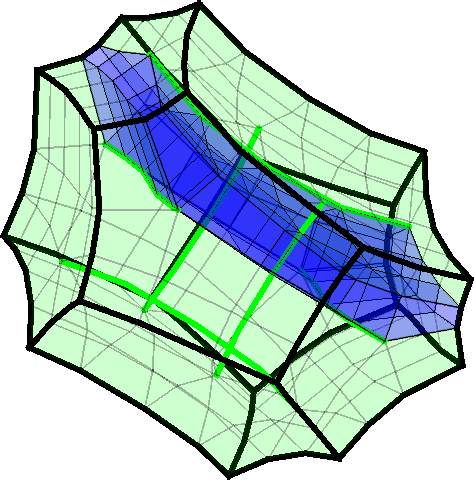}
    \includegraphics[width=.115\textwidth]{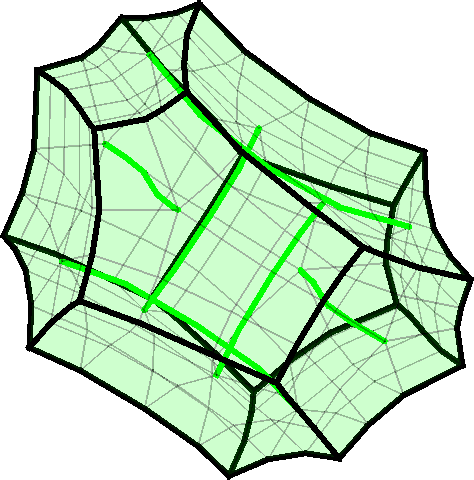}
\caption{The (0,0,12) singular node is decomposed into six valence 5 curves via seven sheet inflations.}
    \label{fig:0012}
\end{figure*}

\section{Decomposing General Singular Nodes}
Given that the eight singular nodes of valence 3, 4, or 5 are decomposable into singular curves, a natural next question is whether decomposition extends to higher valence singular nodes. 
In fact, the decomposition of the (0,0,12) and (2,0,6) both already required decomposing singular nodes with valence 6: (0,4,4,1), (0,2,8,1) and (1,3,3,1). We will refer to previously known decomposable singular nodes and their associated sphere triangulations as \emph{base cases}.
To generalize decomposability of singular nodes we offer the following result. 
\begin{proposition}
\label{prop:gendecomp}
Given a sphere triangulation $\T$ with some vertex $u$ of degree larger than 5, there exists a splitting such that either the number of vertices in both resulting triangulations decreases or the resulting triangulations are base cases.
\end{proposition}
\begin{proof}
The local neighborhood of $u$ is an umbrella $\U$ of at least 6 triangles. The boundary of this umbrella is a cycle of at least 6 vertices denoted by $\C$.
To construct a splitting of $\T$ into triangulations of fewer vertices, we need a pair of vertices $a$ and $b$ adjacent to $u$ that are at least 3 edges apart from each other in $\C$ such that there is path $p$ from $a$ to $b$ through the interior of $\T - \U$. This construction is illustrated in \autoref{fig:proof}. 
The sequence of edges $[(u a), p, (b u)]$ partitions $\T$ into $\D_1$ and $\D_2$ where each disk triangulation has at least 2 interior vertices. 
Since splitting a sphere triangulation replaces all vertices on the interior of either side with just one new vertex each, both resulting triangulations will have fewer vertices than $\T$. 
For readability, we leave more detailed construction of the splitting to supplementary materials. 
\end{proof}

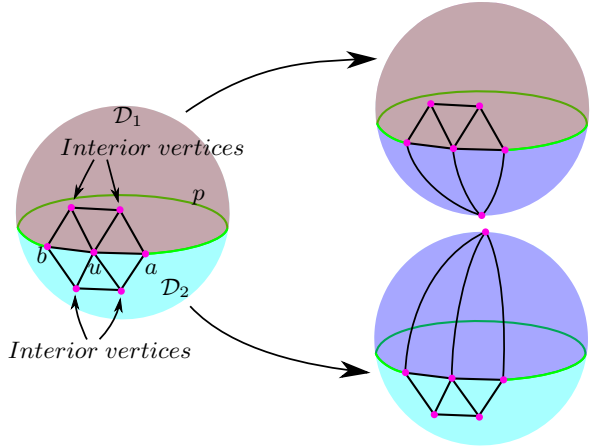
\begin{figure}
    \centering
    \includegraphics[width=1\columnwidth]{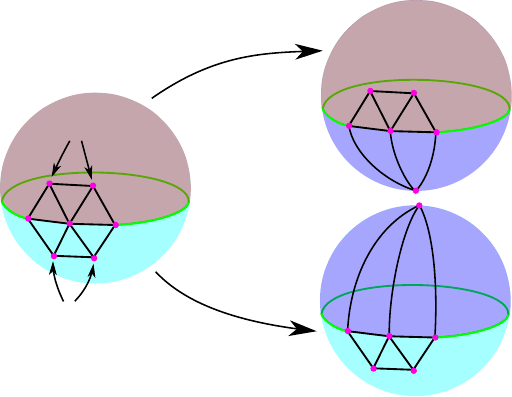}
    \put(-209,69){$b$}
    \put(-189.6,65){$u$}
    \put(-168,65){$a$}
    \put(-150,93){$p$}
    \put(-180,122){$\D_1$}
    \put(-162,57){$\D_2$}
    \put(-200,110){\small$Interior\;vertices$}
    \put(-220,33){\small$Interior\;vertices$}
    \caption{Illustration of how to find a cycle such that splitting along that cycle results in two sphere triangulations, each with fewer vertices. The only requirement is that there is a vertex $u$ of degree $\geq6$. The required cycle is then $[(ua),p,(bu)]$.}
    \label{fig:proof}
\end{figure}

Applying the splitting in \autoref{prop:gendecomp} could result directly in base cases, where the rest of the decomposition is already known. If the splitting does not result in base cases, then it produces triangulations with fewer vertices. This can be repeated until there are not enough vertices to have a degree 6 vertex. Since sheet inflation at a singular node corresponds to splitting of a sphere triangulation, \autoref{prop:gendecomp} allows us to find a sequence of sheets whose inflation results in singular nodes that 
have lower than valence 6 singular edges. We have already enumerated singular decompositions for all singular nodes with valence lower than 6 and can therefore decompose any singular node into singular curves.

A limitation of \autoref{prop:gendecomp} is that it restricts attention to individual singular nodes while ignoring the full singular graph of the mesh. It can be challenging to extend a sheet known locally around a singular node to the rest of the hex mesh while guaranteeing no self-intersection occurs. We present our simplistic solution to extending sheets in \autoref{sec:sgdecomp} and leave more careful consideration of how to avoid self-intersection to future work.

\section{Results}

\begin{algorithm}[t]
\begin{algorithmic}[1]
 \Procedure{Decompose-Singular-Graph}{$\H$}
 \Do
    \State $N \gets$ \Call{GetRandomSingularNode}{$\H$}
    \If{\Call{OnlyHasValence345}{$\H$,$N$}}
        \State $C \gets$ \Call{GetHardcodedCut}{$\H$, $N$}
    \Else
        \State $C \gets$ \Call{GetGeneralCut}{$\H$, $N$}
    \EndIf
    \State $S \gets$ \Call{PropagateCut}{$\H$, $C$}
    \State $H \gets$ \Call{SheetInflation}{$\H$, $S$}
 \doWhile{$N\neq\emptyset$}
 \State \Return $\H$
 \EndProcedure 
\end{algorithmic}
 \caption{Decomposes all singular nodes of a hex mesh into singular curves.}
 \label{alg:dec}
\end{algorithm}

\subsection{Singular Graph Decomposition}
\label{sec:sgdecomp}

We develop a procedure to perform singular mesh decompositions on general hex meshes. Pseudocode for this procedure is given in \autoref{alg:dec} and \autoref{alg:prop}
First, we randomly select a singular node. For any singular node with valence restricted to 3, 4, or 5, we hard code a subset of faces adjacent to the node to be inflated. If the node has valence 6 or higher, we use \autoref{prop:gendecomp} (denoted GetGeneralCut in \autoref{alg:dec}) to select these faces. These faces form a partial sheet that decomposes the initially selected node, but need to be extended through the rest of the mesh in order to be inflatable.

Next we propagate the partial sheet throughout the hex mesh following \autoref{alg:prop}. Let a face be \emph{parallel} to the partial sheet if they share a regular edge but share no adjacent hexes. We greedily add parallel faces to the partial sheet until no more parallel faces can be found. Next we look for any interior singular vertices on the boundary of the partial sheet. 
If such a vertex is found, then we 
compute the smallest number of new faces that need to be added to the partial sheet so that its boundary excludes this singular vertex. This is denoted by Put-v-In-S in \autoref{alg:prop} and is equivalent to a graph shortest path computation on the triangulation representing this singular vertex. 

These two steps are repeated until no more parallel faces can be found, and the boundary of the partial sheet is entirely on the boundary of the hex mesh. If at any stage of the algorithm, the partial sheet became non-manifold then the sheet propagation algorithm has failed. If the sheet is manifold then
we inflate it resulting in the decomposition of at least one singular node.
All results shown were generated by \autoref{alg:dec}.

\begin{algorithm}[t]
\begin{algorithmic}[1]
 \Procedure{PropagateCut}{$\H$, $S$}
 \State $\Q \gets $\Call{GetFaces}{$\H$}
 \While{$\exists f\in \Q: $\Call{Parallel}{$\H$, $S$, $f$}   }
    \State $S \gets S \cup f$
 \EndWhile
 
 \State $\V \gets $\Call{GetVertices}{$\H$}
 \State $\V_S \gets $\Call{GetInteriorSingularVertices}{$\H$}
 \If{$\exists v \in (\partial S \cap \V_S)$ }
    \State $S \gets$ \Call{Put-v-In-S}{$\H$, $S$, $v$}
 \Else
    \If{\Call{NonManifold}{$\H$, $S$}}
        \State \Return ERROR
    \Else
        \State \Return $S$
    \EndIf
 \EndIf
 
 \State \Return \Call{PropagateCut}{$\H$, $S$}
 \EndProcedure
\end{algorithmic}
 \caption{Propagates a partial sheet into a full sheet recursively. }
 \label{alg:prop}
\end{algorithm}

Applying our decomposition to a hex mesh of a sphere reveals that it has the same singular graph structure as that of a padded tetrahedron. \autoref{fig:sphere_and_tet} shows this correspondence where inflating one sheet that passes through seven singular nodes, simultaneously decomposes three of them. The end result is a singular graph composed of four (4,0,0) singular nodes. One of these nodes has singular curves that all connect directly to the boundary. The other four of these nodes connect to each other and the boundary via valence 3 singular curves in a tetrahedral arrangement. This singular graph is exactly what one obtains by padding a hex mesh of a regular tetrahedron i.e. padding a (4,0,0) node.

\begin{figure}
    \centering
    \includegraphics[width=.235\columnwidth]{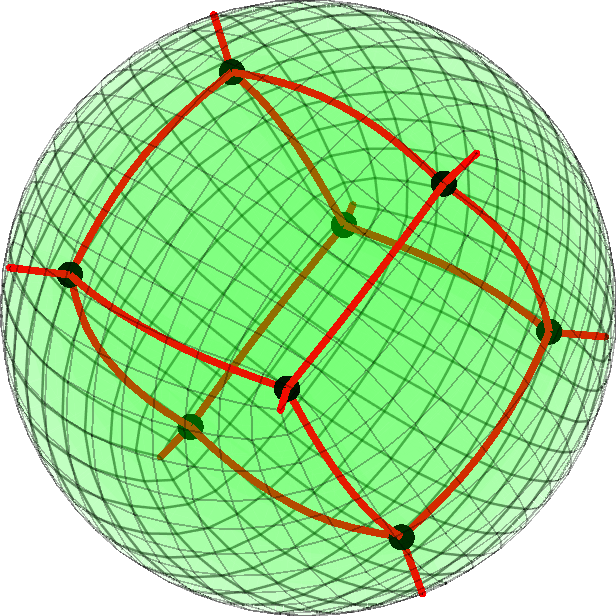}
    \includegraphics[width=.235\columnwidth]{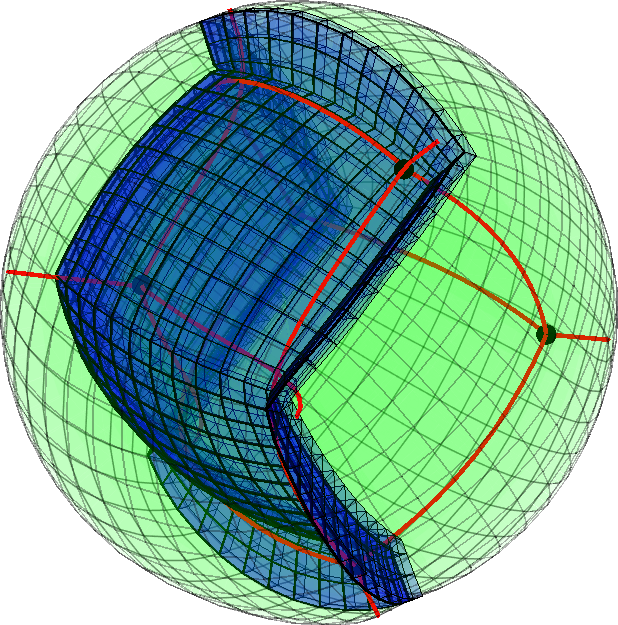}
    \includegraphics[width=.235\columnwidth]{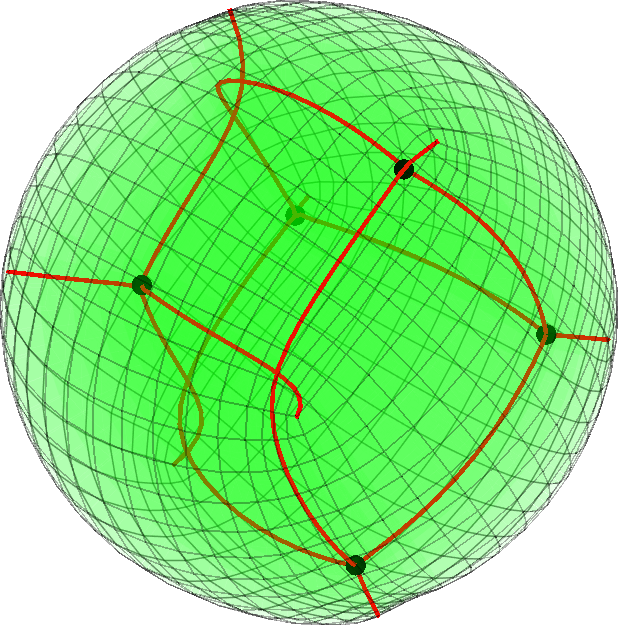}
    \includegraphics[width=.235\columnwidth]{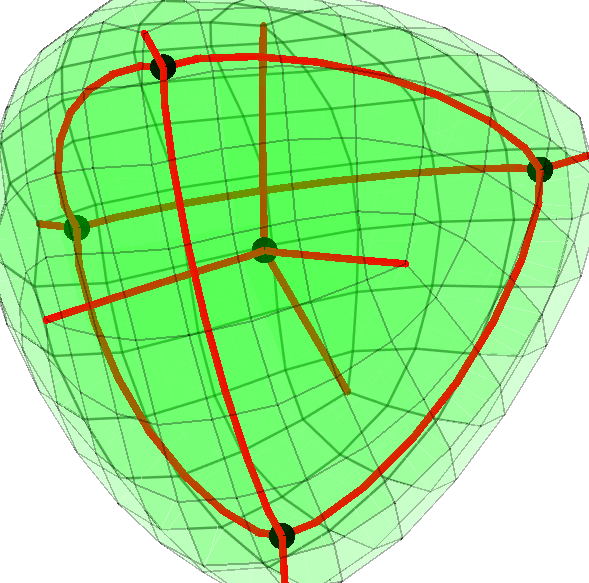}
    \caption{(Left) Hex mesh of sphere with singular graph. (Mid-left) Blue hexes are newly inflated hexes. (Mid-right) Hex mesh post-inflation. (Right) Singular graph of a padded hex mesh of a tetrahedron. The last two images have topologically equivalent singular graphs.}
    \label{fig:sphere_and_tet}
\end{figure}

Changing how the sheet cuts through the singular graph produces different intermediate and final singular graphs. 
In \autoref{fig:sgdecomp}, we decompose a padded cube in two different sequences and show the their intermediate singular graphs.
To improve clarity, we provide schematics of a subset of the singular graphs. The ending singular graphs from both sequences are also topologically distinct i.e. no purely geometric deformation maps one singular graph into the other. They do however appear to invariably contain a single singular cycle. 

The first sheet inflation of the second sequence results in the same singular graph as a padded hex mesh of a triangular prism: a padded (2,3,0). Since the hex mesh of a sphere has the same singular graph as the padded cube, these results indicate that singular graphs for a padded cube, padded tet, and padded triangular prism are identical up to a series of sheet inflation and collapses.

\begin{figure*}
    \centering
    \includegraphics[width=.091\textwidth]{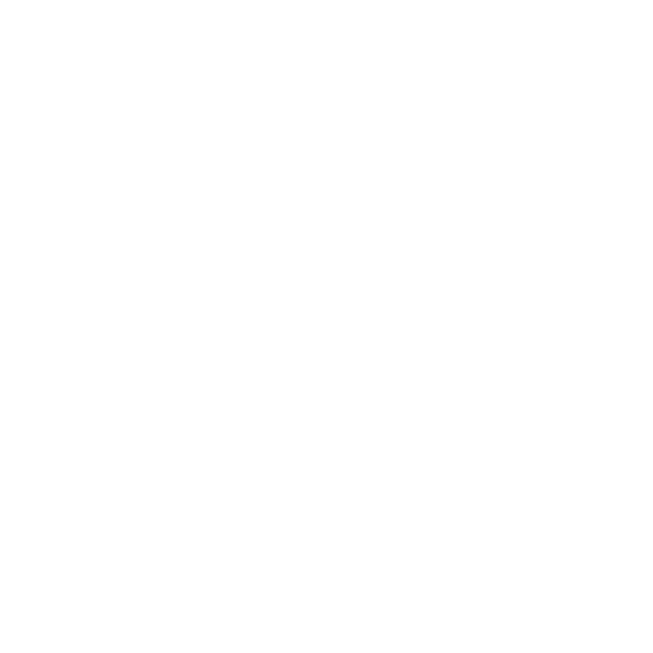}
    \includegraphics[width=.091\textwidth]{figures/schems/blank.png}
    \includegraphics[width=.091\textwidth]{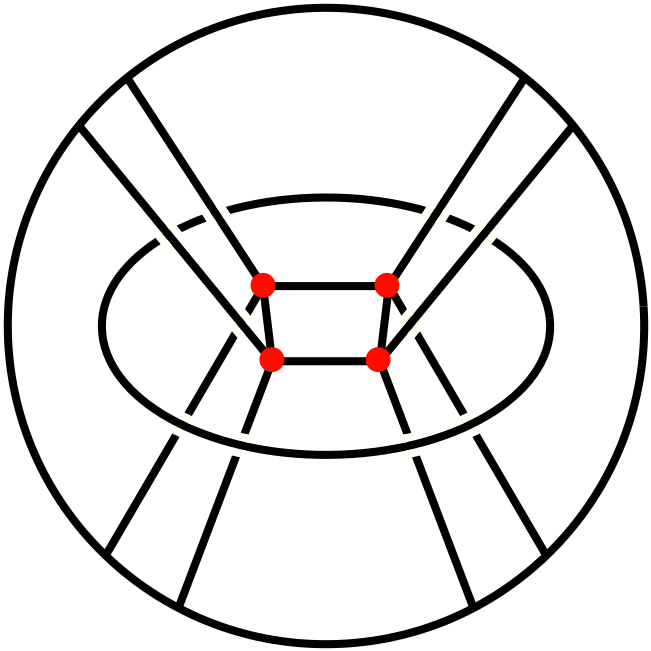}
    \includegraphics[width=.091\textwidth]{figures/schems/g5.png}
    \includegraphics[width=.091\textwidth]{figures/schems/blank.png}
    \includegraphics[width=.091\textwidth]{figures/schems/blank.png}
    \includegraphics[width=.091\textwidth]{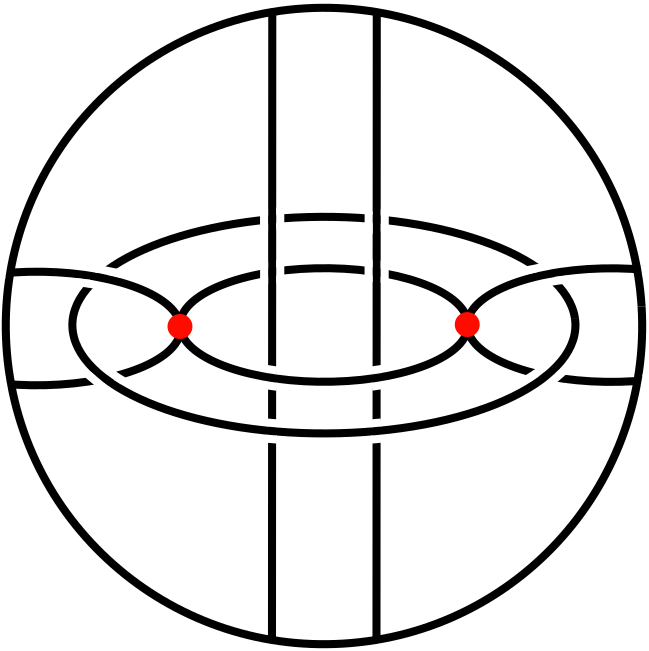}
    \includegraphics[width=.091\textwidth]{figures/schems/g3.png}
    \includegraphics[width=.091\textwidth]{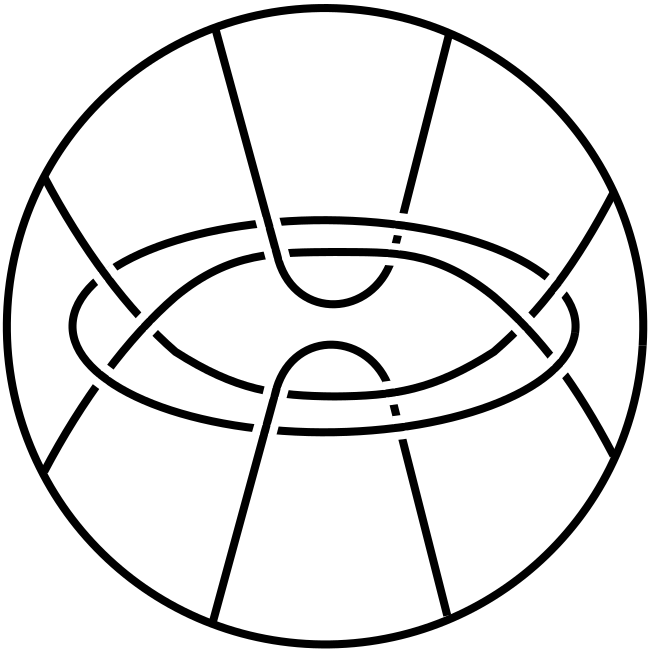}
    \includegraphics[width=.091\textwidth]{figures/schems/g1.png}
    \\
    \includegraphics[width=.091\textwidth]{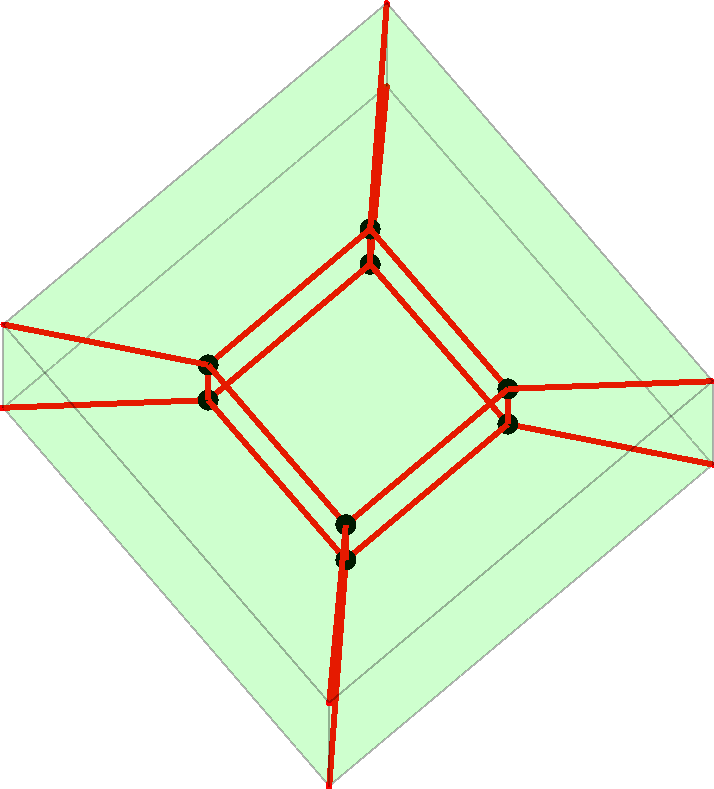}
    \includegraphics[width=.091\textwidth]{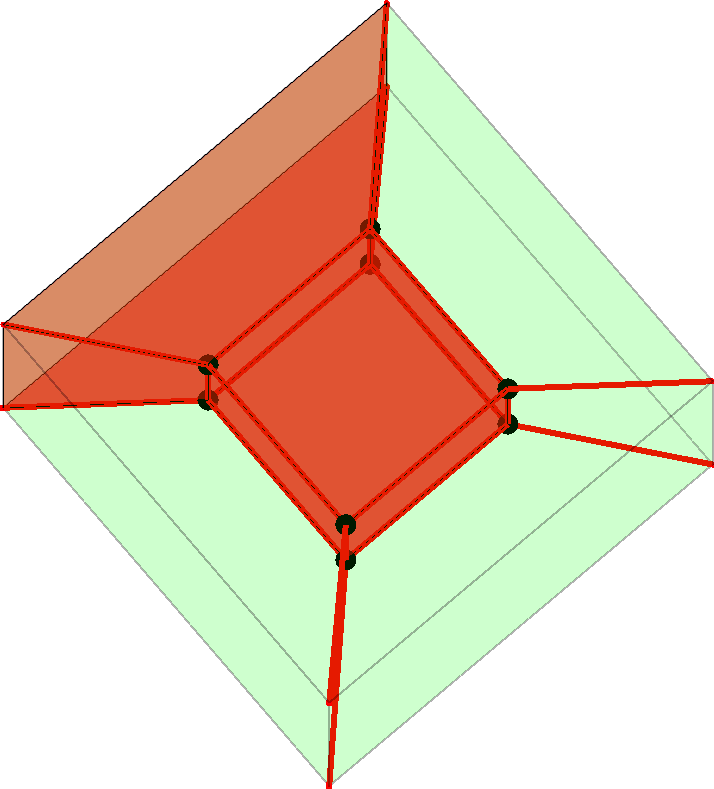}
    \includegraphics[width=.091\textwidth]{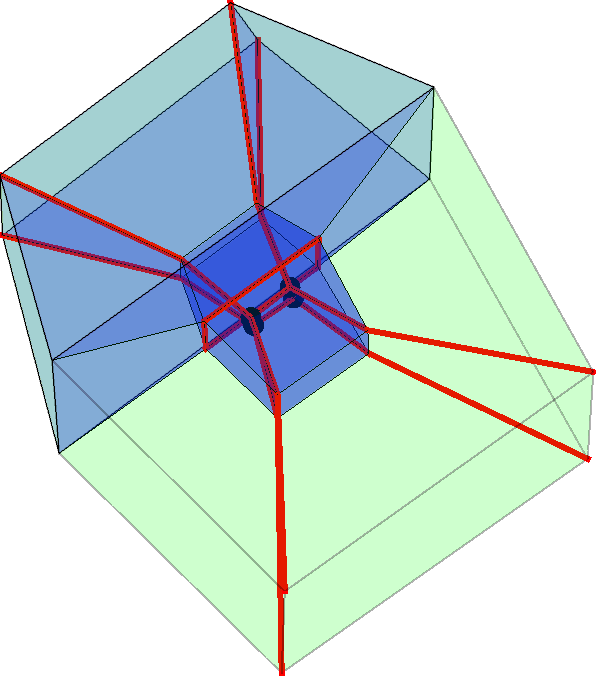}
    \includegraphics[width=.091\textwidth]{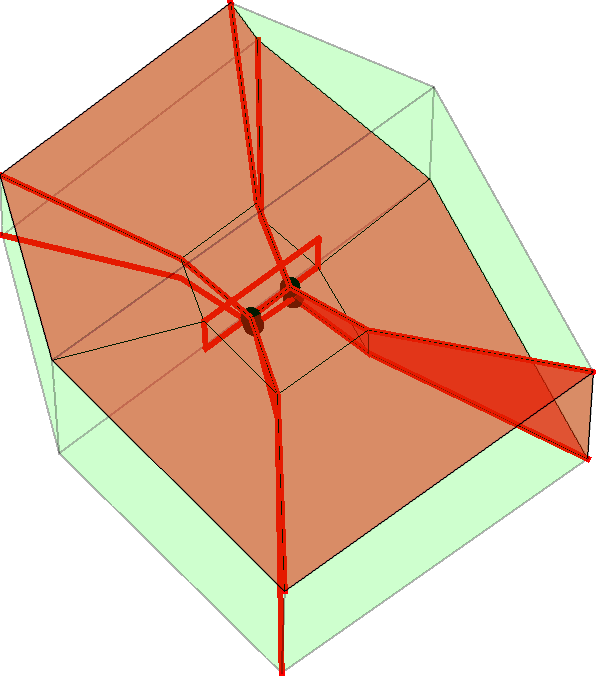}
    \includegraphics[width=.091\textwidth]{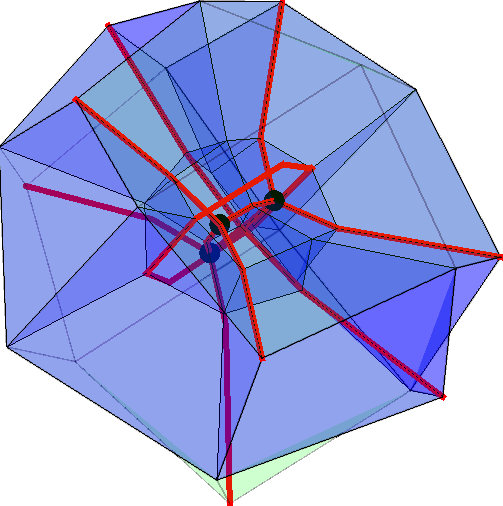}
    \includegraphics[width=.091\textwidth]{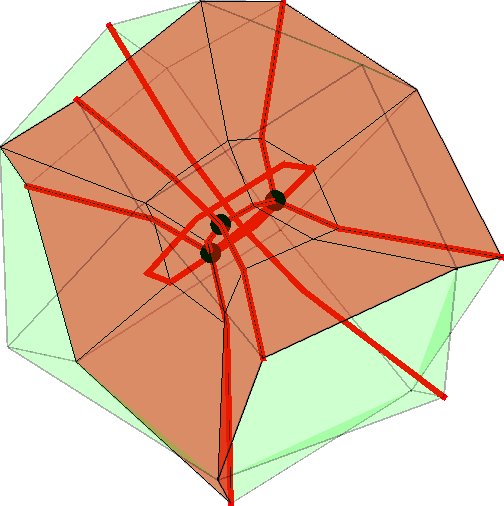}
    \includegraphics[width=.091\textwidth]{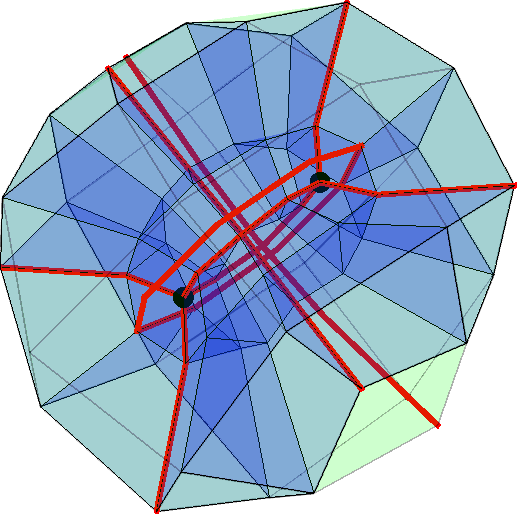}
    \includegraphics[width=.091\textwidth]{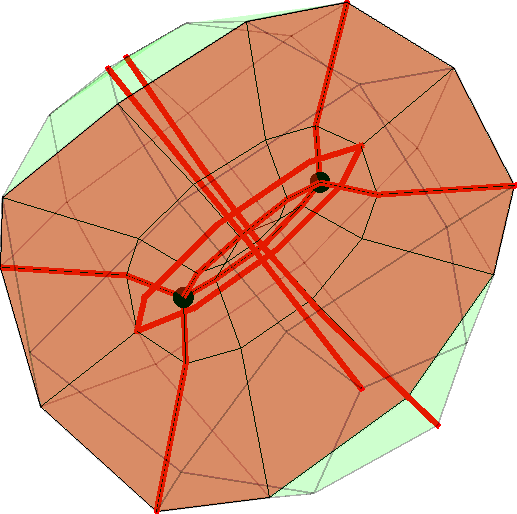}
    \includegraphics[width=.091\textwidth]{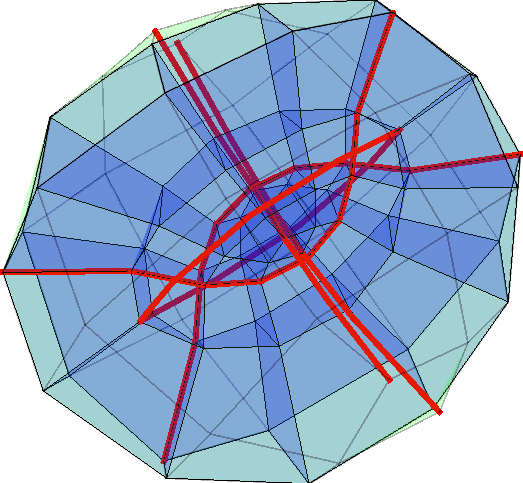}
    \includegraphics[width=.091\textwidth]{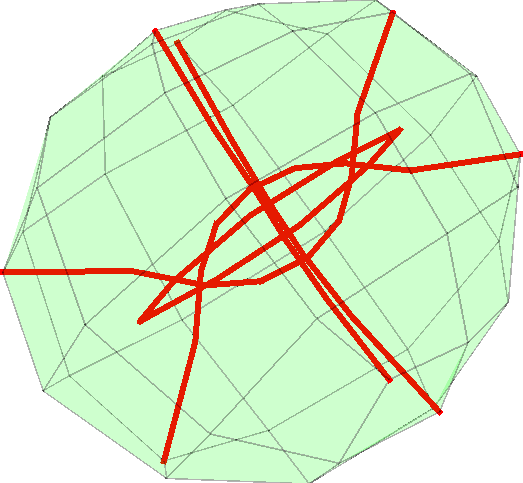}\\
    \includegraphics[width=.091\textwidth]{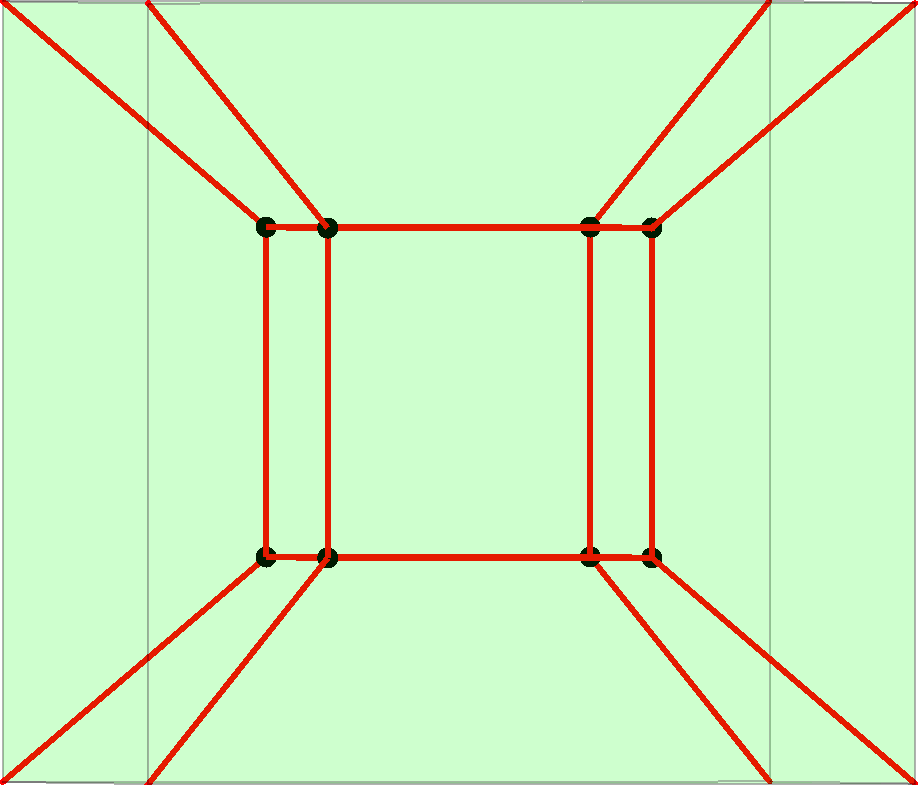}
    \includegraphics[width=.091\textwidth]{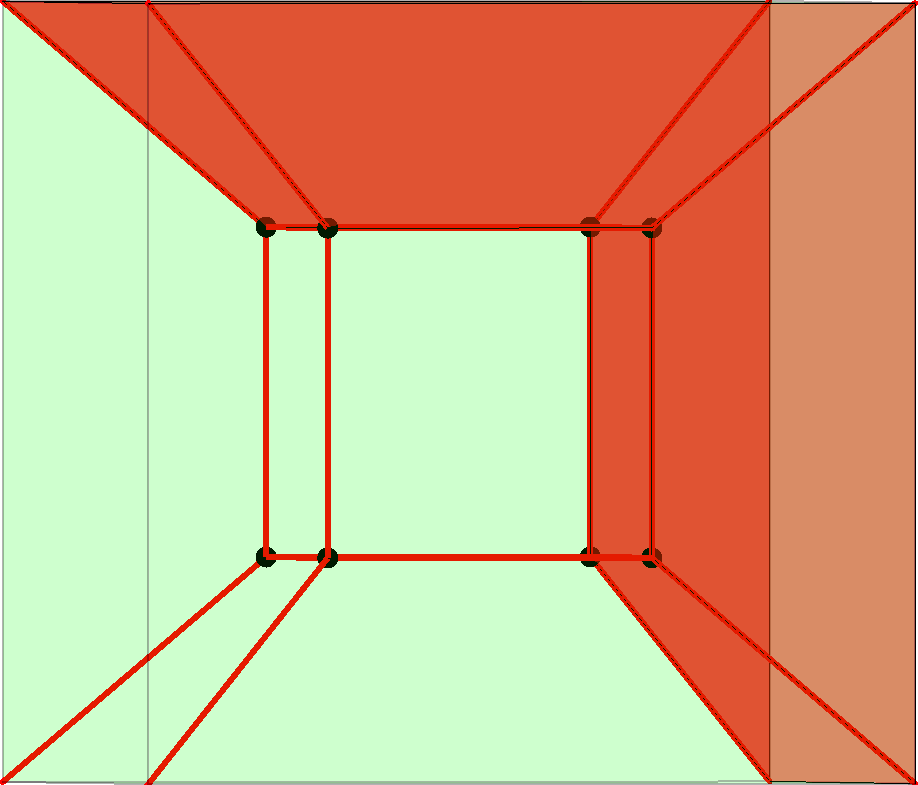}
    \includegraphics[width=.091\textwidth]{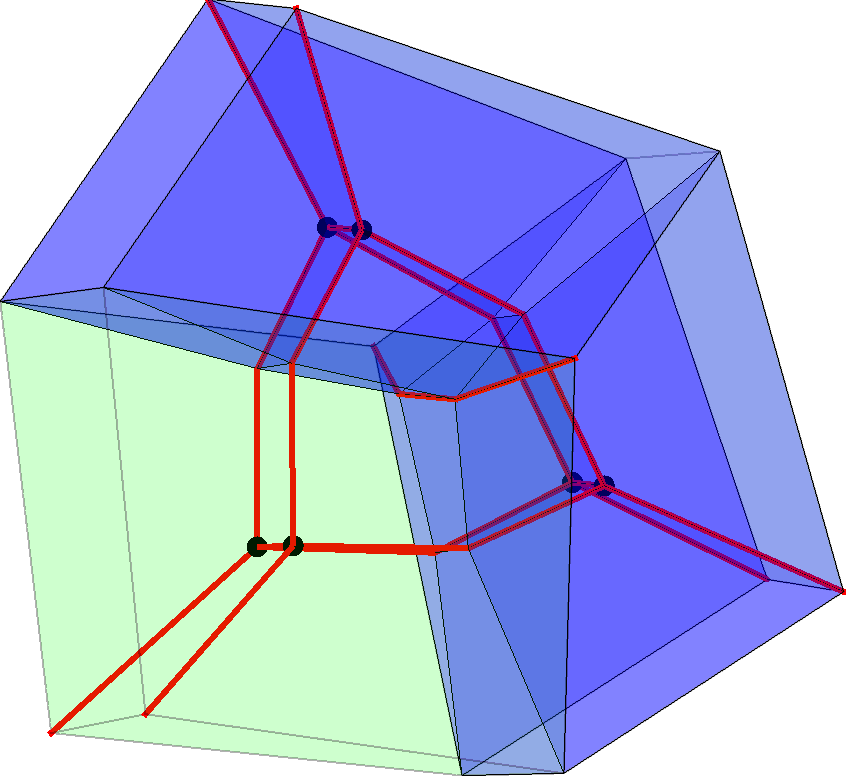}
    \includegraphics[width=.091\textwidth]{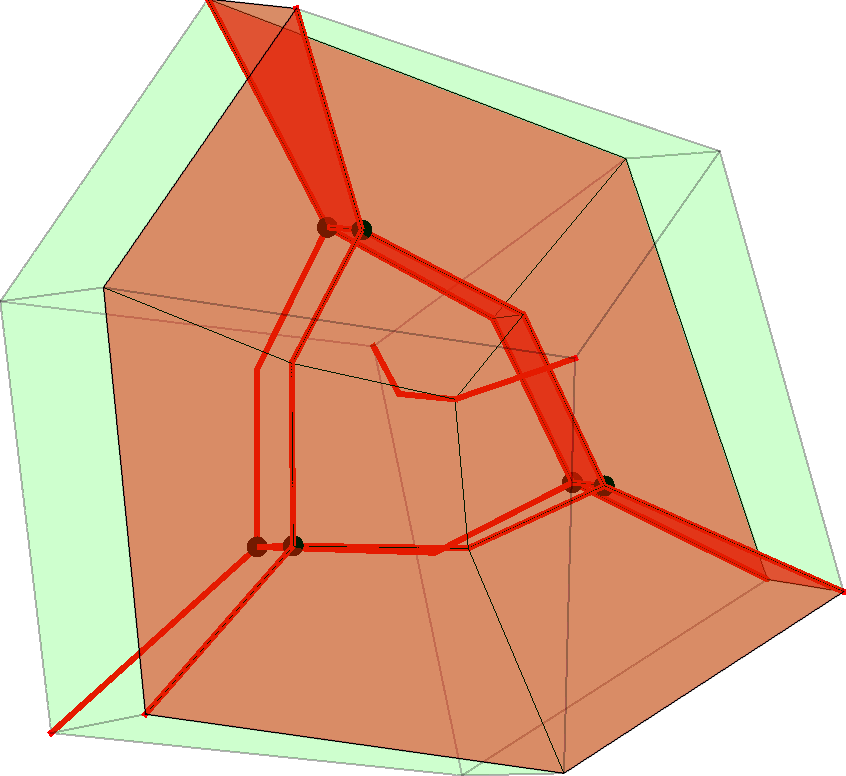}
    \includegraphics[width=.091\textwidth]{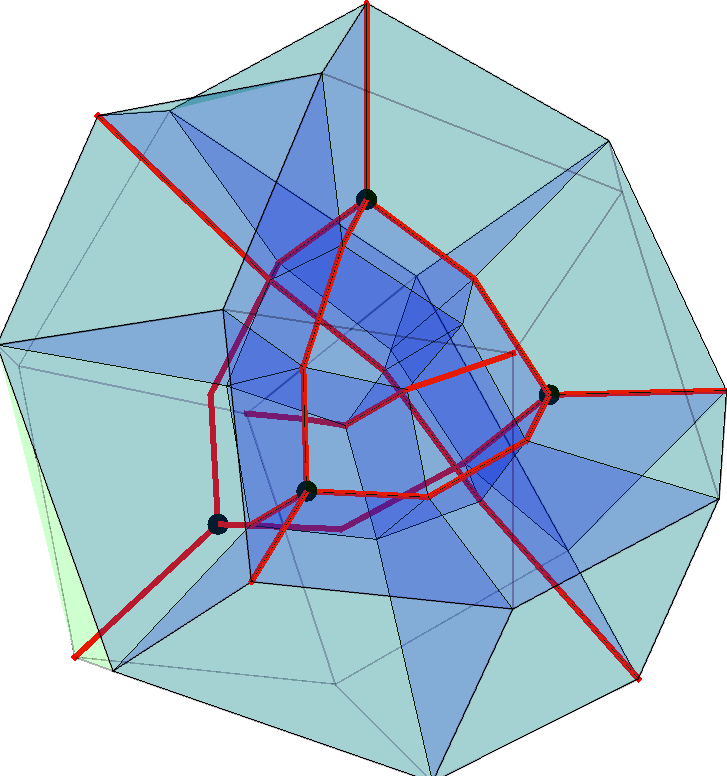}
    \includegraphics[width=.091\textwidth]{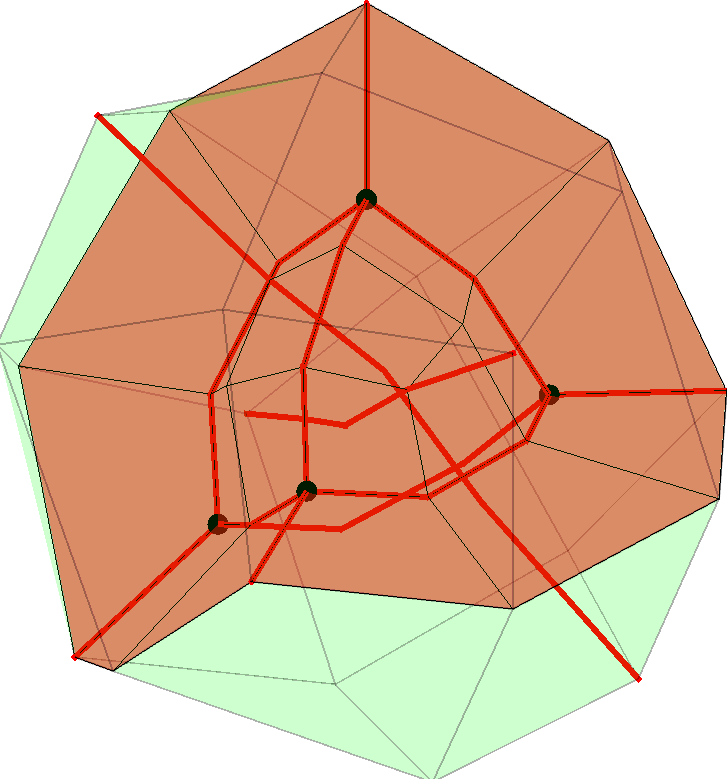}
    \includegraphics[width=.091\textwidth]{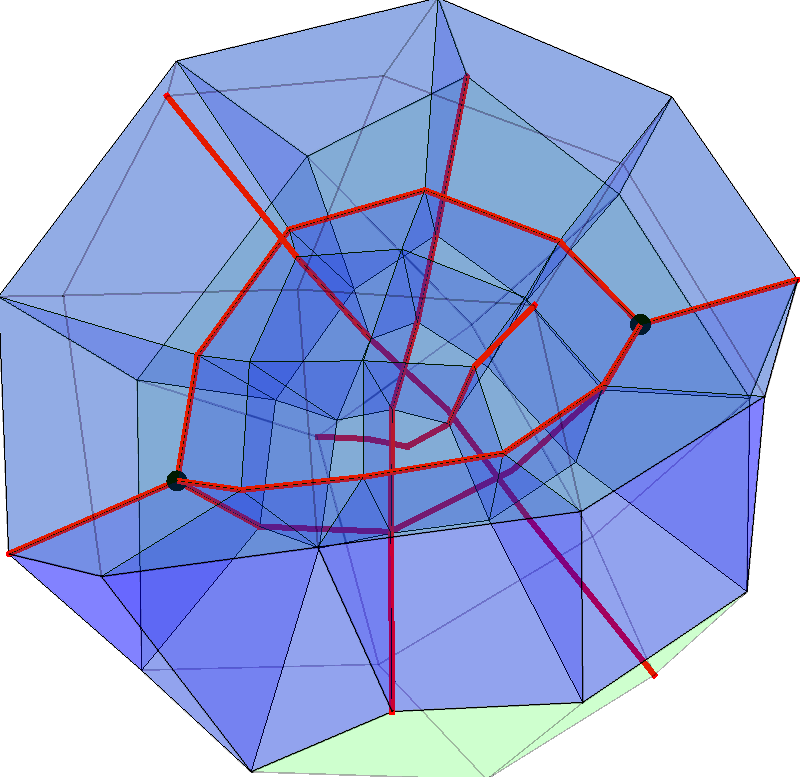}
    \includegraphics[width=.091\textwidth]{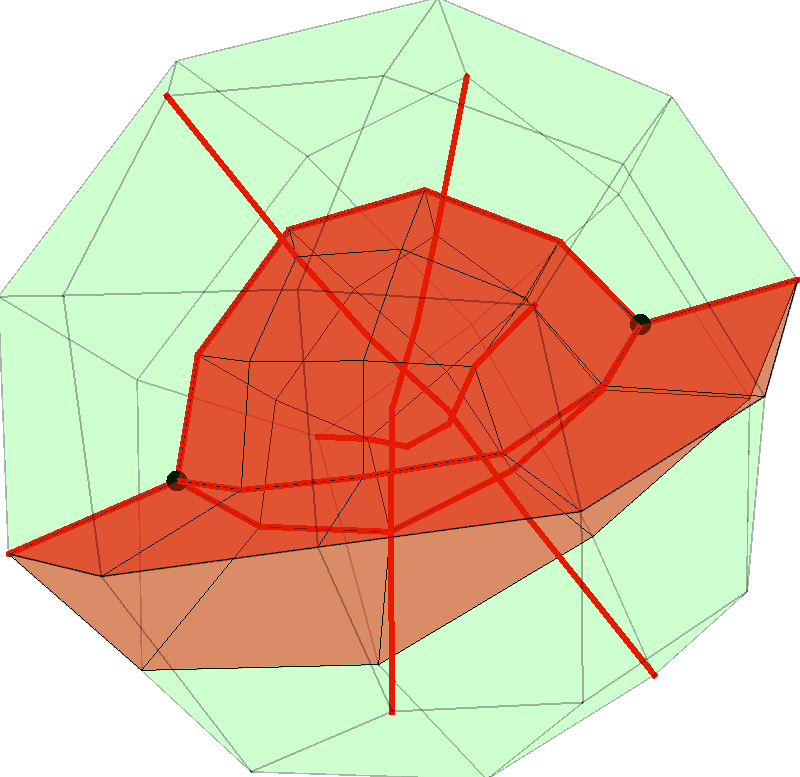}
    \includegraphics[width=.091\textwidth]{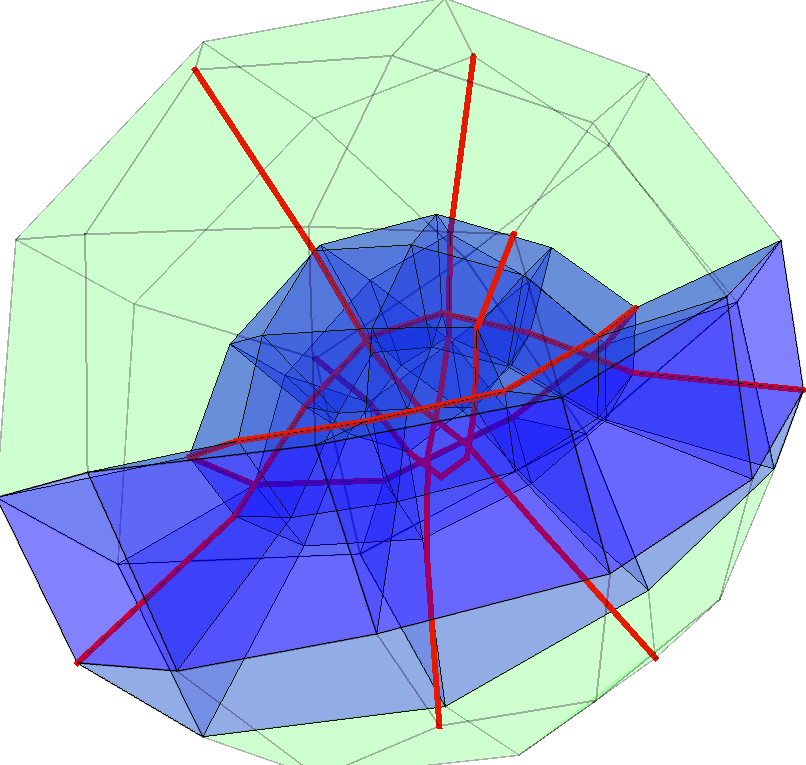}
    \includegraphics[width=.091\textwidth]{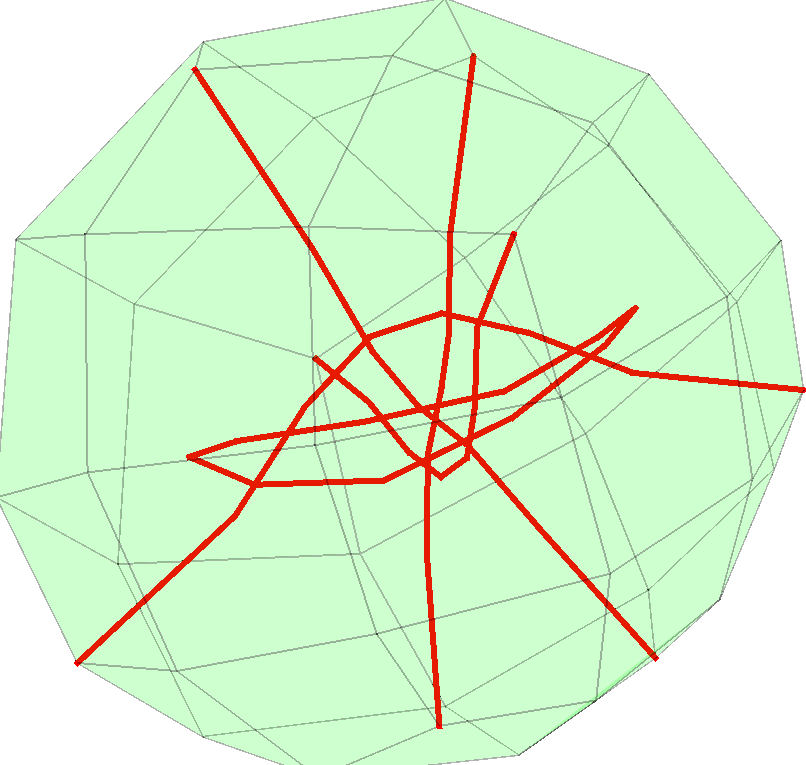}\\
    \includegraphics[width=.091\textwidth]{figures/schems/blank.png}
    \includegraphics[width=.091\textwidth]{figures/schems/blank.png}
    \includegraphics[width=.091\textwidth]{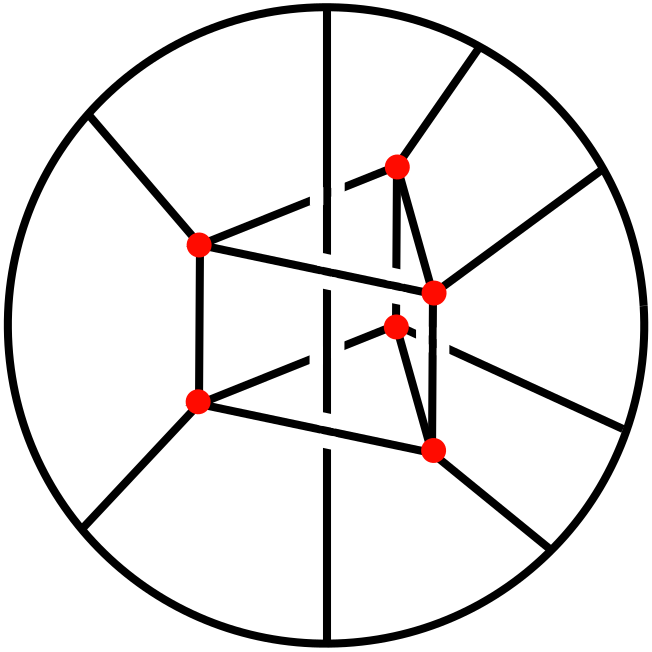}
    \includegraphics[width=.091\textwidth]{figures/schems/g4.png}
    \includegraphics[width=.091\textwidth]{figures/schems/blank.png}
    \includegraphics[width=.091\textwidth]{figures/schems/blank.png}
    \includegraphics[width=.091\textwidth]{figures/schems/blank.png}
    \includegraphics[width=.091\textwidth]{figures/schems/blank.png}
    \includegraphics[width=.091\textwidth]{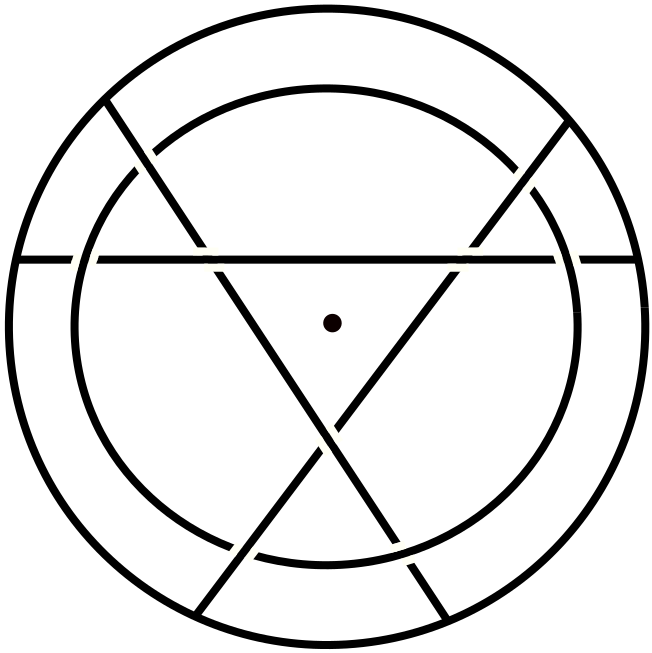}
    \includegraphics[width=.091\textwidth]{figures/schems/g2.png}
    \caption{We show two sequences of singular graph decomposition starting from the same hex mesh on the left to a fully decomposed singular graph on the right. The first row indicates select singular graph schematics for the first sequence. The last row indicates select singular graph schematics for the second sequence. Even though both singular graphs start out identical, the ending singular graphs are different due to different sheet inflations.}
    \label{fig:sgdecomp}
\end{figure*}

In \autoref{fig:LCpadded}, we apply our decomposition to more complex singular graphs. The first two rows depict the decomposition of the \textbf{G1} hex mesh. The starting singular graph consists of 12 nodes connected by 36 singular curves. This graph is successfully decomposed into seven singular curves, one of which is a closed cycle.
The last two rows depict the decomposition of the \textbf{G2} hex mesh.
The starting singular graph consists of 16 nodes connected by 40 singular curves. This graph is successfully decomposed into 12 singular curves, two of which are closed cycles.
While the starting singular graphs are different, both meshes are fully decomposed with the same number of sheet inflations.


\begin{figure*}
    \centering
    \includegraphics[width=.135\textwidth]{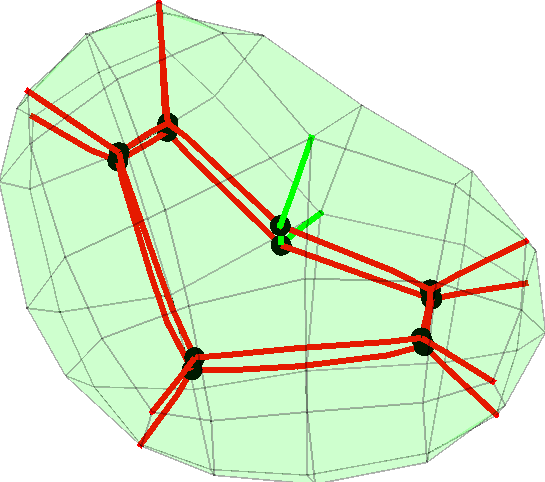}
\includegraphics[width=.135\textwidth]{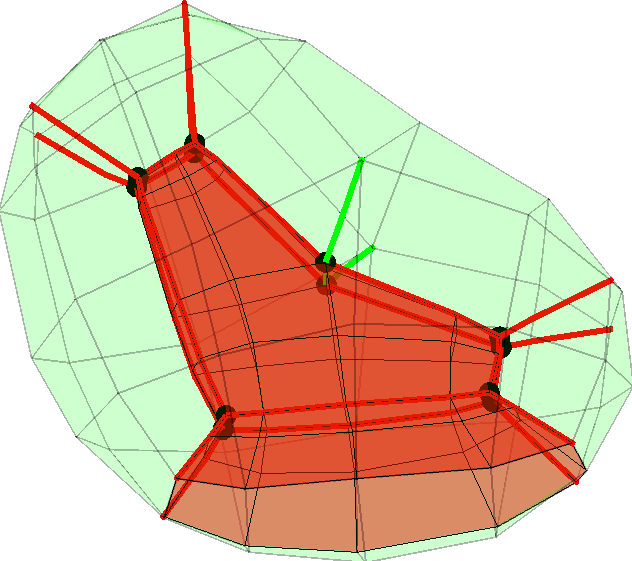}
\includegraphics[width=.135\textwidth]{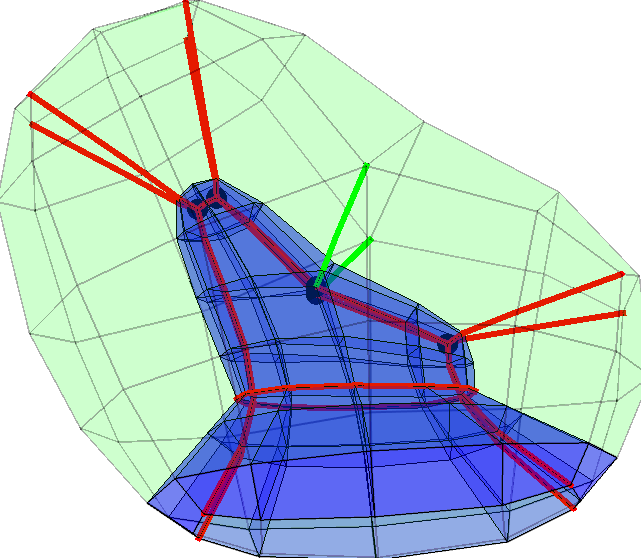}
\includegraphics[width=.135\textwidth]{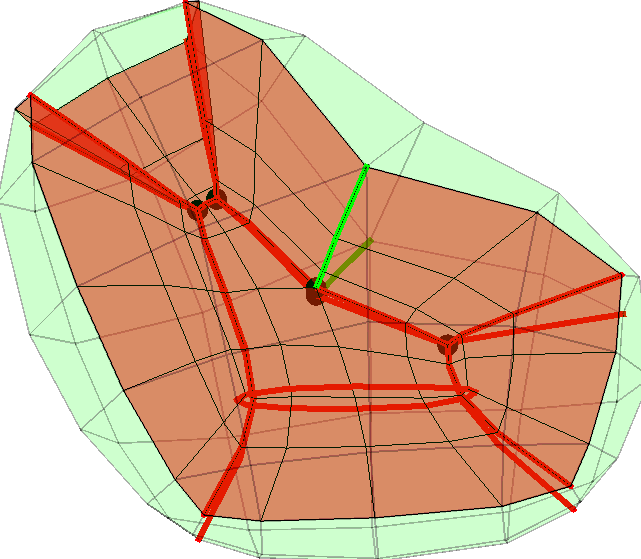}
\includegraphics[width=.135\textwidth]{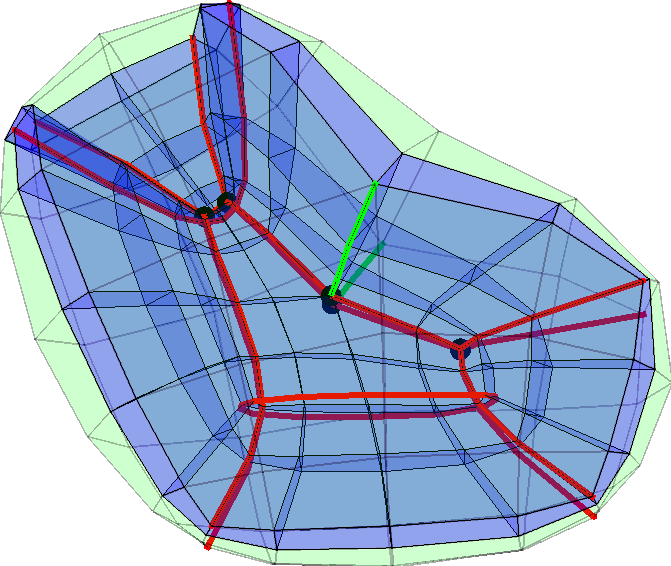}
\includegraphics[width=.135\textwidth]{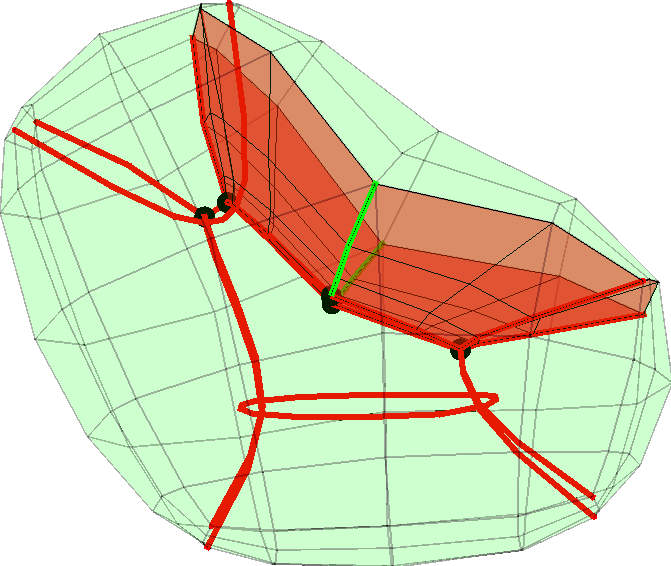}
\includegraphics[width=.135\textwidth]{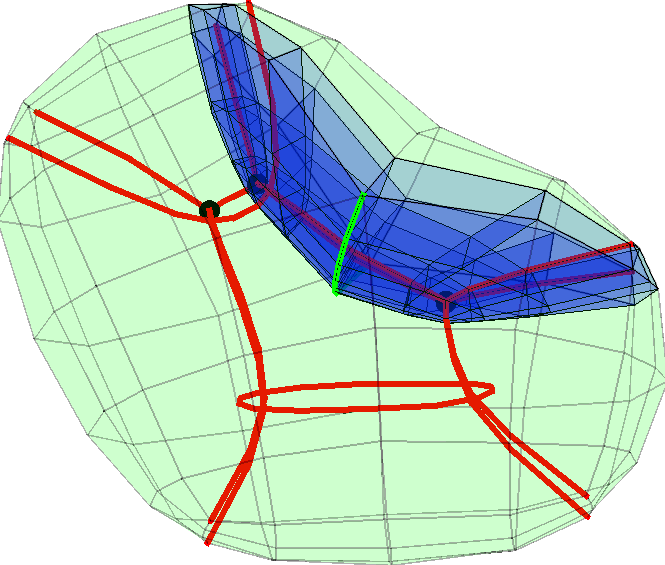}
\\
\includegraphics[width=.135\textwidth]{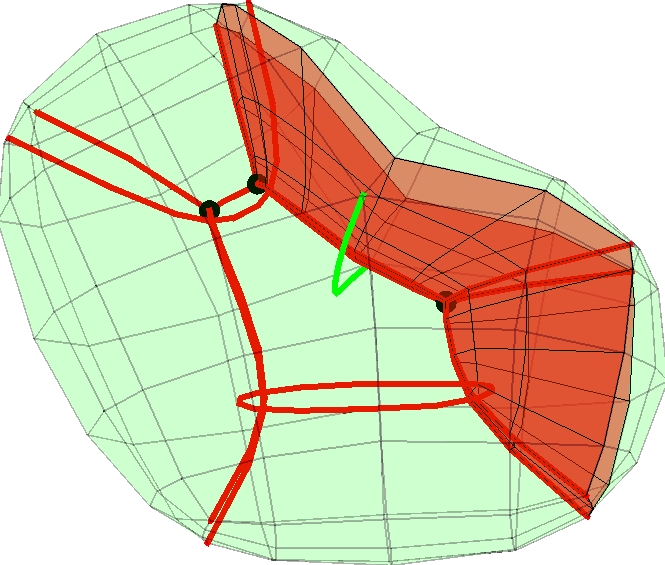}
\includegraphics[width=.135\textwidth]{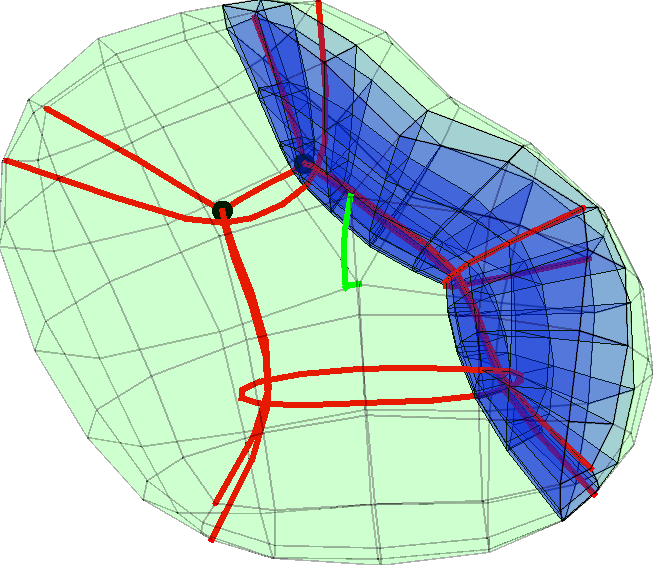}
\includegraphics[width=.135\textwidth]{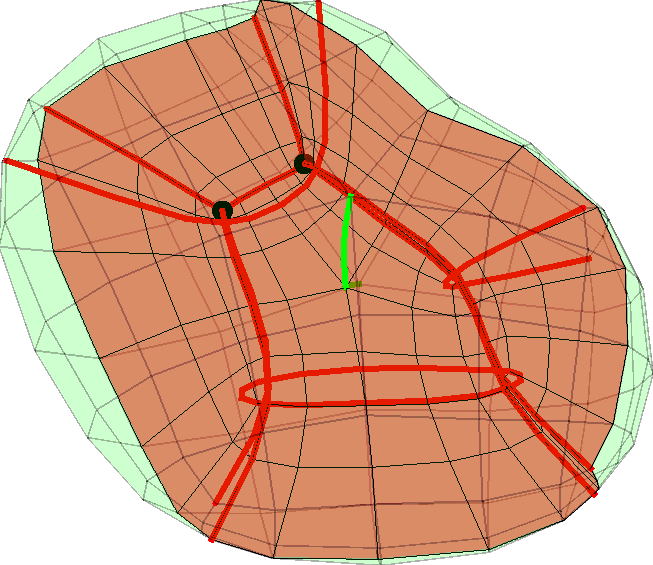}
\includegraphics[width=.135\textwidth]{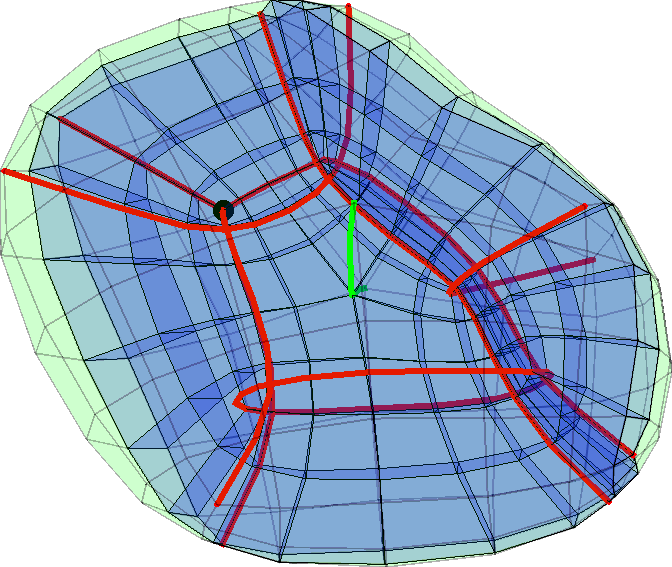}
\includegraphics[width=.135\textwidth]{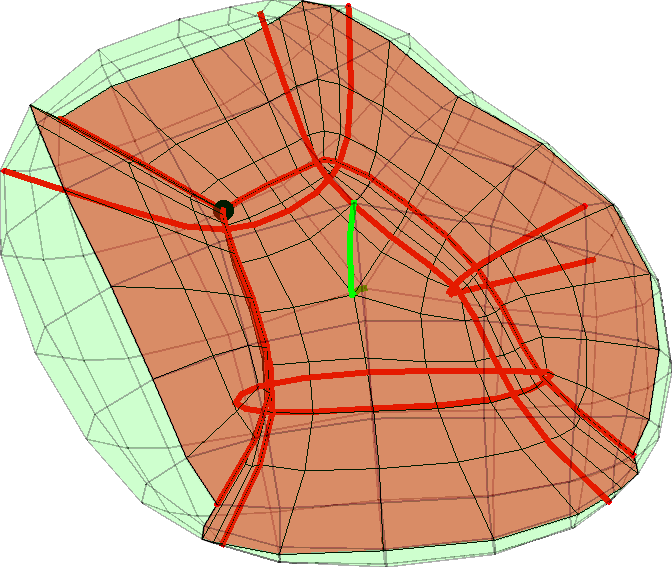}
\includegraphics[width=.135\textwidth]{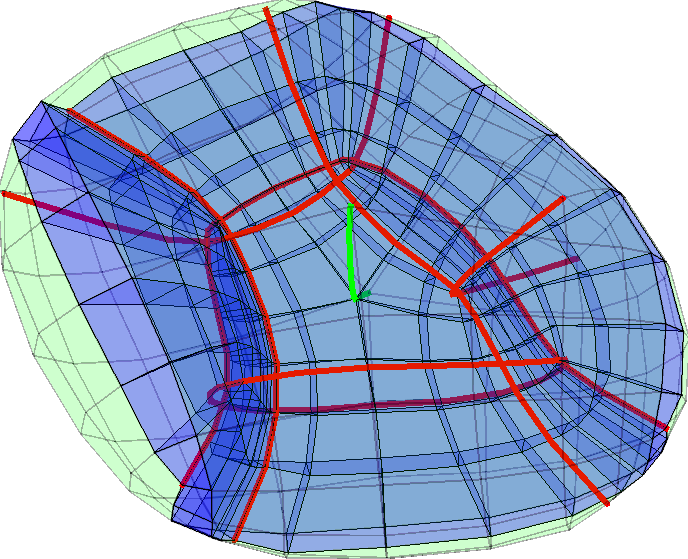}
\includegraphics[width=.135\textwidth]{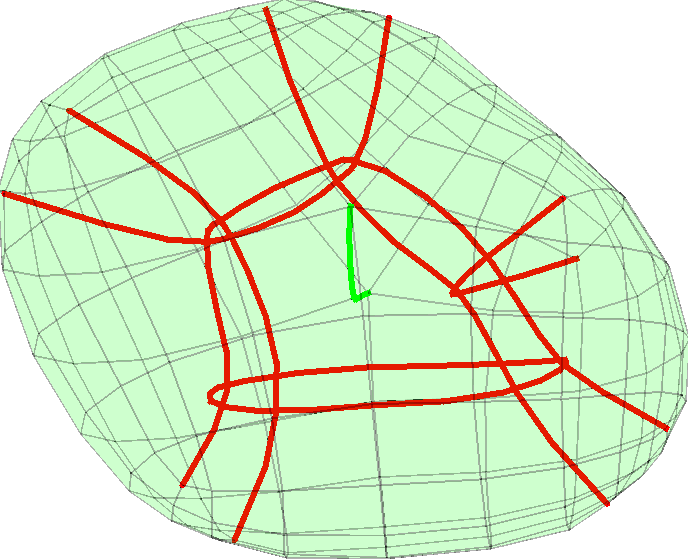}
\\
\hfill $\;$ \hfill
\\
 \includegraphics[width=.135\textwidth]{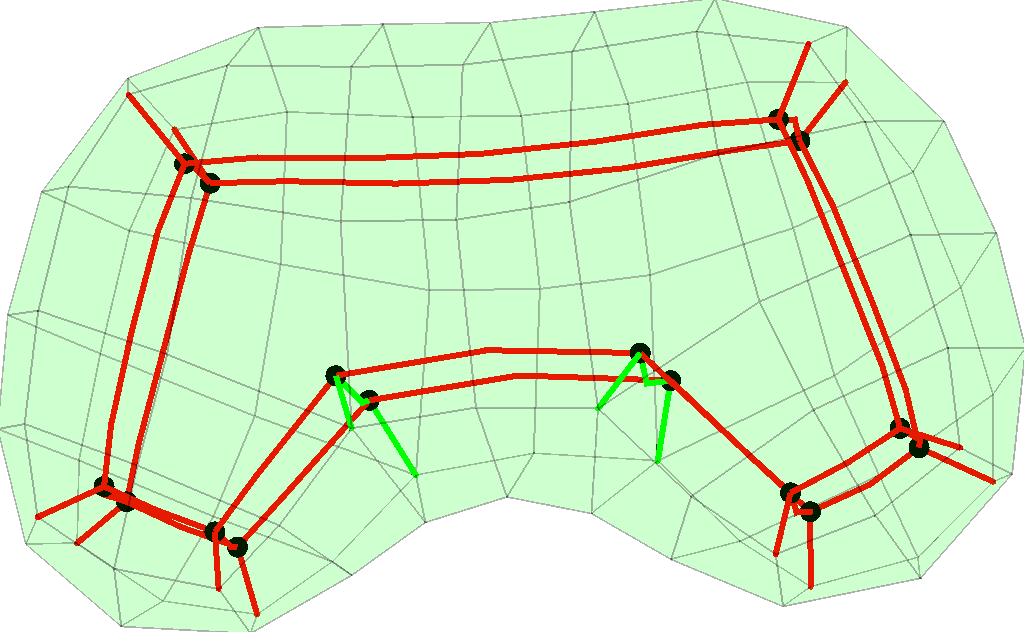}
    \includegraphics[width=.135\textwidth]{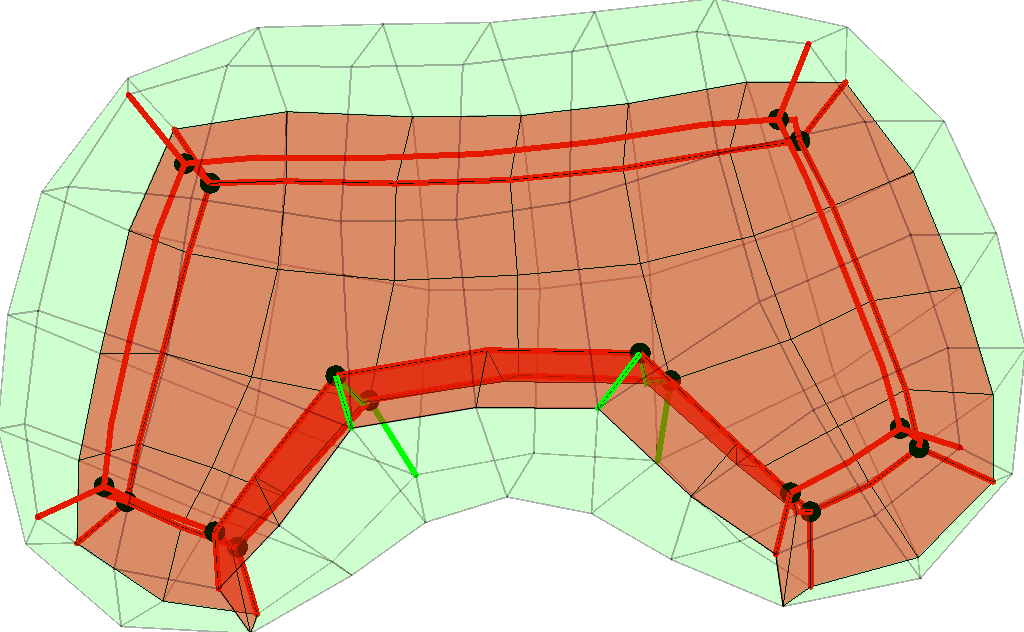}
    \includegraphics[width=.135\textwidth]{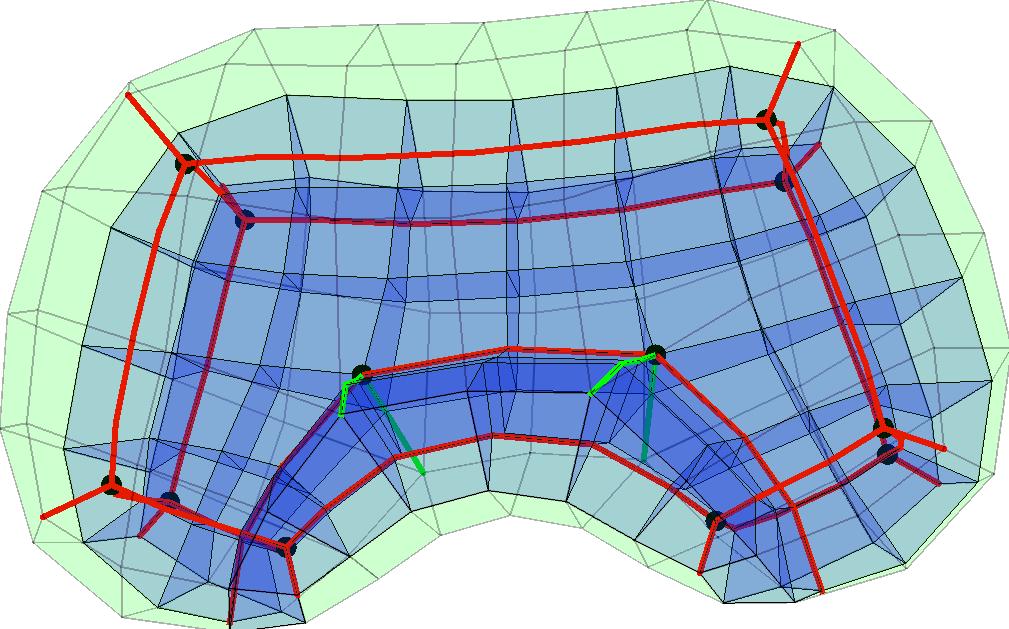}
    \includegraphics[width=.135\textwidth]{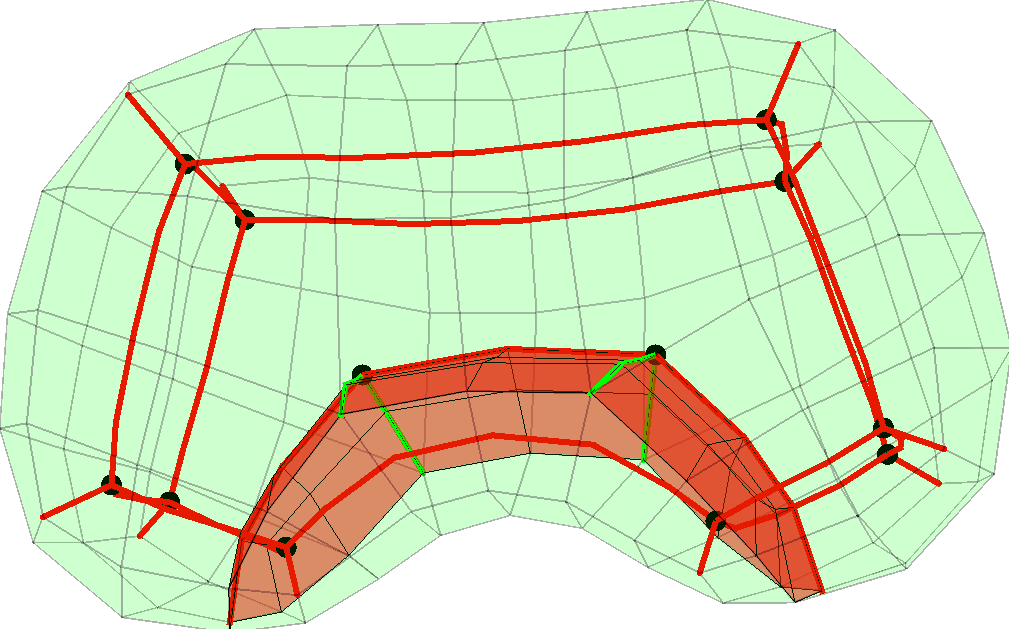}
    \includegraphics[width=.135\textwidth]{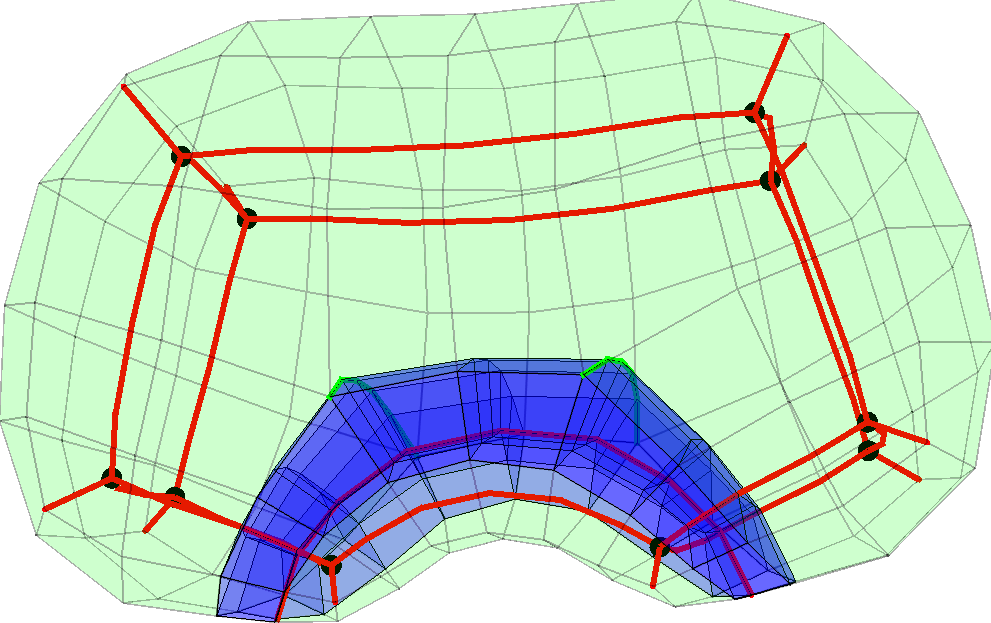}
    \includegraphics[width=.135\textwidth]{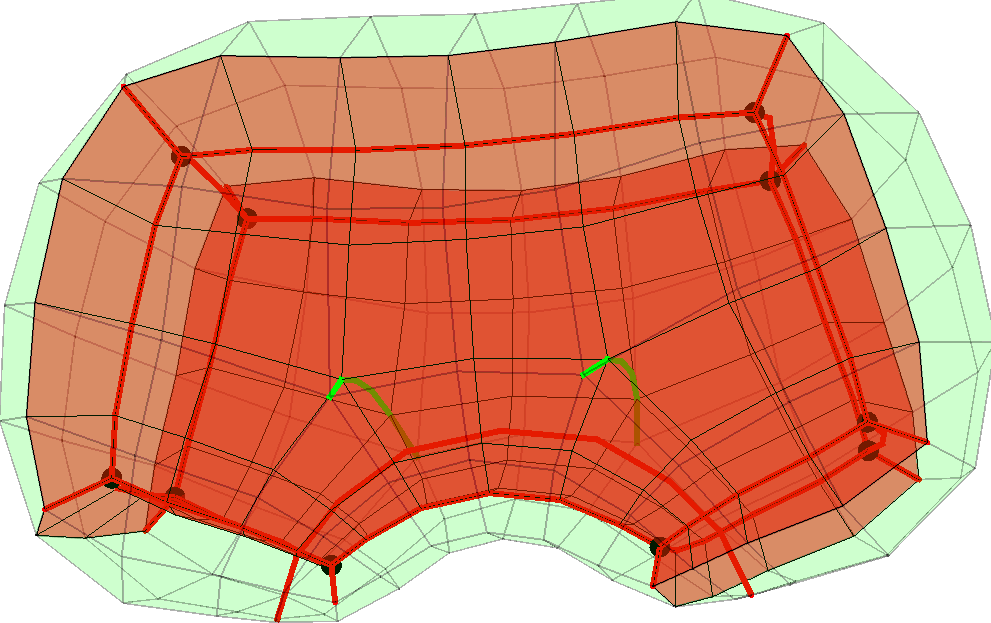}
    \includegraphics[width=.135\textwidth]{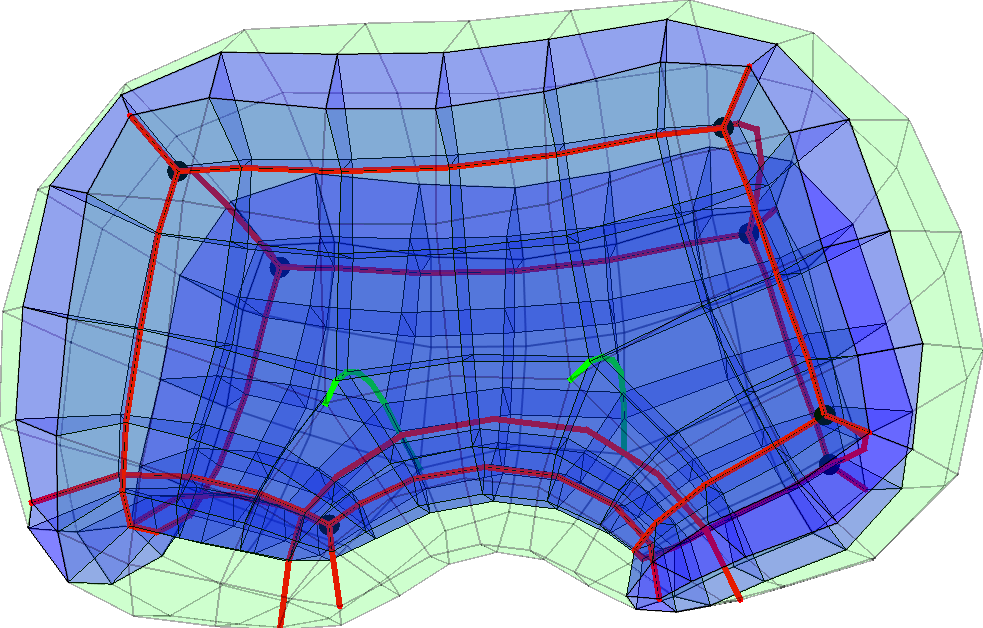}
    \\
    \includegraphics[width=.135\textwidth]{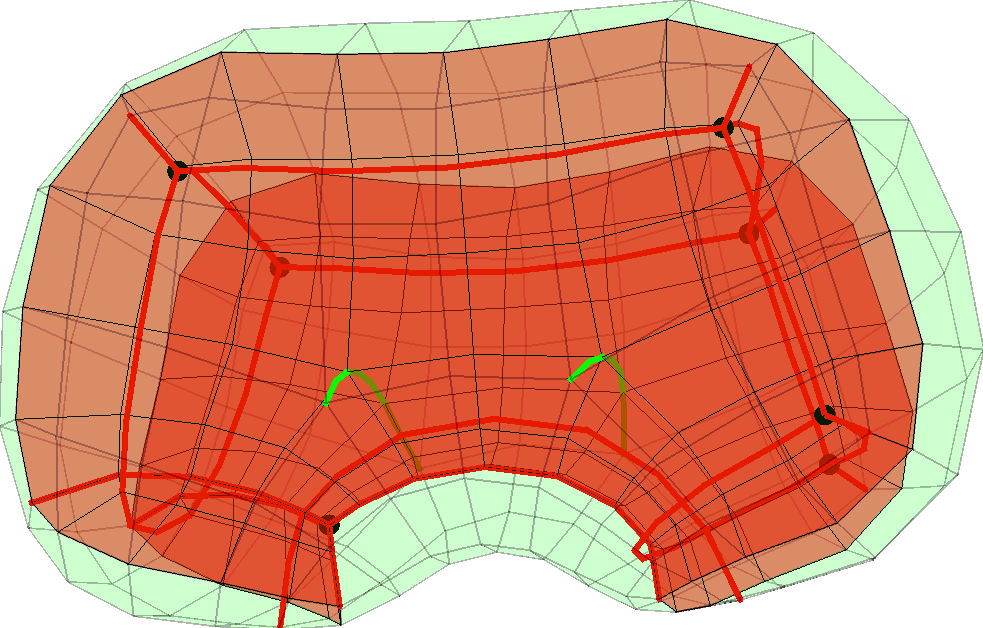}
    \includegraphics[width=.135\textwidth]{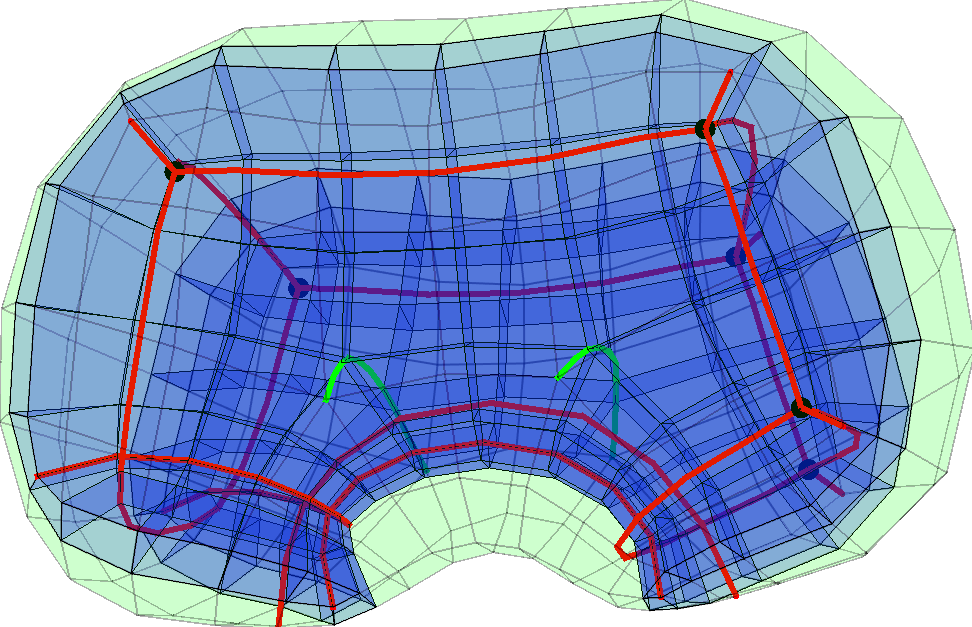}
    \includegraphics[width=.135\textwidth]{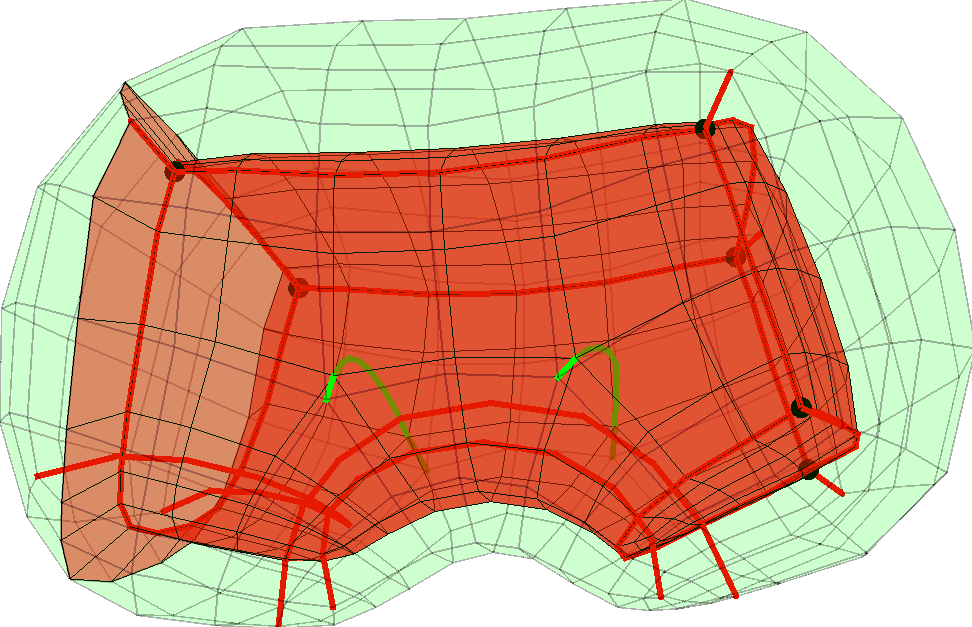}
    \includegraphics[width=.135\textwidth]{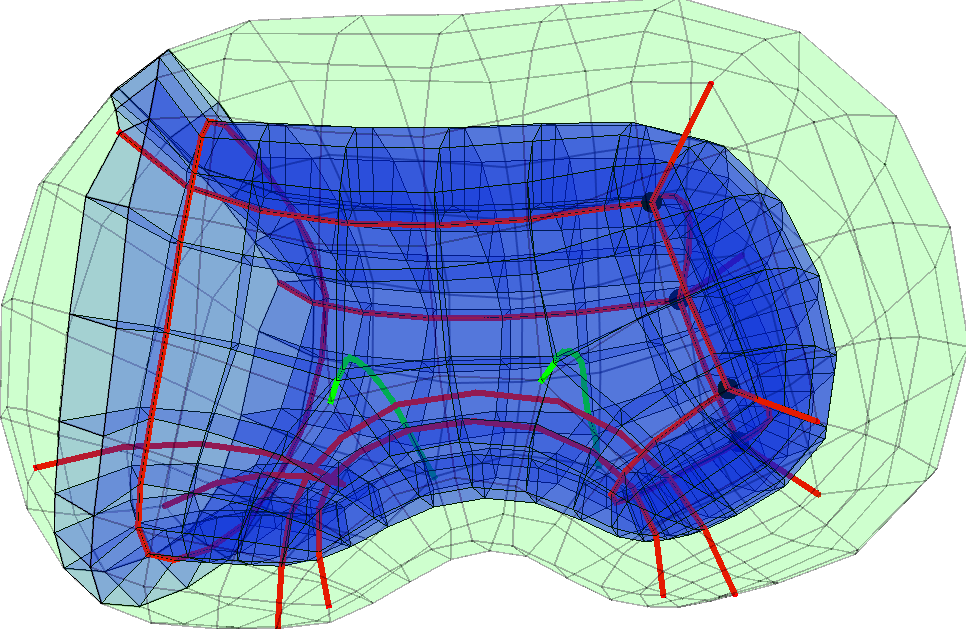}
    \includegraphics[width=.135\textwidth]{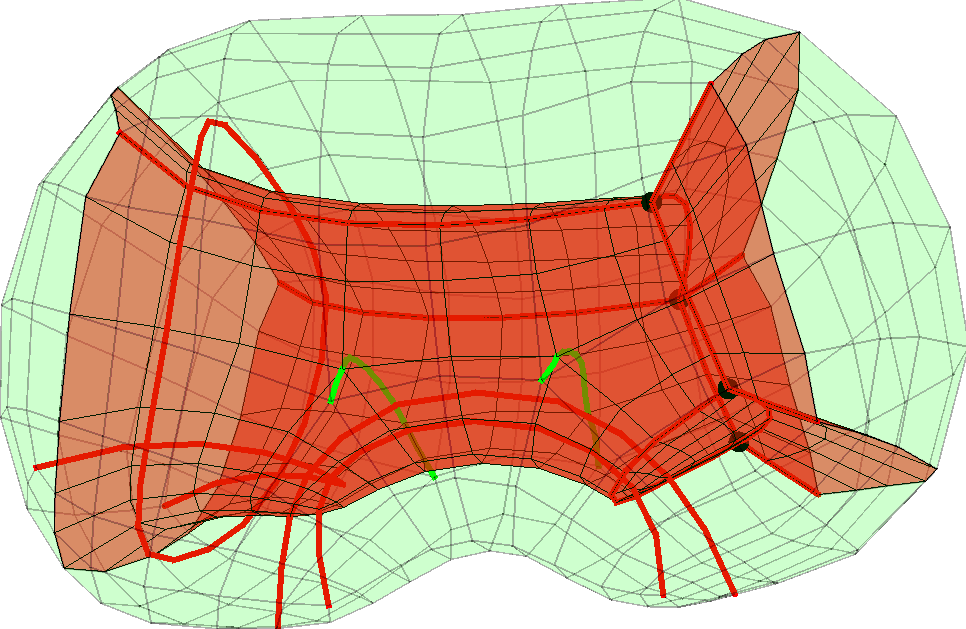}
    \includegraphics[width=.135\textwidth]{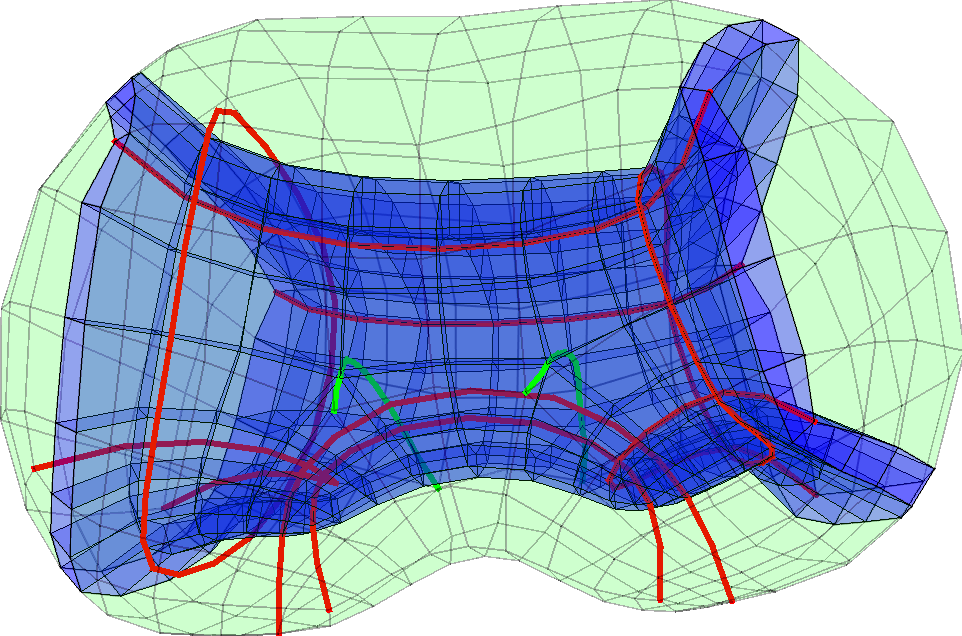}
    \includegraphics[width=.135\textwidth]{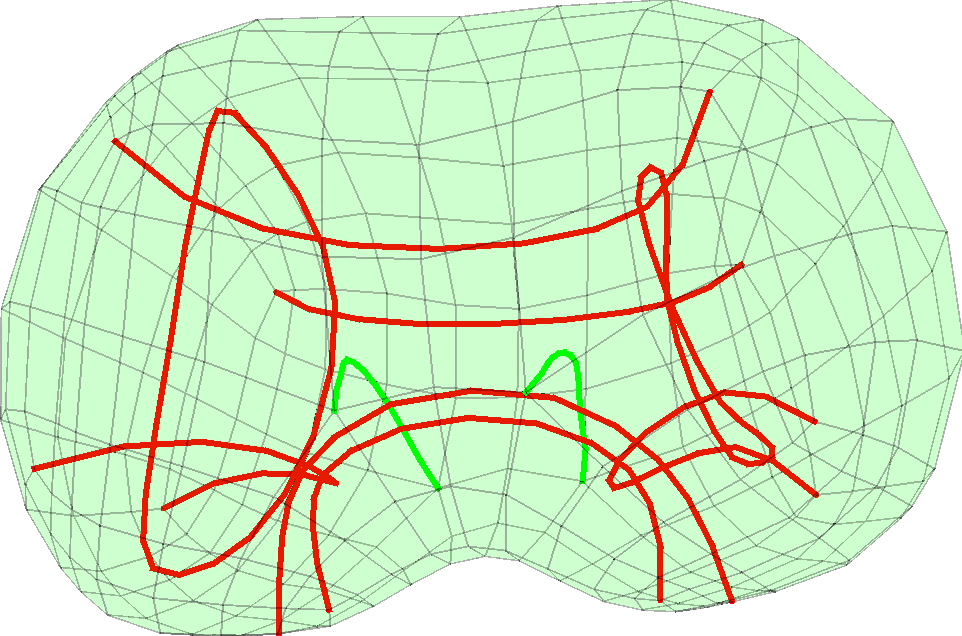}
    \caption{We apply singular decomposition to the \textbf{G1} and \textbf{G2} hex meshes. The first two rows correspond to the sequence of singular graphs from decomposing \textbf{G1}. The last two rows correspond to the sequence of singular graphs from decomposing \textbf{G2}. The number of singular nodes decreases each sheet inflation ultimately resulting in a singular graph with no nodes at all.}
    \label{fig:LCpadded}
\end{figure*}


\begin{figure*}
    \centering
    \includegraphics[width=.115\textwidth]{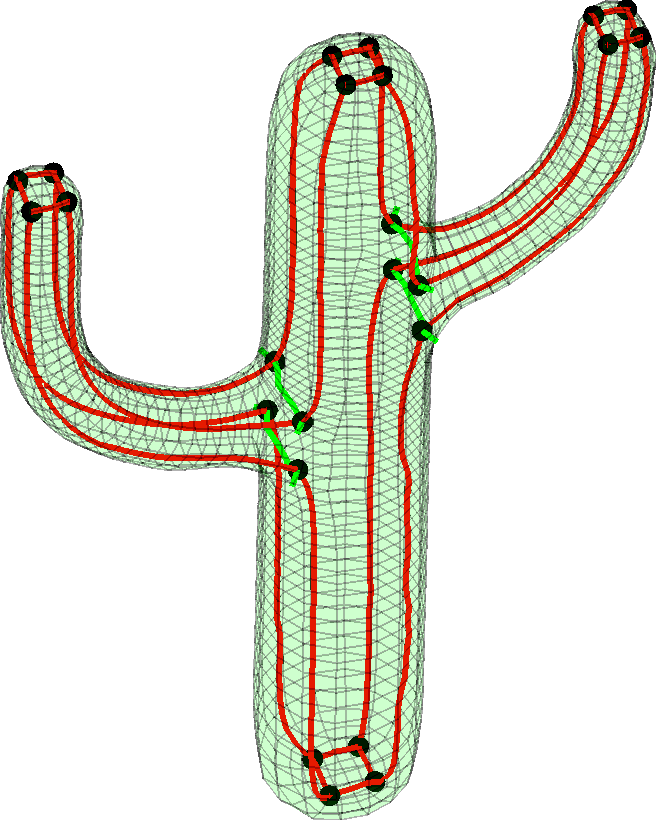}
\includegraphics[width=.115\textwidth]{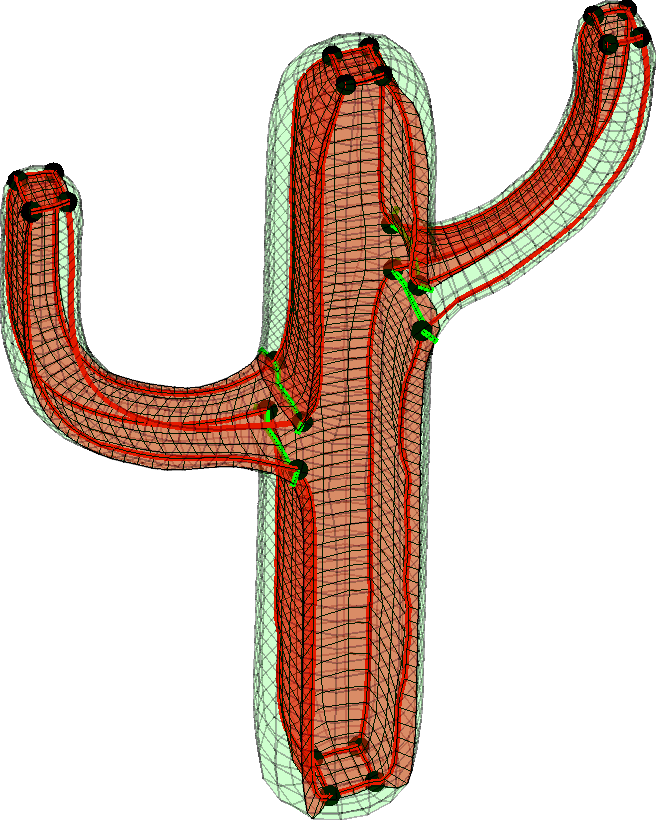}
\includegraphics[width=.115\textwidth]{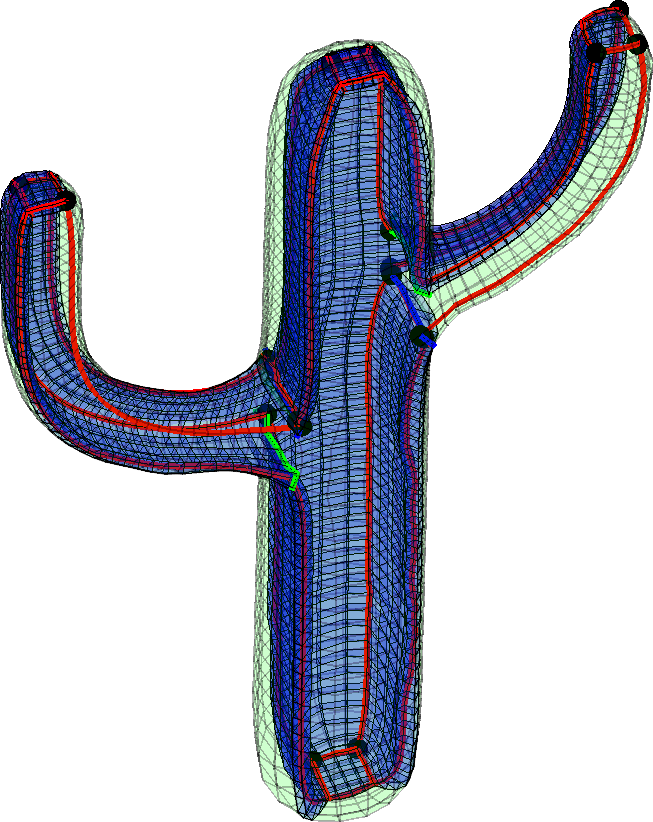}
\includegraphics[width=.115\textwidth]{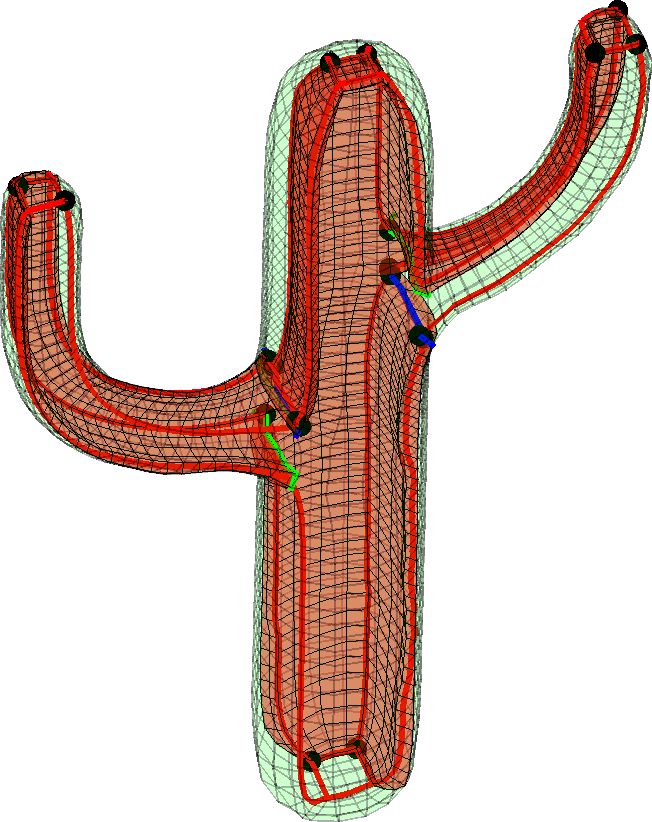}
\includegraphics[width=.115\textwidth]{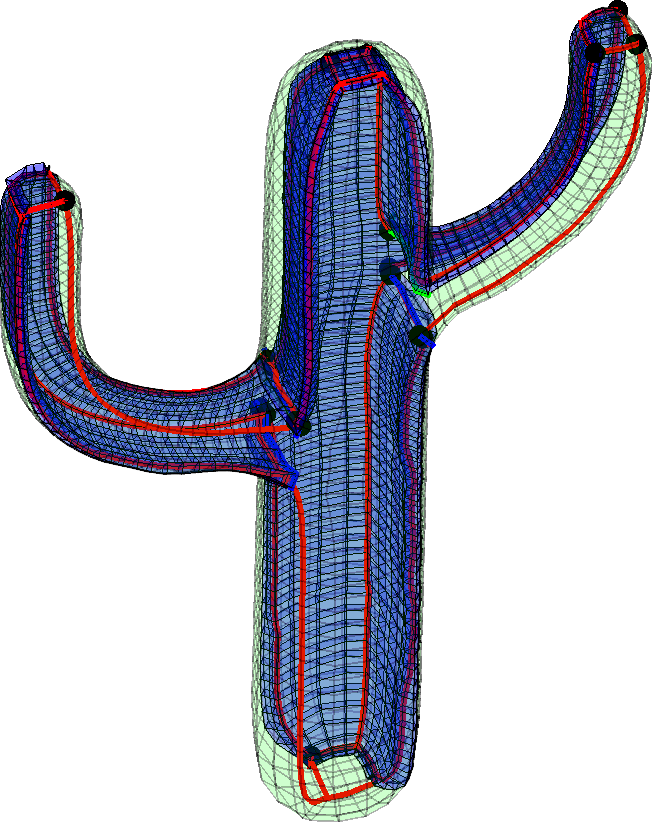}
\includegraphics[width=.115\textwidth]{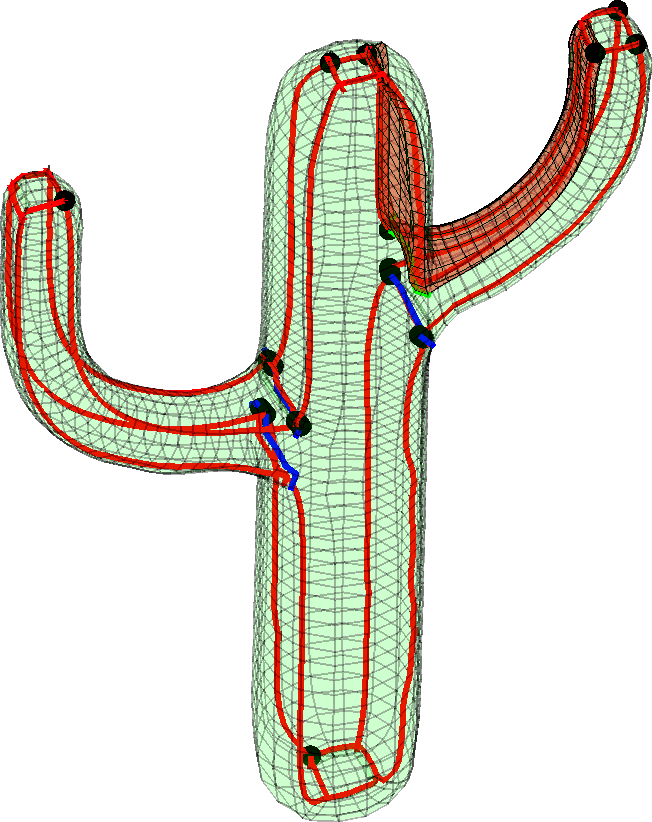}
\includegraphics[width=.115\textwidth]{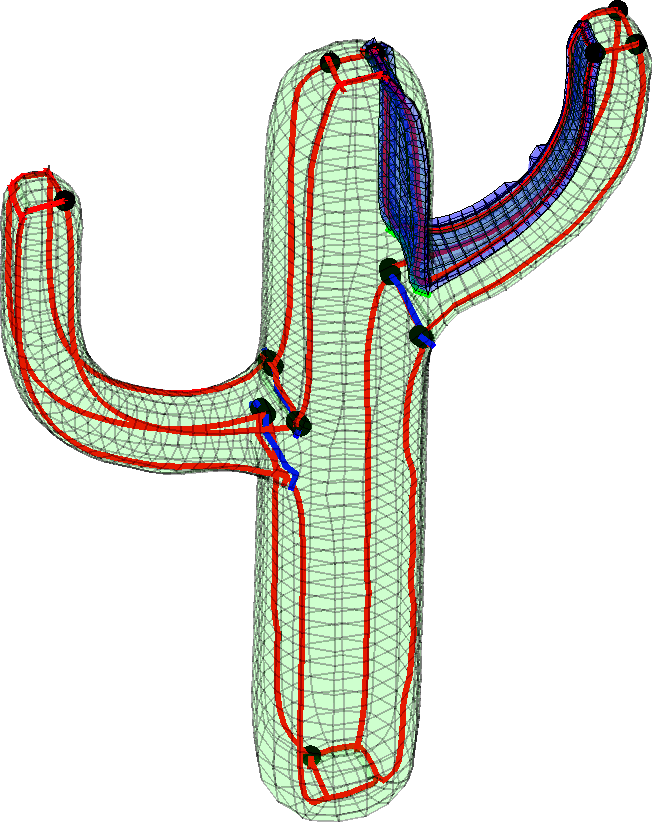}
\includegraphics[width=.115\textwidth]{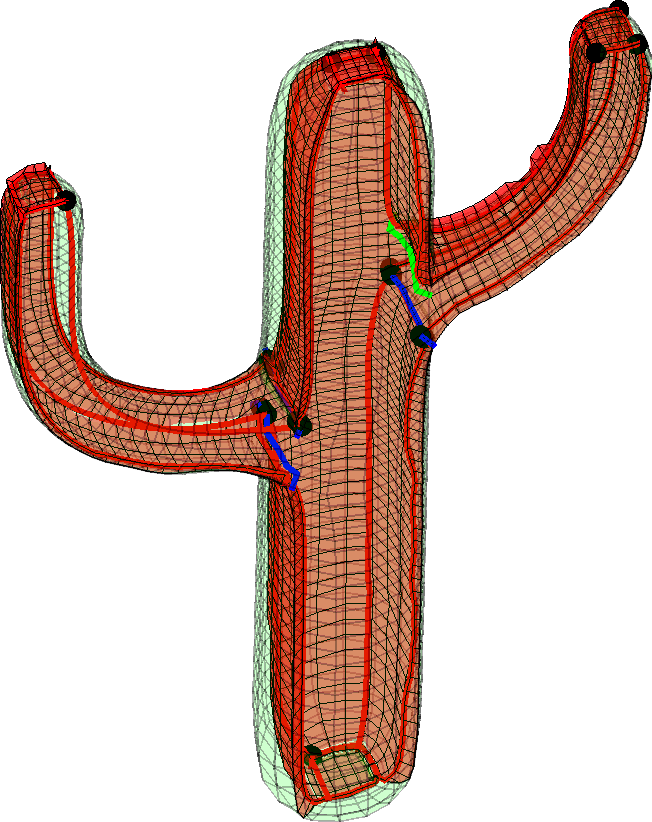}\\
\includegraphics[width=.115\textwidth]{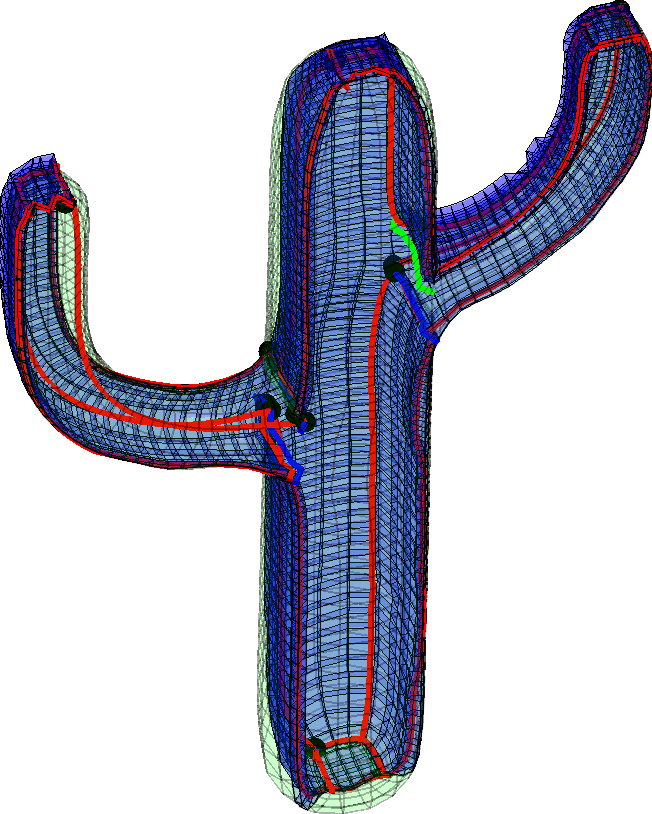}
\includegraphics[width=.115\textwidth]{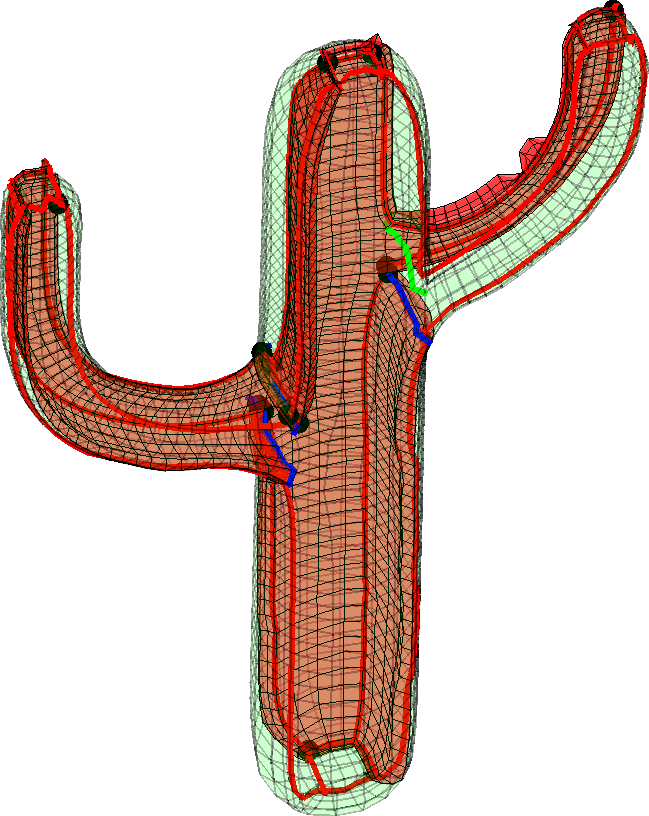}
\includegraphics[width=.115\textwidth]{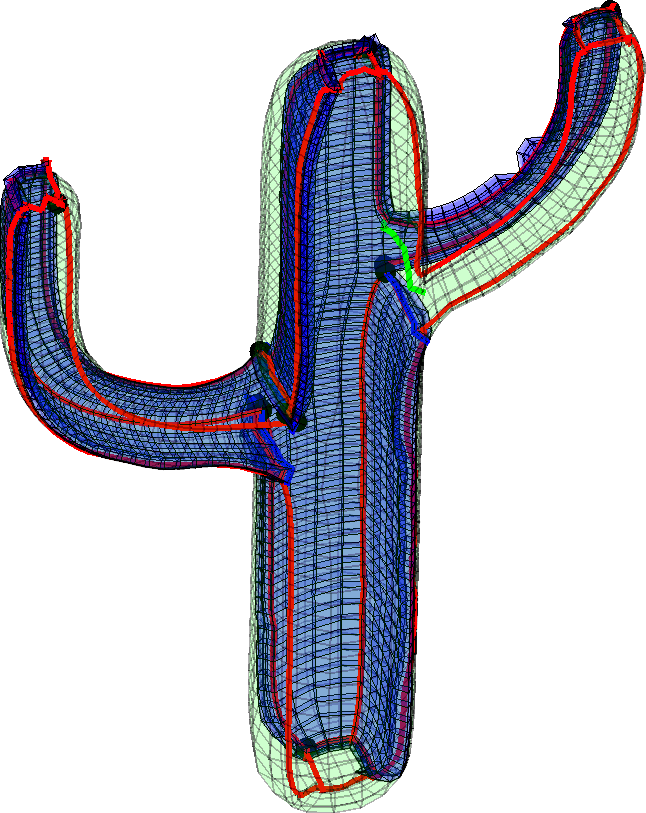}
\includegraphics[width=.115\textwidth]{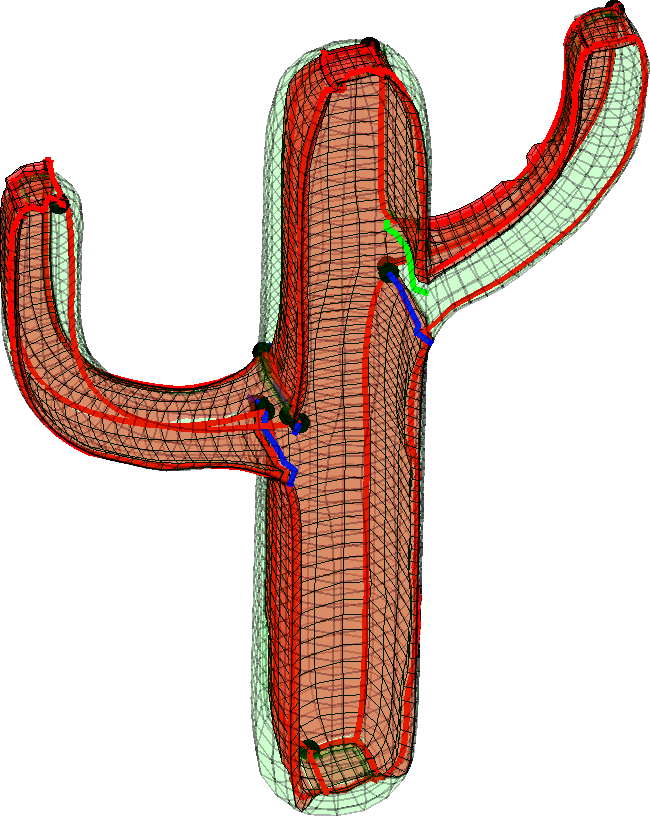}
\includegraphics[width=.115\textwidth]{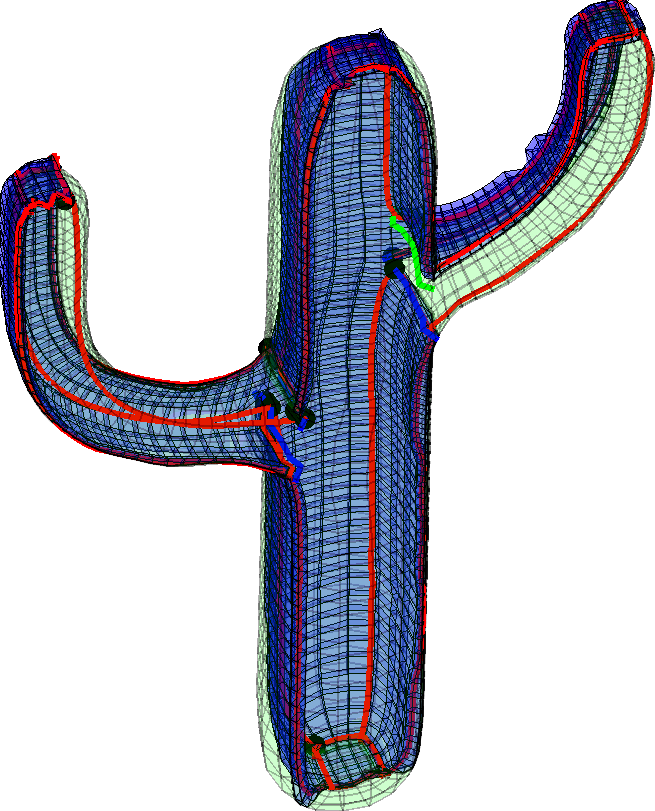}
\includegraphics[width=.115\textwidth]{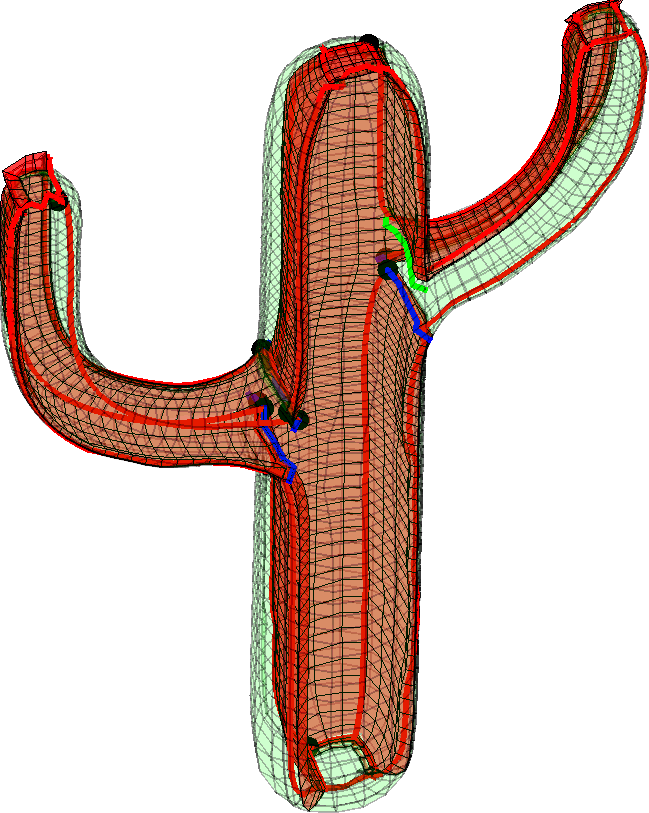}
\includegraphics[width=.115\textwidth]{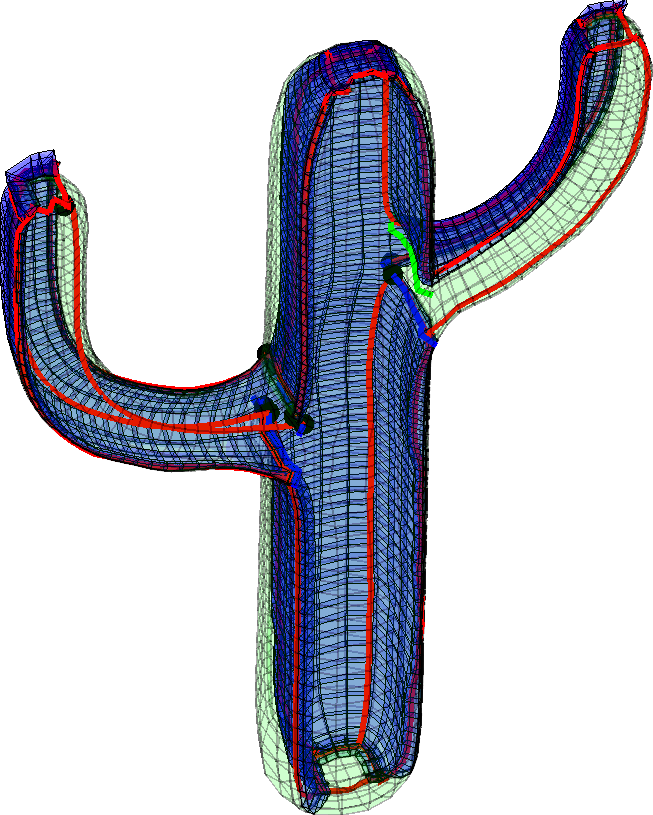}
\includegraphics[width=.115\textwidth]{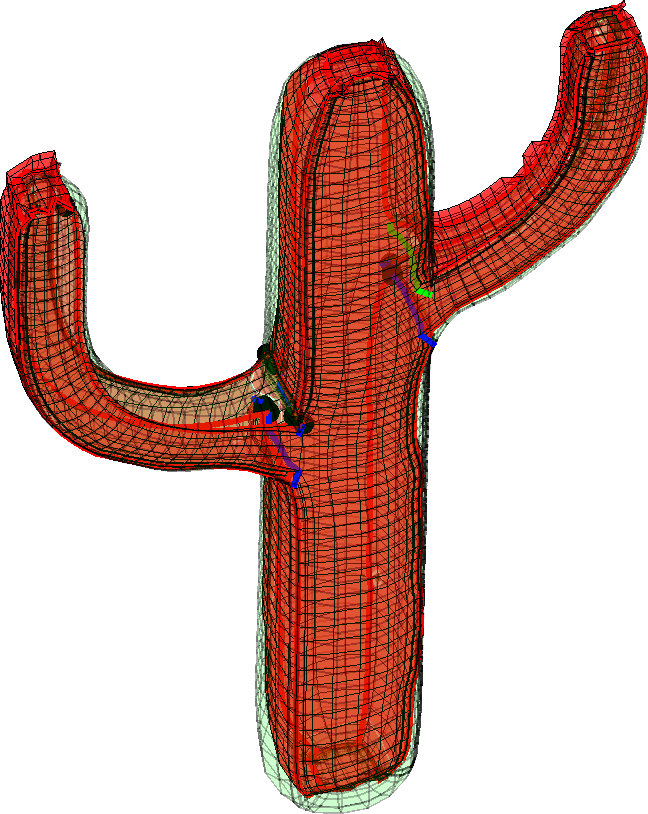}\\
\includegraphics[width=.115\textwidth]{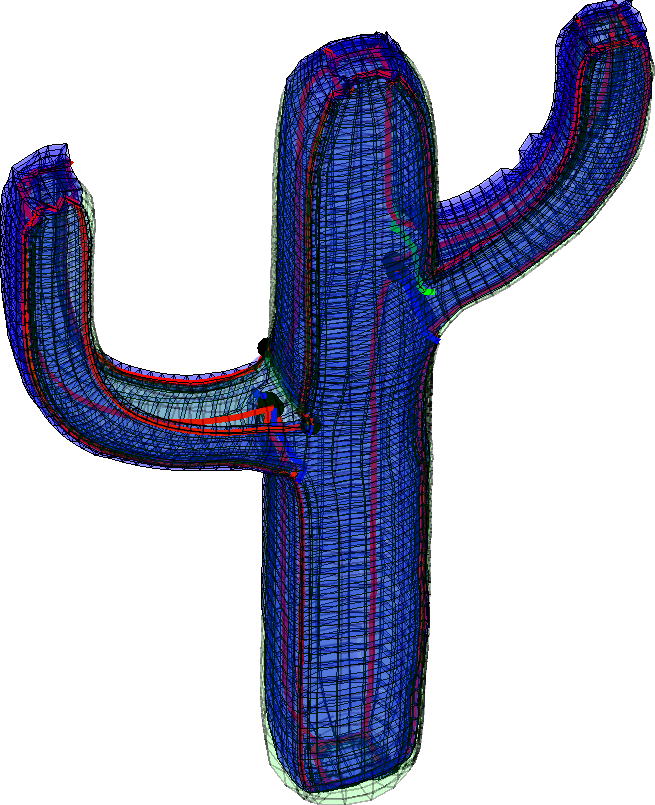}
\includegraphics[width=.115\textwidth]{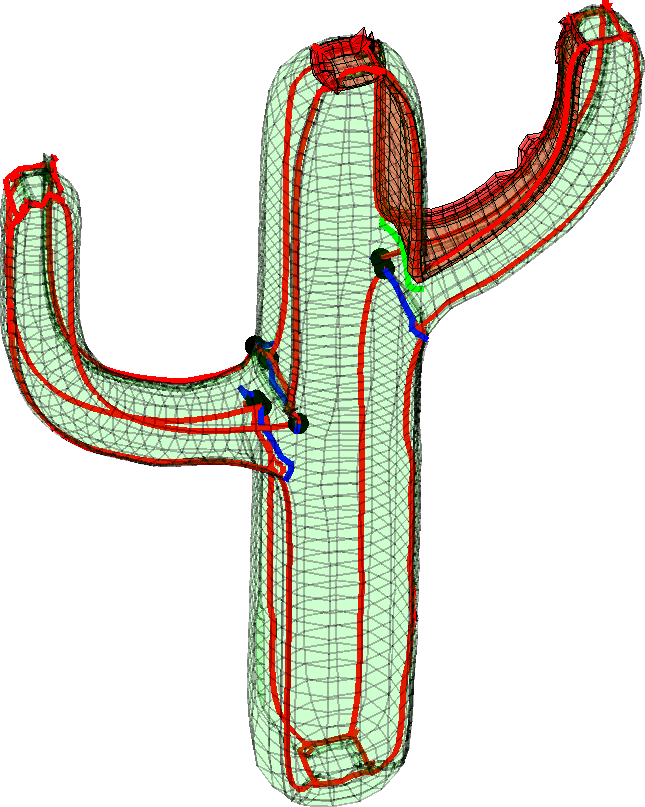}
\includegraphics[width=.115\textwidth]{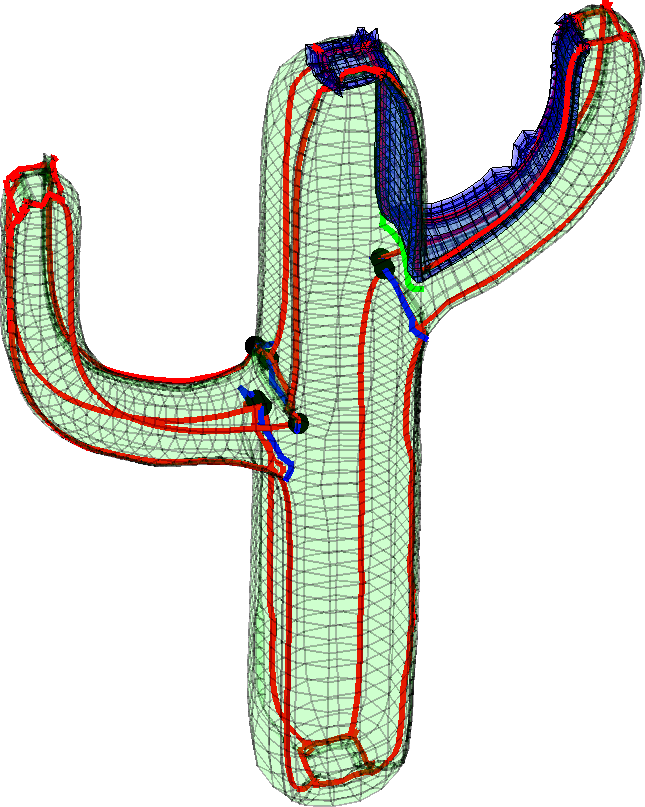}
\includegraphics[width=.115\textwidth]{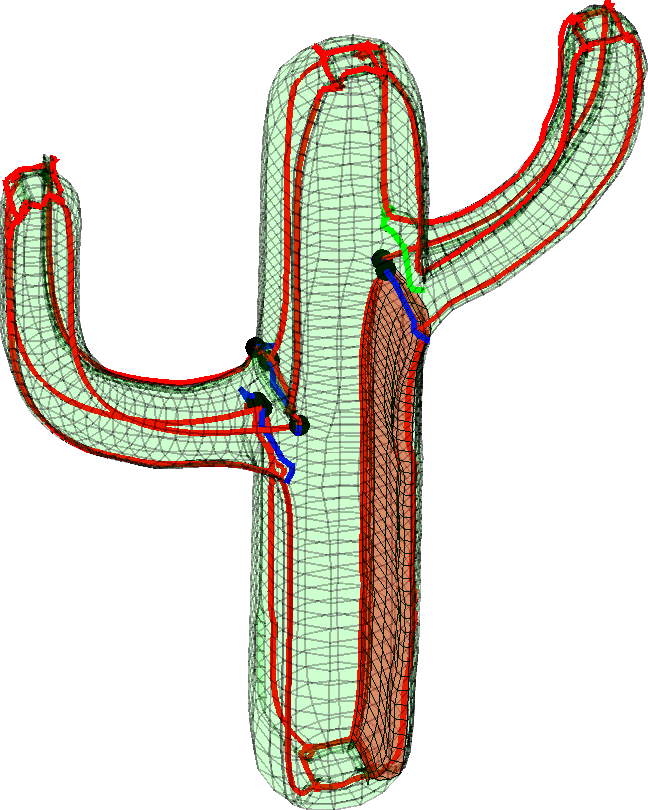}
\includegraphics[width=.115\textwidth]{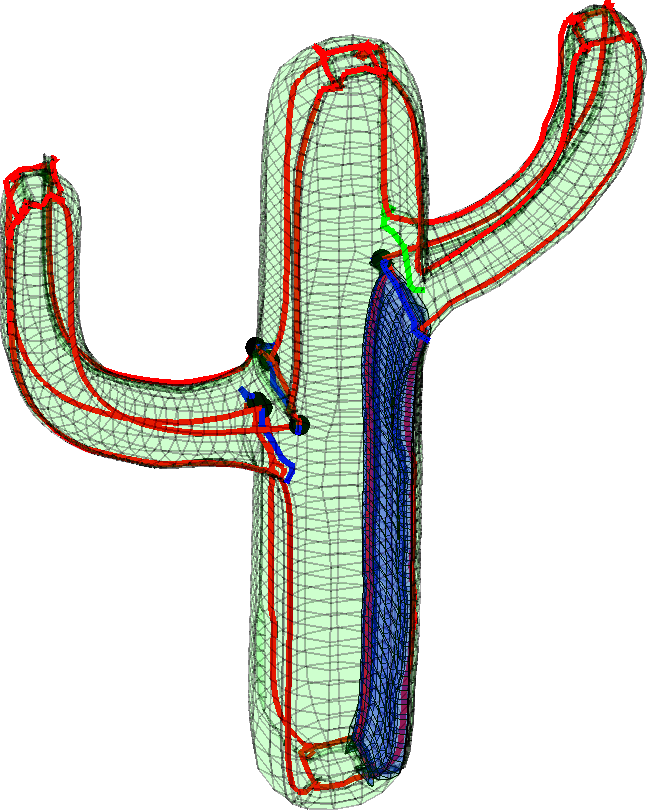}
\includegraphics[width=.115\textwidth]{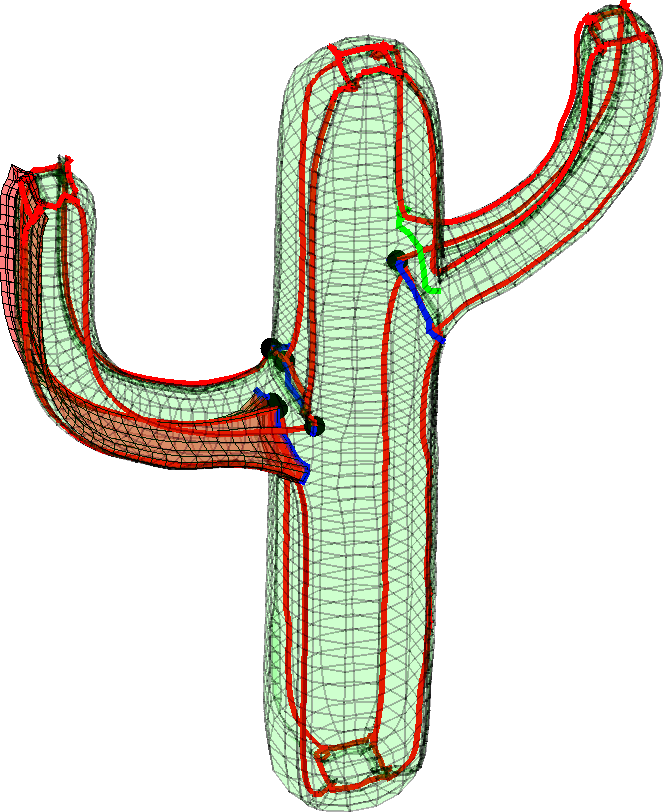}
\includegraphics[width=.115\textwidth]{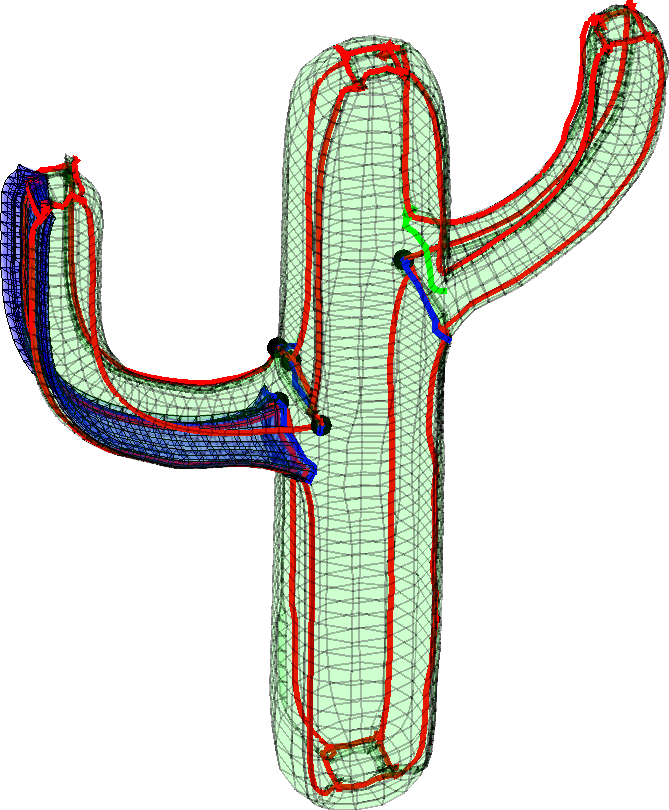}
\includegraphics[width=.115\textwidth]{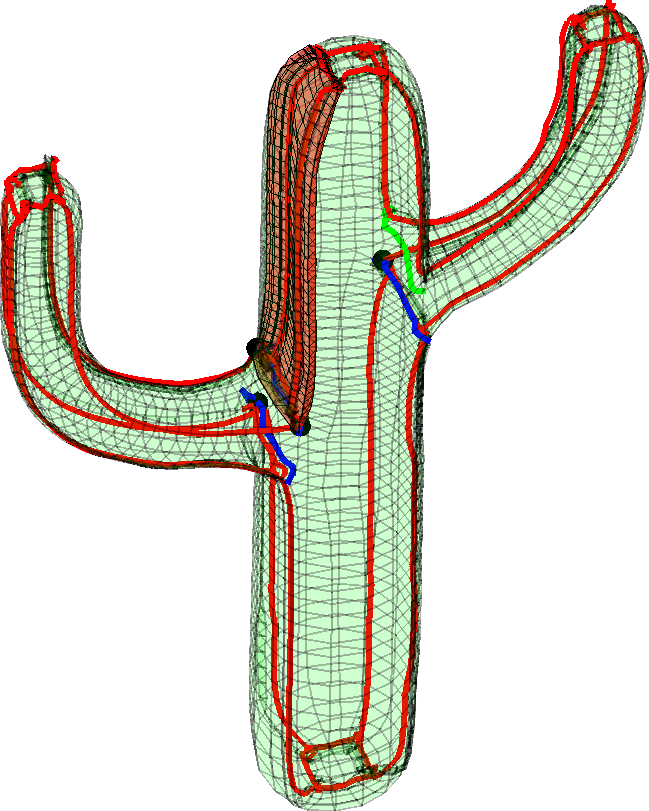}\\
\includegraphics[width=.115\textwidth]{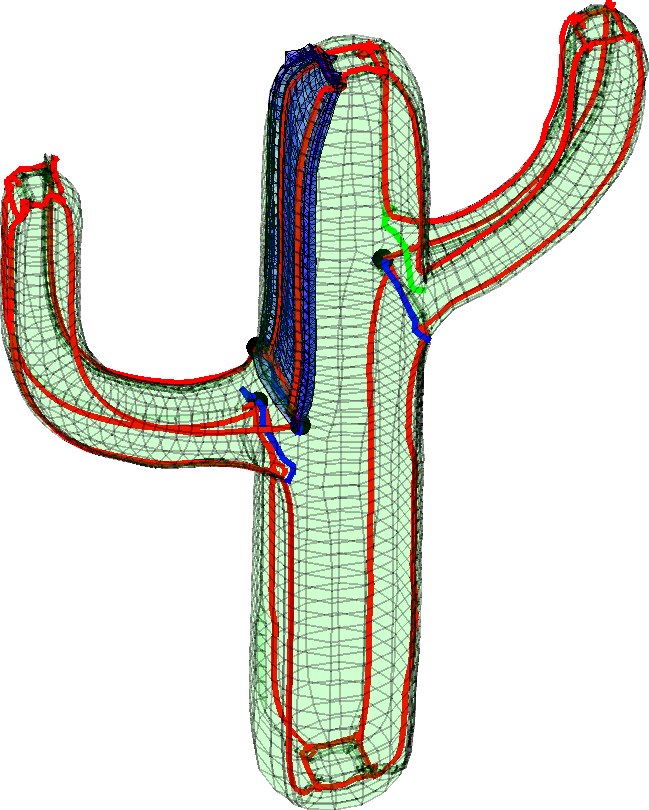}
\includegraphics[width=.115\textwidth]{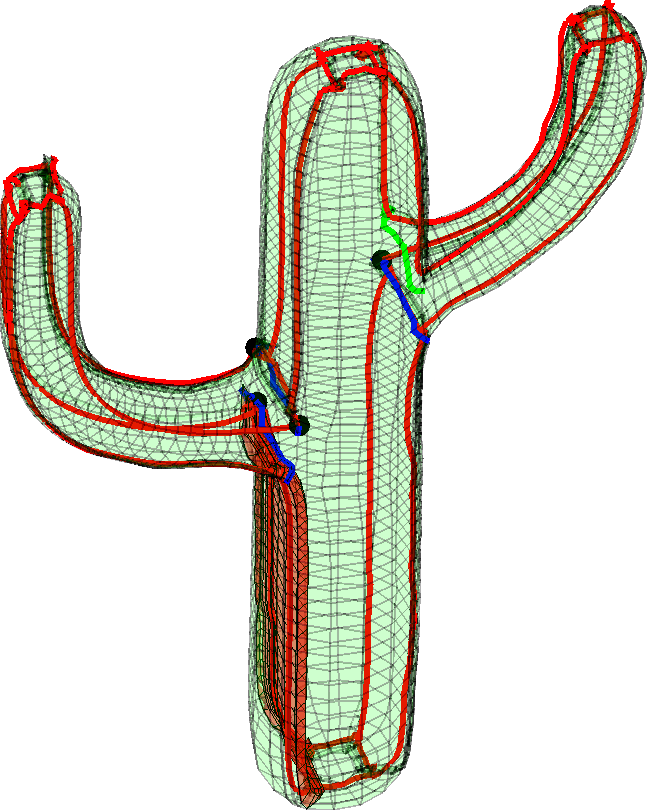}
\includegraphics[width=.115\textwidth]{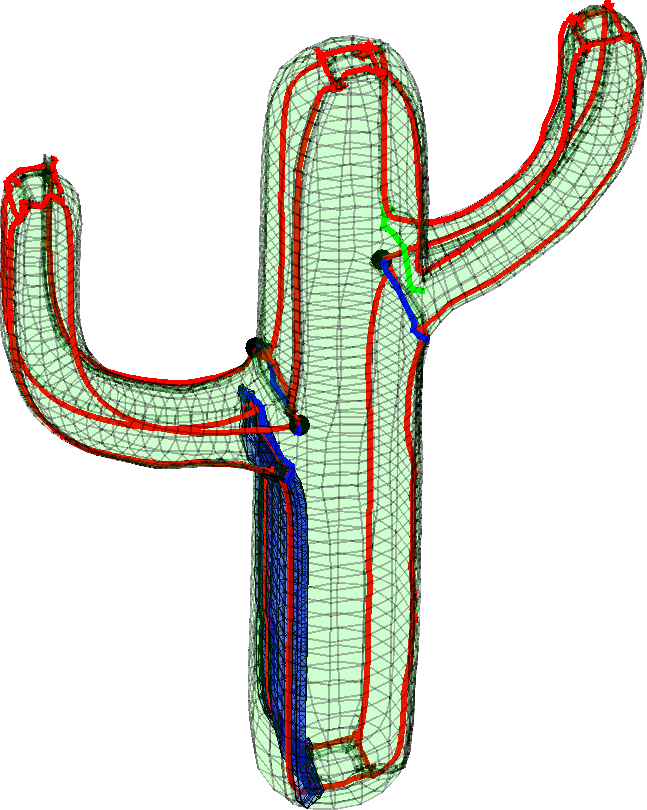}
\includegraphics[width=.115\textwidth]{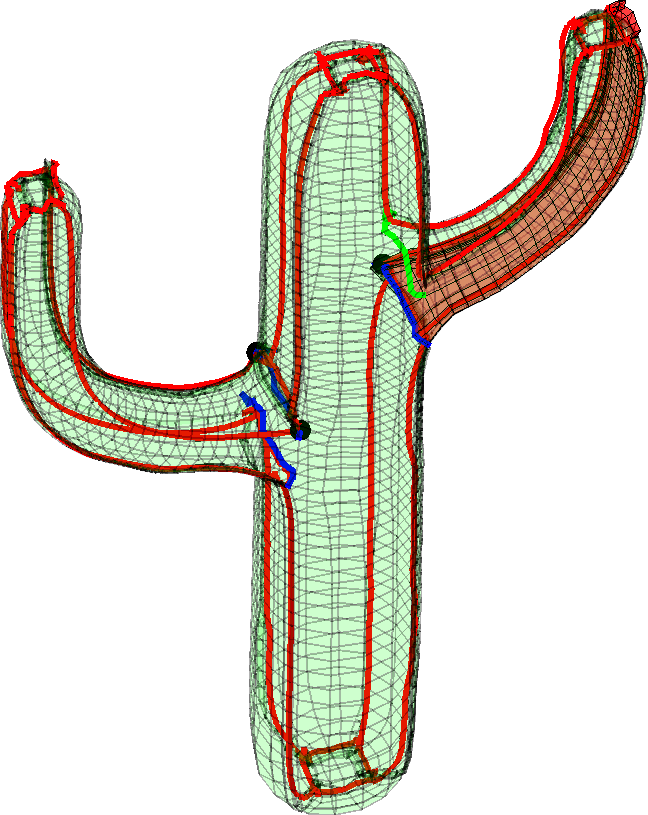}
\includegraphics[width=.115\textwidth]{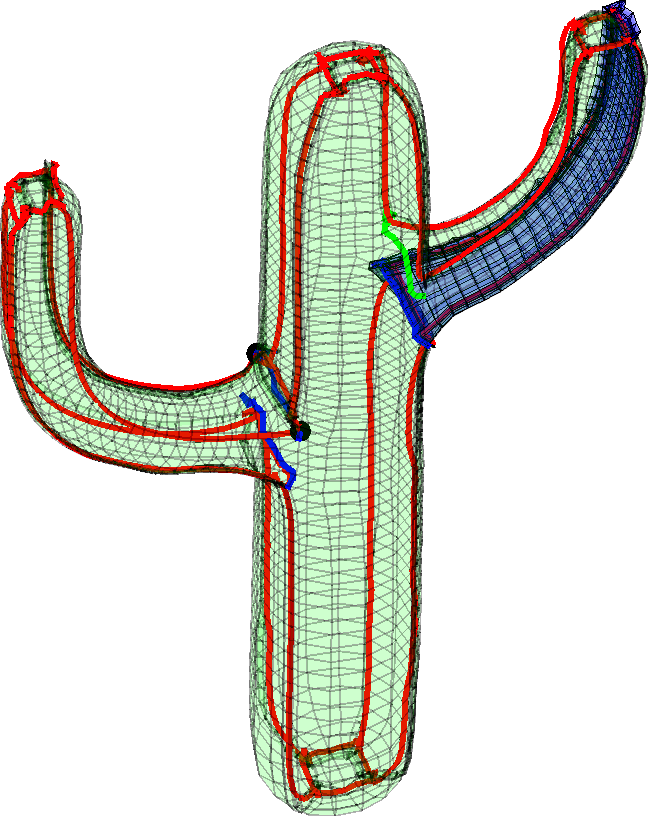}
\includegraphics[width=.115\textwidth]{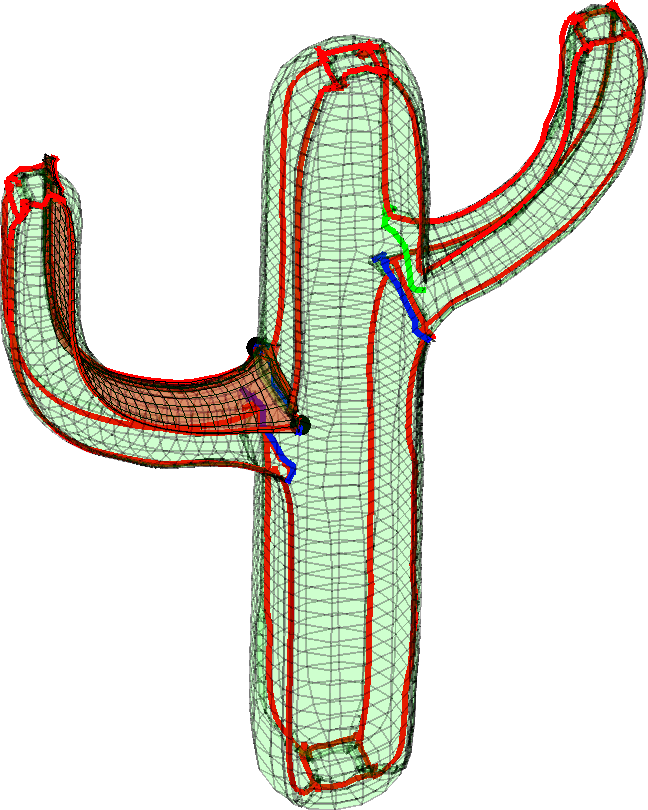}
\includegraphics[width=.115\textwidth]{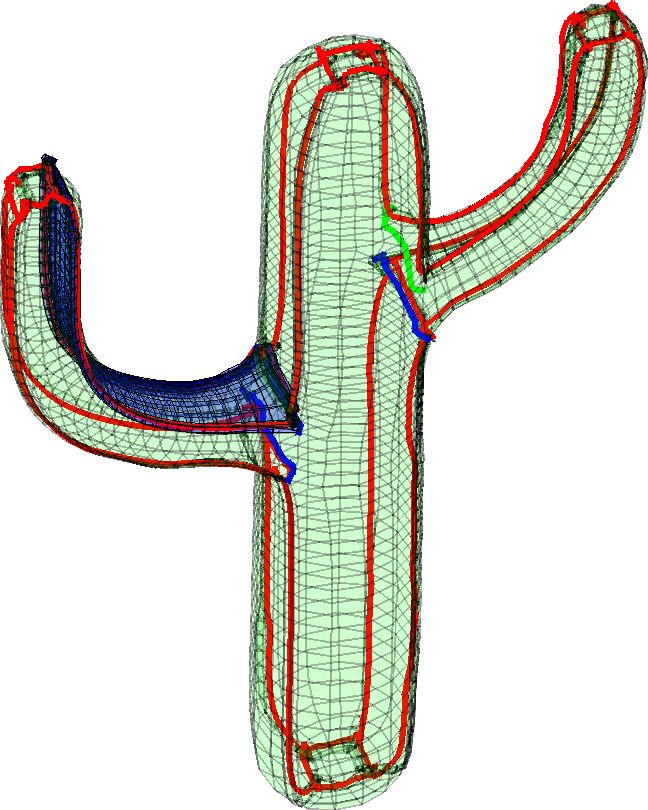}
\includegraphics[width=.115\textwidth]{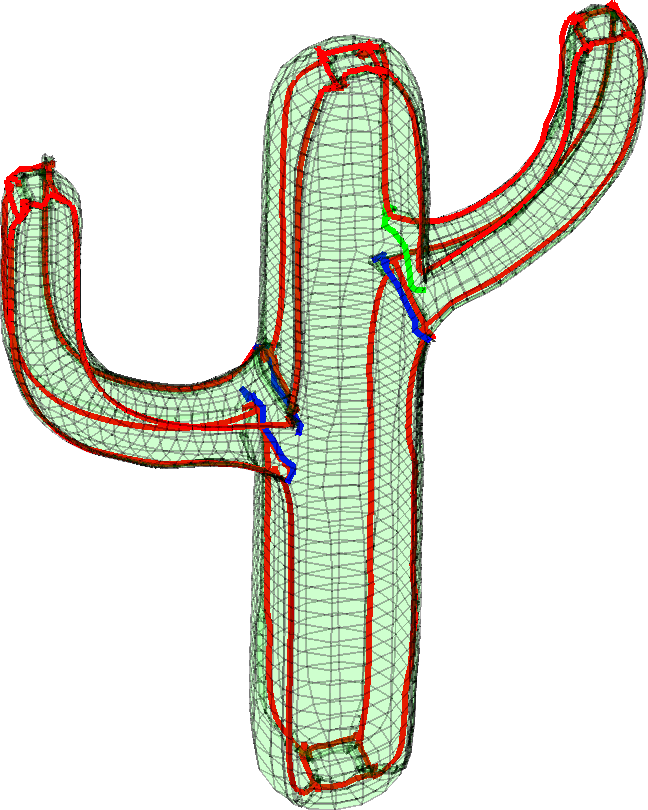}
    \caption{We apply singular decomposition to the cactus mesh from \cite{hexalab}. The number of singular nodes decreases with each sheet inflation ultimately resulting in a singular graph with no nodes at all. While the original singular graph consisted of only valence 3, 4 and 5 nodes, intermediate singular graphs from this sequence contain singular nodes with signature (2,3,0,2). The fully decomposed singular graph has valence 6 curves.}
    \label{fig:cactus}
\end{figure*}

Finally, we apply our decomposition in \autoref{fig:cactus} to the cactus mesh from \cite{hexalab}. While the mesh starts with only singularities of valence 3, 4  and 5, the decomposition results in intermediate singular graphs with nodes of signature (2,3,0,2). The final configuration is seen to contain singular curves of valence 6. Since our goal is only to remove singular nodes, we terminate with valence 6 curves. 

\subsection{Scaled Jacobians}
The minimum scaled Jacobian of a hex mesh is a common metric by which to evaluate distortion of the mesh \cite{cubit}. We maximize the minimum scaled Jacobian before and after singular decomposition of each singular node with free boundaries and present the resulting minimum scaled Jacobians in \autoref{tab:tab1}. Unsurprisingly, singular nodes have lower scaled Jacobians than singular curves.

By symmetry, the minimum scaled Jacobian of any hex mesh, regardless of resolution, containing a (4,0,0) node is upper bounded by $\frac{4}{3\sqrt{3}}=.7698$. 
The same bound for a hex mesh containing a (0,0,12) node is
$\frac{\sqrt{2 (5+\sqrt{5})}}{5} = .761$.
These bounds are exactly attained in \autoref{tab:tab1} for the (4,0,0) and (0,0,12) nodes. 
The same upper bound computed for meshes containing valence 3 singular curves is  $\sin(\frac{2\pi}{3})=.866$ and for meshes containing valence 5 singular curves is $\sin(\frac{2\pi}{3})=.951$. 

By performing a sheet inflation to split singular nodes into singular curves, the minimum scaled Jacobian of the (4,0,0) node is increased to $.86$, almost the theoretic upper bound. For (0,4,4) as well, decomposing the singular node into two valence 5 curves brings the minimum scaled Jacobian to almost the theoretic upper bound. Decomposing (0,0,12) node into six valence 5 curves brings significant improvement to the minimum scaled Jacobian, though it is not as close to the theoretic upper bound due to interactions between singular curves.

Moving towards full singular graphs, we perform the same scaled Jacobian optimization for a sphere mesh.
Maximization of its minimum scaled Jacobian results in a value of .768, close to the upper bound for any mesh containing a (4,0,0) node. We apply our decomposition to this mesh and re-optimize its scaled Jacobian resulting in a significant improvement to .849. We run the same optimizations on the padded tetrahedron, \textbf{G1}, and \textbf{G2} resulting in similar increases in the minimum scaled Jacobian. These results are summarized in \autoref{tab:tab1}.


\begin{table}[]
    \centering
    \begin{tabular}{c|c|c|c}
    Mesh&Original&Decomposed&UpperBound\\
         \hline
(4,0,0)     &            0.769      &   0.86    &   0.866 \\
(2,2,2)     &           0.807       & 0.862     &  0.866 \\
(0,4,4)     &           0.896       & 0.943     &  0.951 \\
(1,3,3)     &           0.822       & 0.865     &   0.866 \\
(0,3,6)     &           0.863       & 0.939     &   0.951 \\
(0,2,8)     &            0.82       & 0.937     &   0.951 \\
(2,0,6)     &           0.745       & 0.856     &   0.866 \\
(0,0,12)    &            0.761      &  0.926    &    0.951 \\
Sphere      &         0.768        &  0.849     &   0.866 \\
Padded Tet  &             0.715     &   0.812   &     0.866 \\
\textbf{G1}    &           0.769       & 0.811     &   0.866 \\
\textbf{G2}    &           0.757       & 0.820     &   0.866 \\
Ellipsoid   &            0.767      &  0.825    &    0.866 
    \end{tabular}
    \caption{For various hex meshes, we indicate the maximized minimum scaled Jacobian before and after singular decomposition. The first column indicates the mesh, the second column indicates before singular decomposition, and the third column indicates after. The fourth column indicates a theoretic upper bound on the minimum scaled Jacobian for the decomposed mesh. It essentially indicates the presence of a valence 3 or 5 singular curve. The maximized minimum scaled Jacobian is invariably higher post singular decomposition.}
    \label{tab:tab1}
\end{table}

\section{Conclusions and Future Work}
This paper presents 
singular nodes as the result of gluing singular curves together at a point and shows
that the reverse can be done via sheet inflation to untangle singular nodes into simple singular curves.
This removes the 3D complexity of singular nodes leaving meshes with lower distortion. We demonstrate this procedure on a variety of meshes showing in all cases that no singular nodes are left behind.

The main limitation of our work is that the local sheets we prescribe for decomposing a singular node are not guaranteed to propagate globally while avoiding self-intersection. This can result in the inability to decompose a singular graph by removing all of its singular nodes. We expect that a valid sheet inflation can always be found and leave its efficient computation to future work. 

While our method decreases the number of singular nodes in a mesh, its \emph{base complex}\cite{gao2015hexahedral} may increase in size. This tradeoff should be considered by the user as they may have to choose between a larger scaled Jacobian or maintaining a small number of base complex cells.

Our results can be extended to design new ways of modifying the singular graph of a mesh. Instead of only decomposing nodes into curves, one can \emph{rewire} singular curves by merging them at a node with sheet collapse, and decomposing them in a different way from how they were combined. 
For example, consider the (4,0,0) node in \autoref{fig:333555}. It's sphere triangulation is a tetrahedron which contains three distinct cycles of length four. Therefore, it is possible to bring two valence 3 singular curves together to form a (4,0,0) node and split them apart again in three distinct ways. Each one results in a different singular graph, none of which require introducing new singularities.

Many works aim to build minimal degree smooth parameterizations of quad meshes with singularities \cite{karvciauskas2016minimal, karvciauskas2019refinable}. These methods do not clearly generalize to the volumetric case where singular nodes may suffer decreased continuity from methods designed for 2D singularities. A promising approach following our work is then to decompose any given singular graph so that no singular nodes exist. We expect that it is easier to adapt quad mesh singular parameterization methods to singular curves that are just 2D singularities extruded into 3D than it is to adapt parameterization methods for singular nodes. Even if one derived a singular node parameterization method for a specific singular node type, there is no guarantee that it extends to any other node type. This problem is made easier by only needing to consider singular curves after decomposition.

\section*{Acknowledgements} 
This project was launched at the Summer Geometry Initiative (SGI) 2021, supported by National Science Foundation grant DMS-2103933, Army Research Office grant W911NF2110095, and generous donations from corporate partners.
Paul Zhang acknowledges the support of the Department of Energy Computer Science Graduate Fellowship and the Mathworks Fellowship. The authors thank Justin Solomon and David Bommes for many valuable discussions.

\bibliography{demoref}

\end{document}


\title{}
\date{}

\begin{figure*}
    \centering
    \includegraphics[width=1\textwidth]{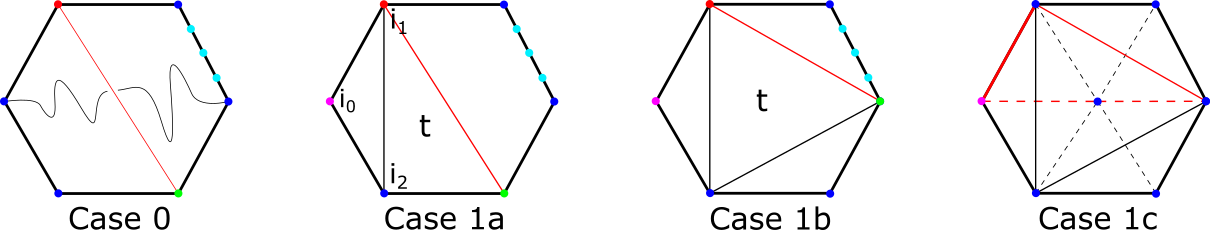}
      \caption{Illustrations of $\T-\U$ in various situations. This figure considers the cases where there are no isolated vertices, or when there is at least one isolated vertex and $t$ has all its vertices on $\C$.}
    \label{fig:cases}
\end{figure*}

\section*{Proof of Prop. 1}
\label{sec:proof_contd}
We show here how to construct a splitting satisfying Prop. 1. For the majority of cases we will construct a path $p$ such that the splitting is $[(u a), p, (b u)]$ as described in the main document. In a few cases, we will explicitly construct the splitting cycle without using $u$ or $p$. Recall from the main document that in order for the splitting to result in sphere triangulations with fewer vertices, it must partition $\T$ into regions with at least two interior vertices each.

We define an \emph{isolated vertex} to be a boundary vertex of $\T-\U$ with degree 2. Isolated vertices are impossible to build path $p$ from since they have no path to the interior of $\T-\U$. For the rest of the proof we will repeatedly reference the labeled cases in \autoref{fig:cases}, \autoref{fig:cases2} and \autoref{fig:cases3}. These figures depict $\T-\U$ with the following color coding. The purple vertex is an arbitrary isolated vertex. The red vertex is labeled $a$, and the green vertex is labeled $b$. 
Cyan vertices are optionally present to indicate that $\C$ has more than or equal to 6 vertices.

\paragraph{Case 0}
First, we consider the case where no vertices on $\C$ are isolated. We can choose $a$ and $b$ arbitrarily as long as they are 3 or more edges apart in $\C$ and compute $p$ as the shortest path between them in the interior of $\T-\U$. If such a path is found, then we have our splitting. If the path does not exist, there must be an interior edge $e$ between boundary vertices on $\C$ separating $a$ from $b$. Since no isolated vertices exist, we can simply choose endpoints of $e$ as the new $a$ and $b$ with $p=e$. This is illustrated by the first image of \autoref{fig:cases}.

The rest of the proof considers the case where at least one isolated vertex $i_0$ exists. Let its neighboring vertices be $i_1$, $i_2$. $(i_0,i_1,i_2)$ must be a triangle in $\T-\U$ with an interior edge $(i_1, i_2)$. Since every interior edge is adjacent to two triangles, flipping across $(i_1, i_2)$ brings us to a triangle we will denote $t$. $t$ can be arranged in a few different ways. 

\paragraph{Case 1a}
The triangle $t$ is completed with a boundary vertex adjacent to  $i_2$. Without loss of generality, $i_1$ is exchangeable with $i_2$. This is illustrated in the second image of \autoref{fig:cases}. The single red edge in $t$ is $p$.

\paragraph{Case 1b}
The triangle $t$ can be completed with a boundary vertex not adjacent to either $i_1$ or $i_2$. This is illustrated in the third image of \autoref{fig:cases}. As long as one or more cyan vertices are present, the red edge separates $\C$ such that both sides have at least 2 interior vertices. 

\paragraph{Case 1c}
This case is identical to Case 1b except there are no cyan vertices i.e. $\C$ has exactly 6 vertices. This configuration is special in that we are unable to split $\T$ into two triangulations each with fewer than 7 vertices. Thus the full sphere triangulation is depicted including $u$ and its neighborhood $\U$ with dashed edges in the fourth image of \autoref{fig:cases}.
For this configuration, we specify the full splitting by the red edges. This splitting turns the sphere triangulation into $$(3,0,3,1) = (2,2,2)+_4(1,3,3).$$
Even though the number of vertices in one of the resulting triangulations is still 7, all vertex degrees are less than or equal to 5, thus they are base cases.

\begin{figure*}
    \centering
    \includegraphics[width=1\textwidth]{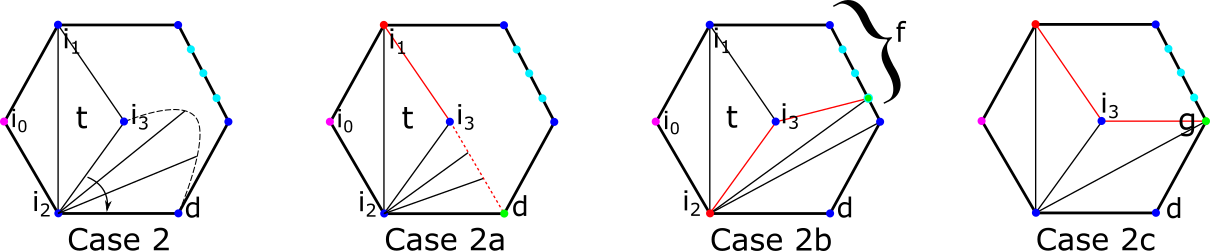}
    \caption{Illustrations of $\T-\U$ in various situations. This figure considers the cases where $t$ has one vertex on the interior of $\T-\U$ and $\C$ has more than 6 vertices.}
    \label{fig:cases2}
\end{figure*}

\paragraph{Case 2}
The last possibility is that triangle $t$ is completed with one interior vertex $i_3$ of $\T-\U$. We illustrate this in the first image of \autoref{fig:cases2} and refer to its labels. Since $i_2$ is adjacent to both $i_3$ and $d$ there must be a path from $i_3$ to $d$ that is entirely adjacent to $i_2$. This is indicated by the dashed black edge. The rest of Case 2 enumerates different possibilities for this path.

\paragraph{Case 2a}
If the path is fully on the interior of $\T-\U$, then $p$ is indicated by the red edges in the second image of \autoref{fig:cases2}.

\paragraph{Case 2b}
If the path connects $i_3$ to any vertex labeled $f$ in the third image of \autoref{fig:cases2}, then $p$ is indicated by the red edges.
\paragraph{Case 2c}
The only remaining possibility is that the path connects $i_3$ to $g$ depicted in the fourth image of \autoref{fig:cases2}. Assuming at least one cyan vertex is present, the red edges form $p$.

\begin{figure*}
    \centering
    \includegraphics[width=1\textwidth]{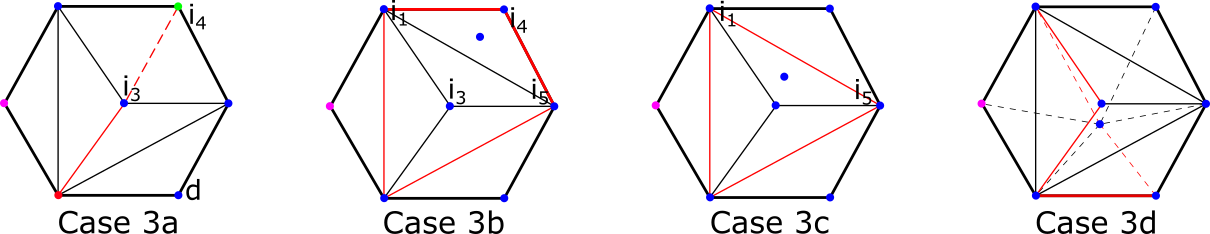}
    \caption{Illustrations of $\T-\U$ in various situations. This figure considers the cases where $t$ has one vertex on the interior of $\T-\U$ and $\C$ has exactly 6 vertices.}
    \label{fig:cases3}
\end{figure*}

\paragraph{Case 3}
If there are no cyan vertices, i.e. $\C$ contains only 6 vertices, then the red edges from Case 2c do not partition $\C$ into two regions of at least two interior vertices. In this setting we consider the last few cases.
\paragraph{Case 3a}
If there is an interior path from $i_3$ to $i_4$ as labeled in the first image of \autoref{fig:cases3}, then the red edges form $p$.
\paragraph{Case 3b, 3c}
The only way for the interior path from case 3a to not exist is if an edge prevents an interior path between $i_3$ and $i_4$. This edge is labeled $(i_1,i_5)$ in the second and third images of \autoref{fig:cases3}. These two images show explicit splittings for $\T$ in the case that there is an interior vertex of $\T-\U$ that is not $i_3$.
\paragraph{Case 3d}
Finally, if there are no interior vertices of $\T-\U$ except for $i_3$, then we know the full sphere triangulation $\T$. This triangulation has signature (4,0,0,4) and is fully depicted in the last image of \autoref{fig:cases3}. The red edges indicate the following splitting:
$$(4,0,0,4)=(3,0,3,1)+_5(1,3,3,1).$$
Both of the resulting sphere triangulations are base cases.

In summary, we have explicitly constructed splittings of $\T$ that result in two sphere triangulations each with fewer vertices than $\T$. The two exceptions are cases 1c and 3d where we construct splittings that result in base cases.